\documentclass[11pt, reqno]{amsart}  
                        
\usepackage{amsmath}
\usepackage{amsfonts}
\usepackage{amssymb}
\usepackage{amsthm}
\usepackage{mathrsfs}
\usepackage{mathtools}
\usepackage{dsfont}
\usepackage{tikz-cd}
\usepackage{physics}
\usepackage{bm}
\usepackage{bbm}
\usepackage[all]{xy}
\usepackage{pdfpages}
\usepackage{enumerate}
\usepackage{float}
\usepackage{todonotes}
\usepackage{graphicx}
\usepackage{color}
\usepackage[colorlinks,linkcolor=blue,urlcolor=blue,citecolor=blue,destlabel=true]{hyperref}
\usepackage{comment}
\usepackage[capitalise]{cleveref}
\numberwithin{equation}{section}

\newtheorem{thm}{Theorem}[section]
\newtheorem*{thm*}{Theorem}
\newtheorem*{prop*}{Proposition}
\newtheorem{theorem}[thm]{Theorem}
\newtheorem{cor}[thm]{Corollary}

\newtheorem{prop}[thm]{Proposition}
\newtheorem{proposition}[thm]{Proposition}
\newtheorem{lem}[thm]{Lemma}
\newtheorem{lemma}[thm]{Lemma}

\theoremstyle{definition}
\newtheorem{defn}[thm]{Definition}
\newtheorem*{defn*}{Definition}
\newtheorem{definition}[thm]{Definition}

\newtheorem{ex}[thm]{Example}

\newtheorem{rem}[thm]{Remark}
\newtheorem{remark}[thm]{Remark}

\newcommand{\sT}{\mathscr{T}}

\newcommand{\N}{\mathbb{N}}

\newcommand{\id}{\mathrm{id}}

\DeclareMathOperator*{\colim}{colim}

\newcommand{\Z}{\mathbb{Z}}

\newcommand{\C}{\mathbb{C}}

\renewcommand{\to}{\longrightarrow}

\newcommand{\fA}{\mathfrak{A}}

\newcommand{\cQ}{\mathcal{Q}}

\newcommand{\cP}{\mathcal{P}}
\newcommand{\cE}{\mathcal{E}}
\newcommand{\cF}{\mathcal{F}}

\newcommand{\cB}{\mathcal{B}}
\newcommand{\cI}{\mathcal{I}}

\newcommand{\Hslab}[1]{H^{\mathrm{slab},a}}

\renewcommand{\tr}{\mathrm{tr}}

\newcommand{\sP}{\mathscr{P}}

\newcommand{\tn}[1]{\textnormal{#1}}

\newcommand{\defeq}{\vcentcolon=}
\renewcommand{\1}{\mathds{1}}

\newcommand{\bbC}{\mathbb{C}}
\newcommand{\bbS}{\mathbb{S}}

\newcommand{\bbP}{\mathbb{P}}

\newcommand{\bbN}{\mathbb{N}}
\newcommand{\bbZ}{\mathbb{Z}}
\newcommand{\bbE}{\mathbb{E}}

\newcommand{\bbR}{\mathbb{R}}
\newcommand{\Unitary}{\mathrm{U}}
\newcommand{\GL}{\mathrm{GL}}

\newcommand{\cM}{\mathcal{M}}

\newcommand{\E}{\mathcal{E}}
\newcommand{\B}{\mathcal{B}}
\newcommand{\F}{\mathcal{F}}
\newcommand{\I}{\mathcal{I}}

\newcommand{\w}{\mathbf{w}}
\newcommand{\bv}{\mathbf{v}}

\DeclareMathOperator{\vecspan}{\mathrm{span}}
\newcommand{\lchi}{{\leq}\,\chi}
\newcommand{\lchiminus}{{\leq}\,(\chi-1)}

\newcommand{\lop}{L}
\newcommand{\rop}{R}
\DeclareMathOperator{\range}{\mathrm{range}}

\newcommand{\cX}{\mathcal{X}}
\newcommand{\cY}{\mathcal{Y}}

\newcommand{\PU}{\mathbb{P}\Unitary}
\renewcommand{\PV}{\mathbb{P}V}

\title[Classifying Space for MPS Phases]{A Classifying Space for Phases of Matrix Product States}

\author[Spiegel]{Daniel D.\ Spiegel\textsuperscript{1,2}}
\author[Qi]{Marvin Qi\textsuperscript{3,4,5}}
\author[Stephen]{David T.\ Stephen\textsuperscript{3,4,6}}
\author[Hermele]{Michael Hermele\textsuperscript{3,4}}
\author[Pflaum]{Markus J.\ Pflaum\textsuperscript{4,7}}
\author[Beaudry]{Agn\`es Beaudry\textsuperscript{7}}

\thanks{\textsuperscript{1}Department of Mathematics, University of California, Davis}
\thanks{\textsuperscript{2}Center for Quantum Mathematics and Physics, University of California, Davis}
\thanks{\textsuperscript{3}Department of Physics, University of Colorado Boulder}
\thanks{\textsuperscript{4}Center for Theory of Quantum Matter, University of Colorado Boulder}
\thanks{\textsuperscript{5}Kadanoff Center for Theoretical Physics, University of Chicago}
\thanks{\textsuperscript{6}Department of Physics and Institute for Quantum Information and Matter, California Institute of Technology}
\thanks{\textsuperscript{7}Department of Mathematics, University of Colorado Boulder}

\begin{document}

\begin{abstract}
We construct a topological space $\B$ consisting of translation invariant injective matrix product states (MPS) of all physical and bond dimensions and show that it has the weak homotopy type $K(\bbZ, 2) \times K(\bbZ, 3)$. The implication is that the phase of a family of such states parametrized by a space $X$ is completely determined by two invariants: a class in $H^2(X;\Z)$ corresponding to the Chern number per unit cell and a class in $H^3(X;\Z)$, the so-called Kapustin--Spodyneiko (KS) number. The space $\B$ is defined as the quotient of a contractible space $\E$ of MPS tensors 
by an equivalence relation describing gauge transformations of the tensors. 
We prove that the projection map $p\colon \E \rightarrow \B$ is a quasifibration, and this allows us to determine the weak homotopy type of $\B$. 
As an example, we review the Chern number pump---a family of MPS parametrized by $S^3$---and prove that it generates $\pi_3(\B)$.
\end{abstract}

\maketitle

\tableofcontents

\section{Introduction}

\subsection{Motivation and main result}
This paper is concerned with the topology of states of quantum spin chains, specifically those states that can be represented as translation invariant injective matrix product states (MPS). We construct a topological space $\B$ which is a \emph{classifying space} for phases of parametrized families of such states. That is, a family of translation invariant injective MPS parametrized by a compact Hausdorff space $X$ can be defined as a map $X \rightarrow \B$, and the phase of this family can be defined as the corresponding element in the set $[X, \B]$ of homotopy classes of such maps. This realizes the proposal we made in \cite[VIII]{ChartingGroundStates} for such a classifying space.

Our space $\B$ consists of translation invariant injective matrix product states (MPS) of \emph{all} physical and bond dimensions. We prove that  $\B$  is weakly homotopy equivalent to a space  $K(\bbZ, 2) \times K(\bbZ, 3)$, where $K(\Z, n)$ denotes an \emph{Eilenberg-Mac Lane space}, i.e., a space whose only non-zero homotopy group $\Z$ is in degree $n$. Under the correspondence
\[H^n(X;\Z) \cong [X, K(\Z, n)]\]
for $n\geq 0$, this rigorously establishes the following claim. The phase of a family of translation invariant injective MPS determines and is uniquely determined by two invariants: 
\begin{enumerate}[(1)]
\item a class in $H^2(X;\Z)$ corresponding to the Chern number, or Berry phase, per unit cell, which persists as a one-dimensional invariant because we have imposed translation invariance; and,
\item a class in $H^3(X;\Z)$, which is the Kapustin--Spodyneiko (KS) number of the family \cite{KapustinSpodyneiko}, and can be interpreted as a flow of Berry curvature \cite{qpump}.
\end{enumerate}

Our work is largely motivated by Kitaev's conjecture, as described in talks from 2013 \cite{kitaevSimonsCenter} and 2019 \cite{kitaev}. In these talks, Kitaev explained (among other things) how in each spatial dimension $k$, there should be a space of gapped bosonic invertible lattice systems $\cQ_k$, so that phases of families of such systems would correspond to homotopy classes of maps $[X,\cQ_k]$. He predicted the following homotopy types:
\[\cQ_0 \simeq K(\Z,2), \quad \quad \cQ_1 \simeq K(\Z,3),  \quad \quad \cQ_2 \simeq \Z \times K(\Z,4). \] 
This answer for $\cQ_0$ can be understood as follows. An invertible bosonic phase over a space $X$ is uniquely determined by the line bundle of ground states. Complex line bundles are classified by their first Chern class so that
\[\mathrm{Line}_\C(X) \cong H^2(X;\Z) =[X,K(\Z,2)].\] 
It follows that $\cQ_0$ should be a $K(\Z,2)$.  The prediction that $\cQ_1$ is a $K(\Z,3)$ means that phases parametrized by a compact Hausdorff space $X$ in spatial dimension one are in one-to-one correspondence with cohomology classes $H^3(X;\Z)$.  The correspondence is obtained by assigning to a family its KS number \cite{KapustinSpodyneiko, ArtymowiczKapustinSopenko}. 

Under the widely accepted hypothesis that gapped invertible lattice systems can be described by a topological quantum field theory at long range (e.g. \cite{KapustinTurzilloMinyoung}), one can use the work of Freed--Hopkins \cite{freed_hopkins} to describe the weak homotopy type of the spaces $\cQ_k$ for all $k\geq 0$.\footnote{The $\cQ_k$ should be the spaces of the loop spectrum $\Sigma^2 I_\bbZ MSO$.} 
However, precise constructions for the spaces $\cQ_k$ with $k \geq 1$ have not been given in terms of lattice systems, and this paper represents a first step towards a precise definition of $\cQ_1$.

It was shown by Hastings in \cite{Hastings_2007} that the ground state of a gapped Hamiltonian in one spatial dimension satisfies an area law for the entanglement entropy, and such area law states were shown in \cite{VerstraeteCirac2006,Schuch2008} to be efficiently approximated by an MPS.  Motivated by this result, one-dimensional bosonic gapped phases with finite internal symmetry group $G$ were studied by considering the $G$-action on MPS, leading to the classification of bosonic symmetry protected topological (SPT) phases in terms of group cohomology $H^2(G; {\rm U}(1)) \cong H^3(G; \Z)$ \cite{PollmannEntanglementSpectrum, ChenGuWen2011_1, FidkowskiKitaev, SchuchClassifyingQuantumPhasesMPS, ChenGuWen2011_2}.  This in turn was a key piece of evidence for Kitaev's prediction that $\cQ_1 \simeq K(\Z,3)$.  An important part of Kitaev's conjecture is the idea that an invertible $k$-dimensional system with internal symmetry $G$ can be modeled by a continuous map $BG \rightarrow \cQ_k$, where $BG$ is the classifying space of $G$.  Taking $\cQ_1 \simeq K(\Z,3)$, one thus recovers the expected classification of $1$-dimensional invertible systems with $G$ symmetry, namely  
\[[BG, \cQ_1] = H^3(BG; \Z) = H^3(G; \Z).\]  

MPS therefore provide a natural starting point and, after imposing translation invariance, a highly tractable class of states. In fact, there has been a string of recent work constructing phase invariants for families of such states, for example, \cite{PhysRevB.110.035114,OhyamaRyuHigherStructures, ohyama2024higher, Shiozaki:2023xky, ChartingGroundStates}. It already follows from \cite{GeometryofMPS} that, given a fixed bond dimension $\chi$, the classifying space $\cB(\chi)$ for translation invariant injective MPS of constant bond dimension $\chi$ has homotopy type
\begin{equation}\label{eq:Bchi_homotopy_type}
\cB(\chi) \simeq K(\Z,2) \times B\PU(\chi),
\end{equation}
i.e., the product of a $K(\Z,2)$ capturing the Chern number per unit cell with the classifying space of the projective unitary group $\PU(\chi) = \Unitary(\chi)/\Unitary(1)\1$, where here $\1$ is the $\chi\times \chi$ identity matrix. However, our work in \cite{ChartingGroundStates} demonstrates that the bond dimension of an MPS can change under a continuous deformation, and hence there are continuous families of injective MPS whose bond dimension is not constant. Such examples are at the core of the motivation for our construction of $\B$: it suggests that a classifying space for injective translation invariant MPS should, in some sense, be built by somehow gluing the spaces $\cB(\chi)$ for all $\chi$.

This paper achieves that vision. We construct a space $\cB$ whose points correspond to translation invariant injective matrix product states, and $\cB$ can be written as a set-theoretic union
\[
	\cB = \bigcup_{\chi \geq 1} \cB(\chi),
\]
where $\cB(\chi)$ has the homotopy type \eqref{eq:Bchi_homotopy_type} when endowed with the subspace topology obtained from $\cB$. The space $\cB$ is constructed as a quotient by gauge transformations of a contractible space $\cE$ whose points are MPS tensors. The quotient map $p \colon \E \rightarrow \B$ is a quasifibration and the weak homotopy type of $\cB$ is that of $K(\bbZ, 2) \times K(\bbZ, 3)$. Consequently, phases of translation invariant injective matrix product states parametrized by a compact Hausdorff space $X$ are classified by 
\[ [X,\cB ] \cong H^2(X;\Z) \times H^3(X;\Z).\]
The class in  $ H^2(X;\Z) $ corresponds to  a well-defined ``Chern number per unit cell'', and the class in $H^3(X;\Z)$ is the KS number of the phase.

The fact that the space $\B$ we construct is a $K(\bbZ, 2) \times K(\bbZ, 3)$ rather than the $K(\bbZ, 3)$ expected for $\cQ_1$ reflects our restriction to translation invariant states: if we were to relax translation invariance, we expect we could push-off the Chern number per unit cell to the boundary or rather to infinity on an infinite lattice. However, at this point, we have not made this idea rigorous in any sense.

\subsection{Sketch of the paper}

Below, we give a sketch of the construction for $\B$, with full definitions beginning in Section \ref{sec:thespace}. However, it is helpful to first describe a construction of $\cQ_0$ which is a template for our construction of $\B$.

One first identifies a zero-dimensional quantum state as an element of $\bbC \bbP^n$ for some $n$. One can obtain this element of $\bbC \bbP^n$ by first describing the state as an element of the unit sphere $\bbS^{2n+1}$ of the Hilbert space $\bbC^{n+1}$ and then quotienting by a $\Unitary(1)$ gauge freedom, corresponding to an unphysical overall phase. We then increase the ``physical dimension'' $n$ by allowing the state to become entangled with ancillary degrees of freedom. Up to a change of basis, the process of adding unentangled ancilla is given by embedding $\bbS^{2n+1} \rightarrow \bbS^{2m+1}$ by appending zeros to a vector. This descends to a map on quotients $\bbC \bbP^n \rightarrow \bbC \bbP^m$. Taking a colimit as $n \rightarrow \infty$ yields a diagram
\begin{equation}\label{eq:zero_d_diagram}
\begin{tikzcd}
\bbS^{3} \rar \dar & \bbS^5 \rar \dar & \cdots\rar & \bbS^\infty \dar \\
\bbC\bbP^1 \rar & \bbC \bbP^2 \rar & \cdots \rar & \bbC\bbP^\infty.
\end{tikzcd}
\end{equation}
The space $\cQ_0$ is identified as the colimit $\bbC \bbP^\infty$. The map $\bbS^\infty \rightarrow \bbC \bbP^\infty$ is a fiber bundle with typical fiber $\Unitary(1)$ and the total space $\bbS^\infty$ is contractible, so we see that $\cQ_0=\bbC \bbP^\infty$ is a $K(\bbZ, 2)$ from the long exact sequence on homotopy groups corresponding to the bundle.

In higher dimension, gapped ground states of bosonic systems in the thermodynamic limit are described as pure states of a uniformly hyperfinite $C^*$-algebra $\fA$, called a quasi-local algebra. As indicated above, we are in particular concerned with the subset of pure states consisting of translation invariant injective MPS, in part because of the simplicity of describing their gauge freedom. To be precise, an injective MPS tensor is defined to be an array of complex numbers $K^i_{\alpha \beta}$, $i \in \qty{1,\ldots, d}$ and $\alpha, \beta \in \qty{1,\ldots, D}$ for some numbers $d, D \in \bbN$ such that the matrices $K^i$ span $M_D(\bbC)$. The number $d$ is called the \emph{physical dimension} and $D$ is called the \emph{bond dimension} of $K$. Section \ref{sec:pure_state_space} briefly reviews the works of Fannes, Nachtergaele, and Werner \cite{FannesNachtergaeleWernerFCS,FannesNachtergaeleWernerPureStates} to explain how such a tensor gives rise to a pure state on a quasi-local algebra $\fA$. In particular, if two tensors $K$ and $L$ yield the same state, then they must be of the same bond dimension and furthermore are related by a gauge transformation of the form 
\[
L^i = \lambda XK^iX^{-1},
\]
where $\lambda \in \bbC^\times$ and $X \in \GL(D)$. If one further requires the \emph{right-normalization condition}:
\[
\sum_i K^iK^{i*} = \sum_j L^j L^{j*} = \1_{D \times D}, 
\]
then in fact $\lambda \in \Unitary(1)$ and $X \in \Unitary(D)$. We find this more restricted gauge group much easier to work with, and will therefore impose this normalization condition on our tensors. 

We wish to allow the physical dimension to increase as in the zero-dimensional case. In 1d we now also have the bond dimension and would like to allow this to increase arbitrarily as well. Indeed, continuous connected families of injective MPS can have varying bond dimension, such as the Chern number pump described in Section \ref{sec:example}. Therefore it is natural to embed an injective MPS tensor $K$ in an $\bbN \times \bbN \times \bbN$ array $K^i_{\alpha \beta}$ by adding zeros to the original array; the tensor can then be varied to change the physical or bond dimension. 

However, if we only consider tensors of the form 
\[
\mqty(K^i& 0 \\ 0 & 0)
\]
then requiring the right-normalization condition will prevent a continuous connected family of injective MPS from ever changing bond dimension. To rectify this, we allow tensors of the form
\[
\mqty(K^i & 0 \\ M^i & 0)
\]
where $M^i_{\alpha \beta}$ are arbitrary complex numbers, except that only finitely many are nonzero. Finally, we allow gauge transformations of these tensors. An overall phase $\lambda \in \Unitary(1)$ can be absorbed into the $K^i$ and $M^i$ matrices, so we ultimately consider arrays of the form
\begin{equation}\label{eq:A_tensor}
A^i = X\mqty(K^i& 0\\M^i&0)X^*,
\end{equation}
where $X \in \Unitary(\infty) = \colim_D \Unitary(D)$ is a unitary that does not depend on the index $i$. 

Tensors of this form were considered by Ohyama and Ryu in \cite{OhyamaRyuHigherStructures}, although there it was required that $A$ satisfy the right-normalization condition, while we only require $K$ to satisfy this. Following \cite{OhyamaRyuHigherStructures}, we call the bond dimension of $K$ the \emph{essential rank} of the MPS tensor $A$. Under a gauge transformation, a tensor can transform as
\[
X\mqty(K^i & 0 \\ M^i & 0)X^* \mapsto \lambda Y\mqty(K^i&0\\N^i&0)Y^*,
\]
where $\lambda \in \Unitary(1)$, $Y \in \Unitary(\infty)$, and the $N^i_{\alpha \beta}$ are arbitrary complex numbers except that only finitely many are nonzero. 

We define $\E(d, D)$ to be the space of such arrays with $A^i_{\alpha \beta} = 0$ if $i > d$, $\alpha > D$, or $\beta > D$ and we define $\B(d, D)$ to be the quotient of $\E(d, D)$ modulo gauge transformations. In Section \ref{sec:pure_state_space}, we relate this to the pure state space $\sP(d)$ of the quasi-local algebra $\fA(d)$ with on-site physical dimension $d$ via the following proposition. 

\begin{prop*}[\ref{prop:EMPS(d,D)_to_P(d)}]
The map $\E(d, D) \rightarrow \sP(d)$ that associates to each tensor $A$ of the form \eqref{eq:A_tensor} the pure state corresponding to the injective MPS $K$ is well-defined and continuous with respect to the weak* topology on $\sP(d)$. Furthermore, this map factors through a continuous injection $\cB(d, D) \rightarrow \sP(d)$ such that the diagram below commutes.
\[
	\begin{tikzcd}
	\E(d, D) \arrow[d,"p"] \arrow[r] & \sP(d)\\
	\B(d, D) \arrow[ur] & 
	\end{tikzcd}
\]
\end{prop*}

However, in Section \ref{sec:not_embedding} we show that the map $\B(d, D) \rightarrow \sP(d)$ is not in general an embedding when $\sP(d)$ is given the weak* topology, so that the topology on $\B(d, D)$ is strictly finer than the weak* topology. To prove this, we exhibit a weak*-continuous path of pure states in the image of the map $\B(4, 2) \rightarrow \sP(4)$ that is not continuous in $\B(4,2)$. The path is essentially an interpolation between a product state and the ground state of the 1d AKLT model \cite{AKLT1987,AKLT1988,FannesNachtergaeleWernerFCS,SchollwockDMRGMPS}. As the path approaches the product state, the state undergoes a quantum phase transition as described in \cite{QPTinMPS} and has a discontinuity in the topology of $\B(4, 2)$. We hope that understanding the topology of $\B(d, D)$ may shed some light on the open question of appropriately topologizing $\sP(d)$ so that maps like this become embeddings.

We define $\E$ and $\B$ as colimits of the spaces $\E(d, D)$ and $\B(d, D)$, so that we get a diagram analogous to \eqref{eq:zero_d_diagram}:
\[
\xymatrix@C=1.8pc{\E(1, D)\ar[r] \ar[d] & \cdots \ar[r] &  \E(d, D)\ar[r] \ar[d] & \E(d+1, D+1) \ar[r] \ar[d] & \cdots \ar[r] &  \E \ar[d]^-p \\
\B(1, D)  \ar[r] & \cdots\ar[r] &  \B(d, D)\ar[r] & \B(d+1, D+1)\ar[r] & \cdots \ar[r] &  \B.
}
\]
In Section \ref{sec:contractibility} we show the following result, analogous to the zero-dimensional case.

\begin{thm*}[\ref{thm:E_contractible}]
The space $\E$ is contractible.
\end{thm*}  

But unlike in the zero-dimensional case, the map $p:\E \rightarrow \B$ is not a fiber bundle. The difficulty is that the fibers over points of $\B$ with different essential rank are not homeomorphic, nor are gauge transformations described by the free action of a compact group. Although $p:\E \rightarrow \B$ is not a fiber bundle, in Section \ref{sec:quasifibration} we prove the following result.

\begin{thm*}[\ref{thm:quasifibration}]
The map $p:\E \rightarrow \B$ is a quasifibration.
\end{thm*} 

This is precisely the property required for a map like $p$ to induce a long exact sequence on homotopy groups, thus allowing us to access the homotopy type of $\B$. Showing that $p$ is a quasifibration is the main technical challenge of the paper.

As a first step, in Section~\ref{sec:fixed_chi} we study the restriction of $p$ to the spaces of tensors $\E(\chi)$ and states $\B(\chi)$ of fixed essential rank $\chi$.

\begin{prop*}[\ref{prop:BUchi}]
The restriction $p:\E(\chi) \rightarrow \B(\chi)$ is a Serre fibration with fibers $\cF(\chi)$ homotopy equivalent to $\Unitary(1) \times B\Unitary(1)$. Moreover, there exists a commutative diagram
\[
	\begin{tikzcd}
	\cF(\chi) \arrow[r, hook] \arrow[d]& \E(\chi) \arrow[rr,"p"]\arrow[d]&& \cB(\chi) \arrow[d]\\
	\Unitary(1) \times B\Unitary(1) \arrow[r] & B\Unitary(\chi) \arrow[r] & B\bbP\Unitary(\chi) \arrow[r] & B\Unitary(1) \times B\bbP \Unitary(\chi)
	\end{tikzcd}
\]
where the vertical arrows are homotopy equivalences, the map $B\Unitary(\chi)  \to  B\PU(\chi) $ is the map on classifying spaces induced by the group quotient and the map $ B\PU(\chi) \to B(\Unitary(1)\times \PU(\chi)) $ is induced by the inclusion of $\PU(\chi)$ in $\Unitary(1)\times \PU(\chi)$.
\end{prop*}
We note that the above is essentially a $d = D = \infty$ variation of a result of Haegeman, Mari\"en, Osborne, and Verstraete, who proved that the space of MPS tensors is a fiber bundle over the space of MPS \cite{GeometryofMPS}.

The key consequence of the fact that $p\colon \E \rightarrow \B$ is a quasifibration is a long exact sequence on homotopy groups. That is, for any point $A\in \E$ of essential rank $\chi$, letting $p(A)$ be the base point in $\B$ and $\cF(\chi)$ be the fiber over $p(A)$, we have 
 \[
 \cdots \rightarrow \pi_{n}(\cF(\chi), A) \rightarrow \pi_n(\E, A) \rightarrow \pi_n(\B, p(A)) \rightarrow \pi_{n-1}(\cF(\chi), A) \rightarrow \cdots
 \]
Since $\E$ is contractible, we can compute the homotopy groups of $\B$ from the homotopy groups of the fiber $\cF(\chi)$, which from Section~\ref{sec:fixed_chi}, we know to be 
\[\cF(\chi) \simeq   \Unitary(1) \times B\Unitary(1) \simeq K(\Z,1)\times K(\Z,2).\]
Note that the right-hand side is independent of the essential rank $\chi$ of $A$. In fact, in Section~\ref{sec:homotopy_type} we compute not only the homotopy groups of $\B$ but its weak homotopy type. Note that knowing the homotopy groups of a space is, in general, not sufficient to determine its weak homotopy type.

\begin{thm*}[\ref{thm:weakequivalence}]
The space $\B$ has the weak homotopy type $K(\bbZ,2) \times K(\bbZ, 3)$. In particular, $\pi_n(\B) = \bbZ$ if $n = 2$ or $n = 3$ and is trivial for all other $n$.
\end{thm*}

In Section~\ref{sec:example}, we describe  two fundamental examples of parametrized phases from this new perspective. The families are given by continuous functions
\[ \psi_2\colon S^2 \to \B \quad \quad \text{and}  \quad \quad \psi_3 \colon S^3 \to \B,\]
which are generators for the homotopy group $\pi_2(\B)$ and $\pi_3(\B)$. The family over $S^2 \cong \C\bbP^1$ corresponds to the canonical non-trivial phase in dimension $0$ repeated at every site of a one-dimensional lattice. Its non-triviality is only due to the translation invariance imposed on the systems. 
The family over $S^3$ is much more interesting: it is  the Chern number pump from \cite{qpump,ChartingGroundStates}, and captures a true  non-trivial one-dimensional phase. We prove that $\psi_3$ generates $\pi_3(\B)$ directly from the long exact sequence on homotopy groups. In \cite{ChartingGroundStates}, we explained how to associate a gerbe to certain  families of injective MPS and checked that for $\psi_3$, the corresponding class in $H^3(S^3;\Z)$ was a generator, and so the two results are consistent.

\begin{rem}
In this paper, we choose to work in the category of \emph{all} topological spaces, as opposed to some convenient category such as compactly generated spaces. 
 While working in all spaces has the disadvantage of increasing the technicality of the presentation, the benefit is a very detailed understanding of the point set topology of the spaces constructed. This gives us quite a bit of control over the situation.
\end{rem}

\section{The Space of MPS}\label{sec:thespace}

We now expand on the definitions of the spaces $\E$ and $\B$ described in the introduction. Let us begin with a reminder of the definition of an injective MPS tensor.

\begin{definition}
We define an \textit{MPS tensor} to be an array of complex numbers $A^i_{\alpha \beta} \in \bbC$ with \textit{physical index} $i \in \qty{1,\ldots, d}$ and \textit{bond indices} $\alpha, \beta \in \qty{1,\ldots, D}$ for some $d, D \in \bbN$. The natural number $d$ is called the \textit{physical dimension} of the MPS tensor and $D$ is called the \textit{bond dimension}. An MPS tensor $A^i_{\alpha \beta}$ is said to be \textit{injective} if the linear map
\[
M_D(\bbC) \rightarrow \bbC^d, \quad B \mapsto \sum_i \tr\qty(A^iB)\ket{i}
\]
is injective, where $\ket{i}_{i=1}^d$ is the standard basis of $\bbC^d$. Equivalently, an MPS tensor is injective if and only if $M_D(\bbC) = \vecspan\{A^1, \ldots, A^d\}$.
\end{definition}

We would like to define a space consisting of injective MPS tensors of all physical and bond dimensions in such a way that allows for continuous transitions between different physical and bond dimensions.

\begin{definition}
We define $M_\infty(\bbC)$ to be the set of all infinite arrays of complex numbers $A_{\alpha \beta} \in \bbC$, where $\alpha, \beta \in \bbN$ but with only finitely many pairs of indices $\alpha$ and $\beta$ such that $A_{\alpha \beta} \neq 0$. We then define
\[
\cM = \qty{A \in \prod_{i=1}^\infty M_\infty(\bbC): A^i = 0 \textnormal{ for all but finitely many $i \in \bbN$}},
\]
where $A^i \in M_\infty(\bbC)$ is the $i$th infinite matrix in the sequence of infinite matrices $A \in \prod_{i=1}^\infty M_\infty(\bbC)$. Thus, an element $A \in \cM$ is an array of complex numbers $A^{i}_{\alpha \beta}$ with indices $i,\alpha, \beta \in \bbN$ and with only finitely many nonzero entries.  

Given $d, D \in \bbN$, we define
\[
\cM(d, D) = \qty{A \in \cM: A^i_{\alpha \beta} = 0 \tn{ whenever $i > d$ or $\alpha > D$ or $\beta > D$}}.
\]
Note that $\cM(d, D)$ is a finite-dimensional vector space and therefore has a unique topology making it a Hausdorff topological vector space.
We see that any MPS tensor of physical  dimension $d$ and bond dimension $D$ can be realized as an element of $\cM(d, D)$ by adding zeros to the array. 
\end{definition}

Observe that if $A \in M_\infty(\bbC)$ and $B$ is an $\bbN\times \bbN$ infinite matrix, possibly with infinitely many nonzero entries, then for any $\alpha, \gamma \in \bbN$,
\begin{align*}
(BA)_{\alpha \gamma} = \sum_\beta B_{\alpha \beta}A_{\beta \gamma} \qqtext{and} (AB)_{\alpha \gamma} = \sum_\beta A_{\alpha \beta}B_{\beta \gamma}
\end{align*}
are sums of finitely many nonzero terms, hence they are well-defined complex numbers. If $B$ has only finitely many nonzero entries per column, then
$(BA)_{\alpha \gamma}$ is nonzero only for finitely many $\alpha, \gamma \in \bbN$, hence $BA \in M_\infty(\bbC)$. Similarly, if $B$ has finitely many nonzero entries per row, then $AB \in M_\infty(\bbC)$.

Define $\Unitary(\infty)$ to be the group of $\bbN \times \bbN$ matrices of the form
\[
X = \mqty(X_0 &&& \\ &1&& \\ &&1 & \\ &&&\ddots ),
\]
where $X_0 \in \Unitary(D)$ for some $D \in \bbN$ and where all blank entries are zero. 
We identify $\Unitary(D)$ 
with its image in $\Unitary(\infty)$.
Thus,
\begin{align*}
\Unitary(\infty) = \bigcup_{D \in \bbN} \Unitary(D).
\end{align*}
By our previous remarks, $\Unitary(\infty)$ has a well-defined left group action by conjugation on $M_\infty(\bbC)$. Likewise, an element $X \in \Unitary(\infty)$ has a left action on $A \in \cM$ defined by conjugating each $A^i$ by $X$. 

In general, given $A \in \cM$ and an infinite matrix $B$, we will use $BA$ and $AB$ to denote the element of $\cM$ given by multiplying each $A^i$ on the left or right with $B$, assuming this multiplication gives a well-defined element of $\cM$.

\begin{definition}\label{defn:Idchi}
An MPS tensor $K$ of physical dimension $d$ and bond dimension $\chi$ satisfies the \emph{right-normalization condition} if
\[
\sum_{i=1}^d K^iK^{i*} = \1
\]
where the $\1$ is the $\chi \times \chi$ identity matrix. If an MPS tensor satisfies the right-normalization condition, then we will say it is \emph{right-normalized}. 
We let $\I(d,\chi) \subset M_\chi(\bbC)^d$ be the subspace of right-normalized injective MPS tensors of physical dimension $d$ and bond dimension $\chi$.
\end{definition}
\begin{rem}
  It is worth noting that the set of injective MPS tensors is an open subset of $M_\chi(\bbC)^d$, while the set of MPS tensors satisfying the right normalization condition is a closed subset. Therefore $\cI(d, \chi)$ is the intersection of an open and a closed subset of $M_\chi(\bbC)^d$, or in other words locally closed. Hence, and that is what we mainly need in the following, $\cI(d, \chi)$ is a locally compact space.  
\end{rem}

\begin{definition}\label{def:EGL}
Given $d, D, \chi \in \bbN$ with $\chi \leq D$, we define $\E(d, D, \chi)$ to be the set of those $A \in \cM(d, D)$ of the form:
\begin{equation}\label{eq:EUMPS_def}
A^i = X\mqty(K^i & 0\\ M^i &0) X^{-1},
\end{equation}
where $X \in \Unitary(D)$, the $M^i$ are any $(D - \chi) \times \chi$ matrices, and $K\in \cI(d,\chi)$.
Note that both sides of \eqref{eq:EUMPS_def} may be understood as infinite matrices with finitely many nonzero entries.  Following \cite{OhyamaRyuHigherStructures}, we call $\chi$ the \emph{essential rank} of $A$. Taking unions over the essential rank, we define
\begin{align*}
\E(d, D) \defeq \bigcup_{\chi \leq D} \E(d, D, \chi) \qqtext{and} \E(d, D, \lchi) \defeq \bigcup_{\chi' \leq \chi} \E(d, D, \chi').
\end{align*}
We topologize $\E(d, D, \chi)$, $\E(d, D, \lchi)$, and $\E(d, D)$ as subspaces of $\cM(d, D)$.
\end{definition}

\begin{definition}
Given $A \in \cM(d,D)$, we define
\[
\lop(A) = \sum_{i} A^{i*}A^i \qqtext{and} \rop(A) = \sum_i A^iA^{i*}.
\]
The letters $\lop$ and $\rop$ tell us which factor the star goes on. We note that $\lop(A), \rop(A) \in M_D(\bbC)$ and $\lop(A)$ and $\rop(A)$ are positive matrices and continuous functions of $A \in \cM(d,D)$. We define $Q(A)$ to be the projection onto the image of $\lop(A)$, i.e.,
\[
Q(A) = \theta(\lop(A)),
\]
where $\theta$ is the Heaviside step function:
\[
\theta(x) = \begin{cases} 0 &\tn{if $x \leq 0$} \\ 1 &\tn{if $x > 0$.} \end{cases}
\]
\end{definition}

If $A \in \E(d, D, \chi)$ with $X$, $K$, and $M$ as in \eqref{eq:EUMPS_def}, then
\[
\lop(A) = X\mqty(L(K) + L(M) & 0 \\ 0 & 0)X^{-1}.
\]
One can check that injectivity of $K$ implies that $L(K)$ is invertible, hence $L(K) + L(M)$ is also invertible. Therefore in this case
\begin{equation}\label{eq:QA}
Q(A) = X\mqty(\1 & 0\\0 & 0)X^{-1}.
\end{equation}

From the above constructions, we see that $\E(d, D, \chi) \cap \E(d, D, \chi') = \varnothing$ if $\chi \neq \chi'$. Indeed, the essential rank of $A \in \E(d, D)$ is uniquely determined by $A$; it is the rank of $\lop(A)$.

\begin{proposition}
Given $d, D, \chi \in \bbN$ with $\chi \leq D$, the projection valued map $Q:\E(d, D, \chi) \rightarrow M_D(\bbC)$ is continuous.
\end{proposition}

\begin{proof}
As mentioned above, for $A \in \E(d, D, \chi)$ the rank of $\lop(A)$ is equal to $\chi$ and is therefore constant on $\E(d, D, \chi)$. Projection onto the image is continuous on the set of self-adjoint matrices of a given size and rank.
\end{proof}

Although $Q$ is not continuous on $\E(d, D)$, the following lemma is sometimes useful when we want to consider non-constant essential rank.

\begin{lemma}\label{lem:projection_pt_continuity}
Let $d, D \in \bbN$ and $A_0 \in \E(d, D)$. The functions
\begin{align*}
\E(d,D) \rightarrow M_D(\bbC), \quad A \mapsto Q(A)Q(A_0)\\
\E(d,D) \rightarrow M_D(\bbC), \quad A \mapsto Q(A_0)Q(A)
\end{align*}
are continuous at $A_0$. 
\end{lemma}

\begin{proof}
The second function is obtained by composing the first function with Hermitian conjugation, which is continuous, so it suffices to prove continuity of the first function at $A_0$. For this, let $\lambda_0$ be the smallest nonzero eigenvalue of $\lop(A_0)$ and note that it suffices to show that 
\[
\norm{Q(A)Q(A_0)x - Q(A_0)x} \leq 2\lambda_0^{-1}\norm{\lop(A) - \lop(A_0)}
\]
for all $x \in (\ker Q(A_0))^\perp = \range \lop(A_0)$ with $\norm{x} \leq 1$. 
Given such an $x$, there exists $w \in (\ker \lop(A_0))^\perp$ such that $\lop(A_0)w = x$. Since $w \in (\ker \lop(A_0))^\perp$, it follows that $\norm{w} \leq \lambda_0^{-1}$ since $\lambda_0^{-1}$ is the norm of the inverse of the restriction of $\lop(A_0)$ to $(\ker \lop(A_0))^\perp$. Then, using $Q(A)L(A)=L(A)$ and that $Q(A)$ has norm one,
\begin{align*}
\norm{Q(A)Q(A_0)x - Q(A_0)x} &= \norm{Q(A)\lop(A_0)w - \lop(A_0)w}\\
&\leq \norm{Q(A)\lop(A)w - \lop(A_0)w} \\
&\qquad \qquad + \norm{Q(A)\lop(A_0)w - Q(A)\lop(A)w}\\
&\leq 2\norm{\lop(A)w - \lop(A_0)w}\\
&\leq 2\lambda_0^{-1}\norm{\lop(A) - \lop(A_0)},
\end{align*}
as desired.
\end{proof}

\begin{definition}
We now allow the physical dimension $d$ and bond dimension $D$ to vary. 
Taking unions over $d$ and $D$, we define
\begin{align*}
\E = \bigcup_{d, D \in \bbN} \E(d, D), \qquad \E(\chi) = \underset{D \geq \chi}{\bigcup_{d, D \in \bbN}} \E(d, D, \chi), \qquad \E({\leq}\,\chi) = \bigcup_{\chi' \leq \chi} \E(\chi'). 
\end{align*}
As with fixed $d$ and $D$, we have $\E(\chi) \cap \E(\chi') = \varnothing$ if $\chi \neq \chi'$. Furthermore,
\[
\E(\chi) \cap \E(d, D) = \E(d, D, \chi) \qqtext{and} \E(\lchi) \cap \E(d, D) = \E(d, D, \lchi).
\]
We now endow $\E$ with the final topology induced by the inclusion maps $\E(d, D) \hookrightarrow \E$. 
\end{definition}

\begin{remark}\label{rem:topology_discussion}
With the topology above, $\E$ is the colimit (in the category of topological spaces and continuous maps) of the system of spaces $\E(d, D)$ and closed embeddings $\E(d, D) \hookrightarrow \E(d', D')$ for $d \leq d'$ and $D \leq D'$. The spaces $\E(n, n)$ for $n \in \bbN$ are cofinal in this system, so $\E$ is a sequential colimit. It follows that each inclusion $\E(d, D) \hookrightarrow \E$ is a closed embedding, i.e., the subspace topology on $\E(d, D)$ obtained from $\cM(d, D)$ coincides with the subspace topology on $\E(d, D)$ obtained from $\E$. Furthermore, the spaces $\E(d, D)$ are metrizable since they are subspaces of the finite-dimensional topological vector space $\cM(d, D)$, so the colimit $\E$ enjoys some nice basic topological properties. For example, $\E$ is paracompact Hausdorff and compactly generated. By ``compactly generated,'' we mean that $\E$ is a weak Hausdorff $k$-space as defined in \cite{LewisCGTOP}.

We can topologize $\E(\chi)$ and $\E(\lchi)$ as subspaces of $\E$ or as colimits of $\E(d, D, \chi)$ and $\E(d, D, \lchi)$, and these two topologies coincide.
To see this, first observe that $\E(d, D, \lchi)$ is closed in $\E(d, D)$. Indeed, it is the intersection of $\E(d,D)$ with the preimage of the closed set $\qty{B \in M_D(\bbC): \rank B \leq \chi}$ under the continuous map $L:\cM(d, D) \rightarrow M_D(\bbC)$. Since $\E(\lchi) \cap \E(d, D) = \E(d, D, \lchi)$, it follows that $\E(\lchi)$ is closed in $\E$. Then $\E(\chi)$ is the intersection of the closed set $\E(\lchi)$ with the open set $\E \setminus \E(\lchiminus)$. As noted in \cite[Appendix A]{LewisCGTOP}, the intersection of an open and closed subset of a $k$-space is again a $k$-space in the subspace topology.  Thus, the subspace topologies that $\E(\chi)$ and $\E(\lchi)$ inherit from $\E$ are compactly generated topologies. It now follows from Proposition \ref{prop:subspace_colimit_equal} that the subspace and colimit topologies on $\E(\chi)$ and $\E(\lchi)$ coincide.
\end{remark}

\begin{definition}\label{def:gauge_transform}
Given $A \in \cE(\chi)$, we define a \emph{gauge transformation} of $A$ to be a transformation of the form
\[
A^i \mapsto \lambda Z(A^i+\tilde A^i)Z^{-1},
\]
where $(\lambda, Z) \in \Unitary(1) \times \Unitary(\infty)$ and $\tilde A \in \cM$ such that 
\[
Q(A)\tilde A^i =  0 \qqtext{and} \tilde A^i Q(A) = \tilde A^i
\]
for all $i$. If $X$, $K$, and $M$ correspond to $A$ as in \eqref{eq:EUMPS_def}, then this is equivalent to requiring $\tilde A^i$ to be of the form
\[
\tilde A^i = X\mqty(0 & 0 \\ N^i & 0)X^{-1}
\]
for arbitrary matrices $N^i$ with $\chi$ columns and any number of rows. From this description one sees that $\lambda Z (A^i + \tilde A^i)Z^{-1} \in \E(\chi)$ and
\[
Q(\lambda Z (A^i + \tilde A^i)Z^{-1}) = ZQ(A)Z^{-1}.
\]
One can check that gauge transformations define an equivalence relation ${\sim}$ on $\E$.
\end{definition}

\begin{rem}
The equivalence classes defined above can be understood as the orbits of a right group action on $\E$. The group is $G = \Unitary(1) \times (\cM \rtimes \Unitary(\infty))$, where $\Unitary(\infty)$ acts on $\cM$ by conjugating each matrix $(Z.B)^i = ZB^iZ^*$, and the action of $G$ on $\E$ is given by
\[
	A.(\lambda, Z, B) = \lambda Z^*(A + (\1 - Q(A))BQ(A))Z.
\]
However, $G$ is not locally compact and the action is not continuous due to discontinuities in $Q(A)$ as $A$ changes essential rank. For these reasons, we do not find this a particularly helpful perspective for our purposes and will avoid it for the remainder of the paper.
\end{rem}

\begin{definition}
We define the quotient space
\[
\B = \E/{\sim}
\]
with quotient map
\[
p : \E \rightarrow \B.
\]
We define $\B(d, D,\chi)$, $\B(d,D,\lchi)$, $\B(d,D)$, $\B(\chi)$, and $\B(\lchi)$ to be the images of $\E(d, D, \chi)$, $\E(d,D,\lchi)$, $\E(d,D)$, $\E(\chi)$, and $\E(\lchi)$ under $p$, respectively. Each of these spaces $\B({\cdots})$ is given the quotient topology induced by the corresponding space $\E({\cdots})$.
\end{definition}

\begin{remark}\label{rem}
Since $\E(d, D, \lchi)$ is closed in $\E(d, D)$ and saturated with respect to the quotient map $p$, it follows that the quotient topology on $\B(d, D, \lchi)$ coincides with the subspace topology it inherits from $\B(d, D)$. Likewise, since $\E(d, D, \chi)$ is the intersection of the closed saturated set $\E(d, D, \lchi)$ and the open saturated set $\E(d, D) \setminus \E(d, D, \lchiminus)$, we know that the quotient topology on $\B(d, D, \chi)$ coincides with the subspace topology it inherits from $\B(d, D)$.  By the same logic, the quotient topologies on $\B(\lchi)$ and $\B(\chi)$ coincide with the subspace topologies they inherit from $\B$.

 Since $\E(\lchi)$, $\E(\chi)$, and $\E$ have the final topology induced by the subspaces $\E(d, D, \lchi)$, $\E(d, D, \chi)$, and $\E(d, D)$, respectively, one can show that the quotient topology on $\B(\lchi)$, $\B(\chi)$, and $\B$ coincides with the final topology induced by the inclusions $\B(d, D, \lchi) \hookrightarrow \B(\lchi)$, $\B(d, D, \chi) \hookrightarrow \B(\chi)$, and $\B(d, d) \hookrightarrow \B$, respectively, across all $d , D \in \bbN$.
\end{remark}

\begin{proposition}
Let $d, D, d', D' \in \bbN$ with $d \leq d'$ and $D \leq D'$. The inclusion $\B(d, D) \rightarrow \B(d', D')$ is a closed embedding.
\end{proposition}
\begin{proof}
First we show that $\B(d, D)$ is a closed subset of $\B(d', D')$. Observe that
\begin{align*}
p^{-1}(\B(d, D)) \cap \E(d', D') &=  \qty{A \in \E(d', D', {\leq}D): Q(A)A^i = 0 \tn{ for $i > D$}}.
\end{align*}
We know $\E(d', D', {\leq}D)$ is closed in $\E(d', D')$, so it suffices to show that 
\begin{equation}\label{eq:K_is_0_after_D}
\qty{A \in \E(d', D'): Q(A)A^i = 0 \tn{ for $i > D$}}
\end{equation}
is closed in $\E(d', D')$. Let $(A_k)_{k \in \bbN}$ be a sequence in \eqref{eq:K_is_0_after_D} converging to some $A_0 \in \E(d', D')$. By Lemma \ref{lem:projection_pt_continuity} we know $Q(A_0)Q(A_k) \rightarrow Q(A_0)$. Then for all $i > D$ we have
\[
Q(A_0)A_0^i = \lim_i Q(A_0)Q(A_k)A_k^i = 0.
\]
This proves that $A_0$ is in \eqref{eq:K_is_0_after_D}. Thus, $\B(d, D)$ is a closed subset of $\B(d', D')$.

It remains to show that the quotient topology on $\B(d, D)$ obtained from $\E(d, D)$ coincides with the subspace topology it obtains from $\B(d', D')$. Let $\sT_q$ be the quotient topology and let $\sT_s$ be the subspace topology. It is easily checked that $\sT_s \subset \sT_q$ using the universal properties of these topologies. Suppose $C \subset \B(d, D)$ is closed relative to $\sT_q$. We want to show that $C$ is closed in $\B(d', D')$. Since $\B(d', D')$ has the quotient topology, this is equivalent to showing that $p^{-1}(C) \cap \E(d', D')$ is closed in $\E(d', D')$. 

Let $(A_k)_{k \in \bbN}$ be a sequence in $p^{-1}(C) \cap \E(d', D')$ converging to some $A_0 \in \E(d', D')$. Note that $A_k, A_0 \in p^{-1}(\B(d, D))$ since $p^{-1}(\B(d,D))$ is closed in $\E(d', D')$. Let $X_k$, $K_k$, and $M_k$ correspond to $A_k$ as in Definition \ref{def:EGL}. By compactness of $\Unitary(D')$, there exists a subsequence of $(X_{k_m})_{m \in \bbN}$ such that $X_{k_m} \rightarrow X$ for some $X \in \Unitary(D')$. Then $X_{k_m}^*A_{k_m}X_{k_m} \rightarrow X^*A_0X$. Let $B_m$ and $B_0$ be the elements of $M_{D'}(\bbC)^{d'}$ whose first $d$ matrices are equal to those of $X_{k_m}^*A_{k_m}X_{k_m}$ and $X^*A_0X$, respectively, and whose last $d' - d$ matrices are equal to zero. Since $A_k, A_0 \in p^{-1}(\B(d, D))$, this truncation yields $B_m, B_0 \in \E(d, D')$ and $p(B_m) = p(A_{k_m})$ and $p(B_0) = p(A_0)$. We also still have $B_m \rightarrow B_0$.  

We need to truncate again to get elements of $\E(d, D)$. Let 
\[
P = \mqty(\1_{D \times D} & 0\\0 & 0) \in M_{D'}(\bbC).
\]
Since $B_{m} = \mqty(K_{k_m} & 0 \\ M_{k_m} & 0)$ and each $K_{k_m}$ has bond dimension at most $D$, it is clear that $PB_mP \in \E(d, D)$. Since $p(PB_mP) = p(B_m) = p(A_{k_m})$, we see that $PB_mP \in p^{-1}(C)$. Furthermore, $PB_mP \rightarrow PB_0P$ in $M_D(\bbC)^d$. 

Suppose we can show that $PB_0P \in \E(d, D)$. Since $C$ is closed with respect to  $\sT_q$, we know $p^{-1}(C) \cap \E(d, D)$ is closed in $\E(d, D)$. Therefore $PB_mP \rightarrow PB_0P$ implies $PB_0P \in p^{-1}(C)$. Thus,
\[
p(A_0) = \lim_m p(A_{k_m}) = \lim_m p(PB_mP) = p(PB_0P) \in C.
\]
This proves that $A_0 \in p^{-1}(C)$, hence $p^{-1}(C) \cap \E(d', D')$ is closed in $\E(d', D')$, as desired.

It remains to show that $PB_0P \in \E(d, D)$. Write
\[
B_0 = Y\mqty(L&0&0\\N_1&0_{(D - \chi) \times (D - \chi)}&0\\N_2&0&0_{(D' - D) \times (D' - D)})Y^*
\]
where $Y \in \Unitary(D')$, $L$ is a right-normalized injective MPS of bond dimension $\chi \leq D$, and $N_1$ and $N_2$ are arbitrary matrices. Write $Y$ in block form as
\[
Y = \mqty(Y_{11} & Y_{12} & Y_{13} \\ Y_{21} & Y_{22} & Y_{23} \\ Y_{31} & Y_{32} & Y_{33})
\]
where $Y_{11}$ is $\chi \times \chi$, $Y_{22}$ is $(D - \chi) \times (D - \chi)$, and $Y_{33}$ is $(D' - D) \times (D' - D)$. Since $B_m \rightarrow B_0$ and the last $D' - D$ columns of each $B_m$ are zero, we know the last $D' - D$ columns of $B_0$ are zero, hence the last $D' - D$ columns of $Y^*B_0$ are also zero. But the top right block of $Y^*B_0$ is $LY_{31}^*$. Since this must be zero and since $L$ is an injective MPS, we see that $Y_{31} = 0$.

The first $\chi$ columns of $Y$ lie in $\bbC^D \subset \bbC^{D'}$ and are orthogonal since $Y$ is unitary. They may be extended to an orthonormal basis of $\bbC^D$. In other words, there exist matrices $Z_1$ and $Z_2$ such that the matrix
\[
\tilde Y \defeq \mqty(Y_{11} & Z_1 & 0 \\ Y_{21} & Z_2 & 0 \\0 & 0 & \1_{(D' - D) \times (D' - D)})
\]
is unitary. Furthermore, unitarity of $Y$ implies 
\begin{align*}
Y_{11}^*Y_{11} + Y_{21}^*Y_{21} &= \1\\
Y_{11}^*Y_{12} + Y_{21}^*Y_{22} &= 0\\
Y_{11}^*Y_{13} + Y_{21}^*Y_{23} &= 0.
\end{align*}
We then observe that
\begin{align*}
PB_0P &= \tilde Y\tilde Y^*\mqty(Y_{11} & Y_{12} & Y_{13}\\Y_{21} & Y_{22} &Y_{23} \\ 0 & 0 & 0)\mqty(L&0&0\\N_1&0&0\\N_2&0&0)\mqty(Y_{11}^*&Y_{21}^*&0\\Y_{12}^* & Y_{22}^* & 0 \\ Y_{13}^* & Y_{23}^* & 0)\\
&= \tilde Y \mqty(\1 & 0 & 0 \\ W_1 & W_2 & W_3 \\ 0 & 0 & 0)\mqty(L&0&0\\N_1&0&0\\N_2&0&0)\tilde Y^* = \tilde Y \mqty(L& 0 & 0 \\ \tilde N & 0 & 0 \\ 0 & 0 & 0)\tilde Y^*,
\end{align*}
where $W_1$, $W_2$, $W_3$, and $\tilde N$ are some matrices. From this it is manifest that $PB_0P \in \E(d, D)$, completing the proof.
\end{proof}

\begin{remark}\label{rem:B_space_topologies}
We now observe that $\B(d, D, \chi) = \B(d', D', \chi) \cap \B(d, D)$ and $\B(d, D, \lchi) = \B(d', D', \lchi) \cap \B(d, D)$. It follows that the inclusions $\B(d, D, \chi) \rightarrow \B(d', D', \chi)$ and $\B(d, D, \lchi) \rightarrow \B(d', D', \lchi)$ are also closed embeddings since $\B(d, D)$ is closed in $\B(d', D')$. 

Since $\B(\chi)$, $\B(\lchi)$, and $\B$ are the sequential colimits of these closed embeddings, we know that all the inclusions $\B(d, D, \chi) \rightarrow \B(\chi)$, $\B(d, D, \lchi) \rightarrow \B(\lchi)$, and $\B(d, D) \rightarrow \B$ are closed embeddings. We have now shown that for each of the spaces $\B(\cdots)$, all reasonable topologies we could put on it coincide.

In Corollaries \ref{cor:B_Hausdorff} and \ref{cor:B_colimits_Hausdorff} of \S\ref{sec:pure_state_space} below we will show that $\B(d, D, \chi)$, $\B(d, D, \lchi)$, and $\B(d, D)$ are all Hausdorff. Since quotients of $k$-spaces are $k$-spaces, it follows that $\B(d, D, \chi)$, $\B(d, D, \lchi)$, and $\B(d, D)$ are all compactly generated. Since sequential colimits of compactly generated spaces along closed embeddings are compactly generated, we know $\B(\chi)$, $\B(\lchi)$, and $\B$ are compactly generated.
\end{remark}

\subsection{Mapping \texorpdfstring{$\B(d,D)$}{B(d,D)} into Pure State Space}\label{sec:pure_state_space}

In this section, we will show how $\B(d,D)$ maps into the pure state spaces of quasi-local algebras. Let us first recall the basic construction of matrix product states in the thermodynamic limit for fixed, finite physical and bond dimension.

Given $d \in \bbN$, we define $\fA(d)$ to be the quasi-local algebra on the lattice $\bbZ$ with on-site Hilbert space $\bbC^d$. In other words, for nonempty finite subsets $\Lambda \subset \bbZ$ we define
\[
\fA_{\Lambda}(d) = \bigotimes_{v \in \Lambda} M_d(\bbC)
\]
and the quasi-local algebra is the directed colimit
\[
\fA(d) = \overline{\bigcup_{\Lambda \subset \bbZ} \fA_\Lambda(d)}
\]
over nonempty finite subsets $\Lambda$.  
Let $\sP(d)$ be the set of pure states of $\fA(d)$. Given an injective right-normalized MPS tensor $K^i_{\alpha \beta}$ of bond dimension $\chi$ and physical dimension $d$, we now review how to obtain a pure state $\omega \in \sP(d)$, following Fannes, Nachtergaele, and Werner \cite{FannesNachtergaeleWernerFCS,FannesNachtergaeleWernerPureStates}. 

First, one constructs the completely positive linear map
\[
\bbE: M_\chi(\bbC) \otimes M_d(\bbC) \rightarrow M_\chi(\bbC), \quad \bbE(B \otimes C) = \sum_{i,j} C_{ij}K^{i*}BK^{j}.
\]
Given $C \in M_d(\bbC)$, we can then define a linear map $\bbE_C:M_\chi(\bbC) \rightarrow M_\chi(\bbC)$ by $\bbE_C(B) = \bbE(B \otimes C)$. We also recursively define the completely positive maps $\bbE^{(n)}: M_\chi(\bbC) \otimes M_d(\bbC)^{\otimes n}  \rightarrow M_\chi(\bbC)$ by setting $\bbE^{(1)} = \bbE$ and setting
\[
\bbE^{(n+1)} = \bbE \circ (\bbE^{(n)} \otimes \id_{M_d(\bbC)}).
\]
One can check by induction that
\[
\bbE^{(n)}(B \otimes C_1 \otimes \cdots \otimes C_n) = (\bbE_{C_n} \circ \cdots \circ \bbE_{C_1})(B).
\]

The map $\bbE_\1(B) = \sum_i K^{i*} BK^{i}$ is of particular interest. First of all, it is also completely positive. 
Moreover, injectivity of $K^i_{\alpha \beta}$ implies that $\bbE_\1$ is \textit{irreducible}, meaning the only projections $P \in M_\chi(\bbC)$ such that $PM_\chi(\bbC) P$ is invariant under $\bbE_\1$ are $P = \1$ and $P = 0$. 
The noncommutative Frobenius-Perron theorem of Evans and H{\o}egh-Krohn \cite{MR0482240} then implies that there exists a unique positive invertible $T \in M_\chi(\bbC)_+$ of trace one such that $T$ is an eigenvector of $\bbE_\1$. Likewise, the adjoint of $\bbE_\1$ with respect to the Hilbert-Schmidt inner product, given by $\bbE_\1^*(B) = \sum_i K^{i}BK^{i*}$, is also irreducible and therefore admits a positive eigenvector, unique up to scalar multiplication; we know this eigenvector is $\1$ with eigenvalue $1$ by the right-normalization condition.\footnote{Note that in \cite{FannesNachtergaeleWernerFCS,FannesNachtergaeleWernerPureStates} the opposite normalization is used. There it is assumed that $\bbE_\1(\1) = \1$.} Furthermore, as shown in \cite{MR0482240}, the eigenvalues of these eigenvectors are the same, so $\bbE_\1(T) = T$. It is also shown in \cite{MR0482240} that the corresponding eigenspaces are one-dimensional.

For $n \in \bbN$, the state $\omega$ is then defined on $\fA_{[-n, n] \cap \bbZ}$ by
\begin{align}
\omega(C_{-n} \otimes \cdots \otimes C_n) &=  \tr\qty[\bbE^{(n)}(T \otimes C_{-n} \otimes \cdots \otimes C_n)] \label{eq:MPS_state_formulae1} \\
&=  \tr\qty[ (\bbE_{C_n} \circ \cdots \circ \bbE_{C_{-n}})(T)] .\label{eq:MPS_state_formulae2}
\end{align}
From \eqref{eq:MPS_state_formulae1} we see that $\omega$ is a well-defined positive linear functional on $\fA_{[-n,n] \cap \bbZ}$. From \eqref{eq:MPS_state_formulae2} and the identities $\bbE_\1(T) = T$ and $\bbE_\1^*(\1) = \1$, we see that the value of the state is unchanged as we tensor on identities to the left and right of $C_{-n} \otimes \cdots \otimes C_n$, so that $\omega$ gives a well-defined positive linear functional on $\fA(d)$. The scaling $\tr(T) = 1$ ensures that $\omega(\1) = 1$. 

We note that $\omega$ is a translationally invariant state. 
Furthermore, we recall that $\omega$ is the unique frustration-free ground state of a translation invariant, nearest neighbor interaction \cite[Thm.~5.7]{FannesNachtergaeleWernerFCS}. From this, it follows that $\omega$ is a pure state.  

We also remark that \eqref{eq:MPS_state_formulae2} can be expanded to rewrite $\omega$ more explicitly in terms of the $K^i$ matrices. Setting $C = C_1 \otimes \cdots \otimes C_n$ for ease of notation, we have
\begin{equation}\label{eq:omega_in_terms_of_K}
\omega(C) = \underset{j_1,\ldots, j_n}{\sum_{i_1,\ldots, i_n}} C_{1,i_1j_1}\cdots C_{n,i_nj_n}\tr\qty(K^{i_n*}\cdots K^{i_1*}TK^{j_1}\cdots K^{j_n})
\end{equation}
We can also write the state restricted to $\fA_{[1,n]\cap \bbZ}$ as a mixture of matrix product states on $(\bbC^{d})^{\otimes n}$ with various boundary conditions. Precisely, if we diagonalize $T = \sum_\alpha \mu_\alpha \ketbra{v_\alpha}$ with orthonormal basis $\ket{v_\alpha}_{\alpha = 1}^\chi$ and we define $B_{\alpha \beta} = \ketbra{v_\beta}{v_\alpha}$, then one can show from \eqref{eq:omega_in_terms_of_K} that $\omega$ restricted to $\fA_{[1,n] \cap \bbZ}$ is represented by the density matrix
\begin{align*}
\varrho_{[1,n] \cap \bbZ} \defeq \sum_{\alpha, \beta} \mu_\alpha \ketbra{\Psi(B_{\alpha \beta})}
\end{align*}
where
\[
\ket{\Psi(B_{\alpha \beta})} \defeq \sum_{j_1,\ldots, j_n} \tr(K^{j_1}\cdots K^{j_n}B_{\alpha \beta})\ket{j_1\cdots j_n}
\]
and $\ket{j_1 \cdots j_n}$ is a standard basis vector for $(\bbC^d)^{\otimes n}$.

Under a gauge transformation of the MPS tensor $K^i \mapsto \lambda XK^iX^*$ with $(\lambda,X) \in \Unitary(1) \times \Unitary(\chi)$, we have a transformations of $T$ as:
\[
T \mapsto XTX^*
\]
Plugging this transformation into \eqref{eq:omega_in_terms_of_K}, we see that $K^i$ and $\lambda XK^i X^*$ define the same state $\omega$. 

Conversely, it follows from \cite{FannesNachtergaeleWernerPureStates} that if a given state $\omega \in \sP(d)$ is induced by two right-normalized injective MPS tensors $K$ and $L$ of physical dimension $d$ and bond dimensions $\chi_K$ and $\chi_L$, then $\chi_K = \chi_L$ and $K$ and $L$ are related by a gauge transformation as above. This concludes our review of \cite{FannesNachtergaeleWernerFCS,FannesNachtergaeleWernerPureStates}.

Given an injective MPS tensor $K$ of physical dimension $d$, we have seen how to define a pure state on $\sP(d)$. We will be considering this construction as a function of $K$, and it is useful to make this functional dependence explicit in our notation. Thus, we shall henceforth denote the induced state by $\omega_K$ rather than simply $\omega$. 
The MPS tensors in $\E(d, D)$ are of the more general form \eqref{eq:EUMPS_def} but nonetheless induce pure states on $\fA(d)$ as the following proposition shows.

\begin{prop}\label{prop:EMPS(d,D)_to_P(d)}
Given $d, D \in \bbN$, the map 
\begin{equation}\label{eq:EMPS(d,D)_to_P(d)}
\E(d, D) \rightarrow \sP(d), \quad A = X\mqty(K&0\\M&0)X^{*} \mapsto \omega_A \defeq \omega_K
\end{equation}
is well-defined and continuous with respect to the weak* topology on $\sP(d)$. There exists a unique function $\B(d, D) \rightarrow \sP(d)$ such that the diagram below commutes.
\[
\begin{tikzcd}
\E(d,D)\arrow[d,"p"]\arrow[r]& \sP(d)\\
\B(d, D) \arrow[ur] & 
\end{tikzcd}
\]
Moreover, $\B(d, D) \rightarrow \sP(d)$ is a continuous injection.
\end{prop}

\begin{proof}
Let $A \in \E(d, D)$ and let $X$, $K$, and $M$ be as in \eqref{eq:EUMPS_def}. 
Let $T$ be related to $K$ as above. Then by direct substitution
\begin{equation}\label{eq:omega_in_terms_of_A}
\tr(K^{i_n*}\cdots K^{i_1*}TK^{j_1}\cdots K^{j_n}) = \tr\qty(A^{i_n*}\cdots A^{i_1*} \tilde T A^{j_1}\cdots A^{j_n})
\end{equation}
where
\begin{align*}
\tilde T = X\mqty(T&0\\0&0)X^{*}.
\end{align*}
Combining \eqref{eq:omega_in_terms_of_K} and \eqref{eq:omega_in_terms_of_A}, we see that
in order to show that $\omega_K$ is a well-defined function of $A$, it suffices to show that $\tilde T$ is a well-defined function of $A$. 

It is easy to check that $\sum_{i} A^{i*}\tilde TA^i = \tilde T$. Let $B \in M_D(\bbC)$ be any matrix satisfying this equation in place of $\tilde T$. Then
\begin{equation}\label{eq:eigenvec_B}
X^*BX = \sum_i \mqty(K^{i*}&M^{i*}\\0&0)X^*BX\mqty(K^{i}&0\\M^i&0).
\end{equation}
We see that 
\begin{equation}\label{eq:eigenvec_amplification}
X^*BX = \mqty( B' &0\\0&0)
\end{equation}
for some $ B' \in M_\chi(\bbC)$. Plugging this into \eqref{eq:eigenvec_B}, we find $\sum_i K^{i*}B'K^i = B'$. Since $1$ is a simple eigenvalue of the map $B' \mapsto \sum_i K^{i*} B' K^i$, we know $B' = \lambda T$ for some $\lambda \in \bbC$. Thus, $B = \lambda \tilde T$.

If $\tr(B) = 1$, then since we also have $\tr(\tilde T) = 1$, we know $\lambda = 1$. We have shown that $\tilde T$ spans the $1$-eigenspace of the linear map $B \mapsto \sum_i A^{i*} B A^i$ and that it is the unique eigenvector with eigenvalue one and trace one. It is therefore a well-defined function of $A$.

To prove weak*-continuity of \eqref{eq:EMPS(d,D)_to_P(d)}, we again combine \eqref{eq:omega_in_terms_of_K} and \eqref{eq:omega_in_terms_of_A} and observe that it suffices to show continuity of $\tilde T$ as a function of $A$. Since the linear map $B \mapsto \sum_i A^{i*}BA^i$ is a continuous function of $A$ and the 1-eigenspace is one-dimensional and spanned by $\tilde T$, the continuity of $\tilde T$ follows from \cite[Ch.~2, \S5]{Kato}.

It is clear from the definition that if $A_1, A_2 \in \E(d, D)$ and $A_1 \sim A_2$, then $\omega_{A_1} = \omega_{A_2}$. Therefore the map $\E(d, D) \rightarrow \sP(d)$ descends to a unique continuous map $\B(d, D) \rightarrow \sP(d)$. To show that this map is injective, suppose $A_1, A_2 \in \E(d, D)$ and $\omega_{A_1} = \omega_{A_2}$. Write
\[
A_i = X_i\mqty(K_i&0\\M_i&0)X_i^*
\]
as usual.
Since $\omega_{K_1} = \omega_{K_2}$, we know $K_1$ and $K_2$ have equal bond dimension $\chi$ and $K_1 = \lambda WK_2W^*$ for some $\lambda \in \Unitary(1)$ and $W \in \Unitary(\chi)$. Thus,
\[
A_1 = \lambda X_1\mqty(W &0\\0&\1)\mqty(K_2&0\\\lambda^* M_1W&0)\mqty(W^*&0\\0&\1)X_1^* \sim A_2.
\]
This proves that the map $\B(d, D) \rightarrow \sP(d)$ is injective. 
\end{proof}

\begin{cor}\label{cor:B_Hausdorff}
For all $d, D, \chi \in \bbN$ with $\chi \leq D$, the quotient spaces $\B(d, D, \chi)$, $\B(d, D, \lchi)$, and $\B(d,D)$ are Hausdorff.
\end{cor}

\begin{proof}
By Proposition \ref{prop:EMPS(d,D)_to_P(d)}, each space $\B({\cdots})$ has a continuous injection into $\sP(d)$, which is Hausdorff. The result follows. 
\end{proof}

We now define a colimit of pure state spaces as $d \rightarrow \infty$. We have a directed system of linear isometries $\iota_{d'd}:\bbC^{d} \rightarrow \bbC^{d'}$ for $d \leq d'$ defined by appending zeros to the end of a vector in $\bbC^d$. As described in \cite[\S2.2]{HomotopicalFoundations}, we also have a covariant functor from the category of finite-dimensional Hilbert spaces and linear isometries to the category of topological spaces and continuous maps. This functor maps $\bbC^d$ to $\sP(d)$ and maps $\iota_{d',d}$ to the continuous function $\sP(\iota_{d',d}) : \sP(d) \rightarrow \sP(d')$ defined by
\[
\sP(\iota_{d',d})(\omega)(C_{-n} \otimes \cdots \otimes C_n) = \omega(\iota_{d',d}^*C_{-n}\iota_{d',d} \otimes \cdots \otimes \iota_{d',d}^*C_{n}\iota_{d',d})
\]
for all $\omega \in \sP(d)$ and $C_{-n} \otimes \cdots \otimes C_n \in \fA_{[-n,n]\cap \bbZ}(d')$.
Thus, applying this functor yields a directed system of spaces $\sP(d)$ and continuous maps $\sP(d) \rightarrow \sP(d')$. We define 
\[
\sP(\infty) \defeq \colim_{d \rightarrow \infty} \sP(d)
\]
in the category of topological spaces and continuous maps.

We remark that the space $\sP(\infty)$ enjoys some basic topological properties. The maps $\sP(\iota_{d',d})$ are closed embeddings \cite[Prop.~2.12]{HomotopicalFoundations}, hence the inclusions into the colimit $\sP(d) \rightarrow \sP(\infty)$ are also closed embeddings. Since each of the spaces $\sP(d)$ is metrizable \cite[Prop.~4.3.2]{PedersenCAlgAutomorphisms}, hence compactly generated and paracompact Hausdorff, the space $\sP(\infty)$ is also compactly generated and paracompact Hausdorff.

\begin{prop}\label{prop:EMPS_to_P(infty)}
Given $d, d', D, D' \in \bbN$ such that $d \leq d'$ and $D \leq D'$, the diagram
\[
\begin{tikzcd}
\E(d, D) \arrow[d]\arrow[r, phantom, "\subset"] & \E(d', D')\arrow[d]\\
\sP(d) \arrow[r,"\sP(\iota_{d',d})"]& \sP(d')
\end{tikzcd}
\]
commutes, where the vertical arrows are the maps from Proposition \ref{prop:EMPS(d,D)_to_P(d)}. Thus, there exists a unique continuous map $\E \rightarrow \sP(\infty)$ such that
\[
\begin{tikzcd}
\E(d, D) \arrow[d]\arrow[r, phantom, "\subset"] & \E\arrow[d]\\
\sP(d) \arrow[r]& \sP(\infty)
\end{tikzcd}
\]
commutes for all $d, D \in \bbN$.

Furthermore, there exists a unique function $\B \rightarrow \sP(\infty)$ such that the diagram below commutes. 
\[
\begin{tikzcd}
\E \arrow[d,"p"]\arrow[dr]&\\
\B \arrow[r]& \sP(\infty)
\end{tikzcd}
\]
Moreover, the map $\B \rightarrow \sP(\infty)$ is a continuous injection.
\end{prop}

\begin{proof}
Let $A \in \E(d, D)$ with $K$ and $T$ corresponding to $A$ as usual. Let $C = C_{-n} \otimes \cdots \otimes C_n \in \fA_{[-n,n] \cap \bbZ}(d')$. If we follow the upper path $\E(d, D) \subset \E(d', D') \rightarrow \sP(d')$, then the expectation value of $C$ in the resulting state is
\[
\underset{j_{-n},\ldots, j_n}{\sum_{i_{-n},\ldots, i_n}} C_{-n,i_{-n}j_{-n}}\cdots C_{n,i_nj_n}\tr\qty(K^{i_n*}\cdots K^{i_{-n}*}TK^{j_{-n}}\cdots K^{j_n}),
\]
where the indices $i_k$ and $j_k$ run from $1$ to $d'$. However, since $K^{i} = 0$ for $i > d$, we may equivalently sum over each index from $1$ to $d$. If we follow the lower path $\E(d, D) \rightarrow \sP(d) \rightarrow \sP(d')$, then the expectation value of $C$ in the resulting state is the same, but with $C_{k, i_kj_k}$ replaced by $(\iota_{d',d}^*C_k\iota_{d',d})_{i_kj_k}$ and with the indices summed from $1$ to $d$. Since $(\iota_{d',d}^*C_k\iota_{d',d})_{i_kj_k} = C_{k,i_kj_k}$, these expectation values are clearly the same.

The statements about $\cB$ are immediate from Proposition \ref{prop:EMPS(d,D)_to_P(d)}.
\end{proof}

\begin{cor}\label{cor:B_colimits_Hausdorff}
The spaces $\B(\chi)$, $\B(\lchi)$, and $\B$ are all Hausdorff.
\end{cor}

\begin{proof}
By Proposition \ref{prop:EMPS_to_P(infty)}, $\B$ has a continuous injection into $\sP(\infty)$, which is Hausdorff. Therefore $\B$ is Hausdorff. The spaces $\B(\chi)$ and $\B(\lchi)$ are Hausdorff since they are subspaces of $\B$, by Remark \ref{rem:B_space_topologies}.
\end{proof}

\subsection{An Example of a Quantum Phase Transition}\label{sec:not_embedding}

We show now that the map $\B(d, D) \rightarrow \sP(d)$ in Proposition \ref{prop:EMPS(d,D)_to_P(d)} is in general not a topological embedding when $\sP(d)$ is given the weak* topology. Let us call this map $\xi:\B(d, D) \rightarrow \xi(\B(d, D))$. To show that $\xi$ is not a homeomorphism, we will exhibit a weak*-continuous path in $\xi(\cB(d, D))$ such that $\xi^{-1}$ composed with this path is not continuous in $\cB(d, D)$. 

The path we choose is an example of a quantum phase transition as studied in \cite{QPTinMPS} and is inspired by the examples in that paper. We set $d = 4$ and $D = 2$. We first define a path in $M_2(\bbC)^4$. Given $g \in [0,1]$, we define
\begin{alignat*}{2}
K^1(g) &= \mqty(\sqrt{1 - g^2} & 0 \\0 & \sqrt{1 - g^2}) &\qquad K^2(g) &= \sqrt{\frac{2}{3}}\mqty(0 & g \\0 & 0)\\
K^3(g) &= \sqrt{\frac{1}{3}}\mqty(-g&0\\0&g) &\qquad K^4(g) &= -\sqrt{\frac{2}{3}}\mqty(0&0\\g&0).
\end{alignat*}
At $g = 0$, this is an MPS representation of a product state of the first physical basis vector, while at $g = 1$, this is an MPS representation of the AKLT state on the second, third, and fourth physical basis vectors \cite{SchollwockDMRGMPS}.
We see that $K(g) \in \E(4, 2)$ for all $g \in (0, 1)$. Furthermore, the map $\bbE_\1(B) = \sum_i K^{i*}BK^i$ is self-adjoint with spectrum $\qty{1, 1-(4/3)g^2}$, where $1 - (4/3)g^2$ is triply degenerate. The unique positive, invertible, trace one eigenvector of $\bbE_\1$ is given by $T = \frac{1}{2}\1$. Setting $C = C_1 \otimes \cdots \otimes C_n \in \fA(d)$, the equation
\[
\omega_g(C) = \frac{1}{2}\underset{j_1,\ldots, j_4}{\sum_{i_1, \ldots, i_4}} C_{1,i_1j_1} \cdots C_{n,i_nj_n}\tr(K^{i_n*}(g)\cdots K^{i_1*}(g)K^{j_1}(g)\cdots K^{j_n}(g))
\]
therefore gives a weak*-continuous function $(0,1) \rightarrow \xi(\cB(4, 2))$. At $g = 0$, we define 
\[
\omega_0(C) = C_{1,11}\cdots C_{n,11}.
\] 
This is the pure product state represented by the first standard basis vector on each site. It is clear from these equations that $\omega_g \rightarrow \omega_0$ in the weak* topology as $g \rightarrow 0$. Moreover, $\omega_0$ is represented by the element $\tilde K \in \E(4, 2)$ defined by
\[
\tilde K^1 = \mqty(1 & 0 \\0 & 0), \quad \tilde K^2 = \tilde K^3 = \tilde K^4 = \mqty(0 & 0 \\0 & 0).
\]
Thus, $\omega_0 \in \xi(\B(4, 2))$, so $g \mapsto \omega_g$ defines a path $\omega:[0,1) \rightarrow \xi(\B(4,2))$.

To show that  $\xi^{-1} \circ \omega$ is not continuous, we define a continuous map $F:\B \rightarrow \bbR$ such that $F|_{\B(4,2)} \circ \xi^{-1} \circ \omega$ is not continuous. First, define $f:\E \rightarrow \bbR$ by
\[
f(A) = \abs{\sum_{i=1}^\infty \tr(A^i)}.
\]
Then $f$ is manifestly continuous and it is easy to check from Definition \ref{def:gauge_transform} that $f$ is invariant under gauge transformations. Thus, $f$ factors through a continuous map $F:\B \rightarrow \bbR$. 

For $g \in (0,1)$, we have
\[
F(\xi^{-1}(\omega_g)) = F(p(K)) = f(K) = 2\sqrt{1 - g^2} \xrightarrow{g \rightarrow 0} 2.
\]
However, at $g = 0$, we have
\[
F(\xi^{-1}(\omega_0)) = F(p(\tilde K)) = f(\tilde K) = 1.
\]
Thus, $g \mapsto F(\xi^{-1}(\omega_g))$ is not continuous at $g = 0$. This proves that $\xi^{-1}$ is not continuous, so $\B(4,2) \rightarrow \sP(d)$ is not a topological embedding.

\section{Contractibility of \texorpdfstring{$\E$}{E}}\label{sec:contractibility}

Our main objective is to show that the projection $p:\E \rightarrow \B$ is a quasifibration, which will be accomplished in Section \ref{sec:quasifibration}. Once this is done, the long exact sequence for $p$ will provide a tool for computing the homotopy groups of $\B$. This computation will be made rather simple by the fact that $\E$ is contractible, as is shown in this section. This is the one-dimensional analog of the fact for zero-dimensional quantum systems that the limit of the spaces of gauge data $\bbS^\infty = \colim_d \bbS^d$ is contractible. We begin with some propositions and lemmas and prove contractibility in Theorem \ref{thm:E_contractible}.

\begin{proposition}\label{prop:isometry_path}
Let $\phi:\bbN \rightarrow \bbN$ be a strictly increasing function. There exist continuous maps $\Gamma_{ab}:[0,1] \rightarrow \bbR$
defined for all $a,b \in \bbN$ with the following properties:
\begin{enumerate}[{\rm (i)}]
\item\label{ite:Gamma(0)} $\Gamma_{ab}(0) = \delta_{ab}$,
\item\label{ite:Gamma(1)} $\Gamma_{ab}(1) = \delta_{a\phi(b)}$,
\item\label{ite:Gamma_fin_per_col} $\Gamma_{ab}(t) = 0$ for all $t \in [0,1]$ whenever $a > \phi(b)$,
\item\label{ite:Gamma_isometry} $\Gamma(t)^*\Gamma(t) = \1$ for all $t \in [0,1]$, where $\Gamma(t)$ is the $\bbN \times \bbN$ matrix whose entries are $\Gamma_{ab}(t)$.
\end{enumerate}
Here, $\delta_{ab}$ is the Kronecker delta function.
\end{proposition}

\begin{proof}
For ease of notation, given $t \in [0,1]$ let us write
\[
s \defeq \sin\qty(\frac{\pi t}{2}) \qqtext{and} c \defeq \cos\qty(\frac{\pi t}{2})
\]
and let it be understood that $s$ and $c$ are functions of $t$.

Given $b \in \bbN$ such that $\phi(b) \neq b$, there exists a unique $k \in \qty{0} \cup \bbN$ and $l \in \bbN$ such that $b = \phi^k(l)$ and $l \notin \phi(\bbN)$, where $\phi^0 = \id_\bbN$. 
For any $a,b \in \bbN$ such that $\phi(b) \neq b$ and any $t \in [0,1]$, let $b = \phi^k(l)$ with $k$ and $l$ as described, then define
\[
\Gamma_{ab}(t) = \begin{cases} (-s)^k c  &\tn{if } a = l \\ (-s)^{k-j}c^2 &\tn{if } a = \phi^{j}(l) \tn{ for some } 0 < j \leq k \\ s &\tn{if } a = \phi^{k+1}(l) \\ 0 &\tn{otherwise} \end{cases}.
\]
Note that $0^0 = 1$ by definition, as can occur above if $t = 0$. If $\phi(b) = b$, then we define $\Gamma_{ab}(t) = 1$ if $a = b$ and $\Gamma_{ab}(t) = 0$ if $a \neq b$.

We claim that these $\Gamma_{ab}$ have all the desired properties. It is clear that each $\Gamma_{ab}$ is continuous and easily verified that \eqref{ite:Gamma(0)}, \eqref{ite:Gamma(1)}, and \eqref{ite:Gamma_fin_per_col} hold.

All that remains is to prove that $\Gamma(t)$ is an isometry. We compute:
\begin{align*}
[\Gamma(t)^*\Gamma(t)]_{b'b} &= \sum_{a = 1}^{\min(\phi(b), \phi(b'))} \Gamma_{ab'}(t)\Gamma_{ab}(t).
\end{align*}
We want to show that this is a Kronecker delta $\delta_{b'b}$. The above is clearly symmetric in $b$ and $b'$ so suppose $b \leq b'$. 

Suppose $\phi(b) = b$. Then the above reduces to $[\Gamma(t)^*\Gamma(t)]_{b'b} = \Gamma_{bb'}(t)$. If $\phi(b') = b'$, then this is $\delta_{bb'}$ by definition. If $\phi(b') \neq b'$, then we know $b \neq b'$ and we may write $b' = \phi^{k'}(l')$ as before. If $b = \phi^j(l')$ for some $j \in \qty{0,\ldots, k'+1}$, then we know $\phi^j(l') = b = \phi(b) = \phi^{j+1}(l')$. Injectivity of $\phi$ implies $l' = \phi(l')$, which contradicts the definition of $l'$. Therefore $b \neq \phi^j(l')$ for any $j \in \qty{0,\ldots, k'+1}$, so $\Gamma_{bb'}(t) = 0 = \delta_{bb'}$.

Now suppose $\phi(b) > b$ and note that this implies $\phi(b') > b'$. Write $b = \phi^k(l)$ as before. Then:
\begin{align*}
[\Gamma(t)^*\Gamma(t)]_{b'b} &= \sum_{j=0}^{k+1} \Gamma_{\phi^j(l)b'}(t)\Gamma_{\phi^j(l)b}(t) 
\end{align*}
Write $b' = \phi^{k'}(l')$ as usual. Observe that in order for $\Gamma_{\phi^j(l)b'}(t)$ to be nonzero, we must have $\phi^j(l) = \phi^{j'}(l')$ for some $j' \leq k'+1$. Injectivity of $\phi$ and the definition of $l$ and $l'$ then implies that $j = j'$ and $l = l'$. In particular, $b' = \phi^{k'}(l)$ and we have $k' \geq k$ since we have assumed $b \leq b'$. Thus,
\begin{align*}
[\Gamma(t)^*\Gamma(t)]_{b'b} &= \Gamma_{lb'}(t)\Gamma_{lb}(t) + \sum_{j=1}^{k}\Gamma_{\phi^j(l)b'}(t)\Gamma_{\phi^j(l)b}(t) + \Gamma_{\phi^{k+1}(l)b'}(t)\Gamma_{\phi^{k+1}(l)b}(t)\\
&= (-s)^{k'+k}c^2 + \sum_{j=1}^k (-s)^{k'+k-2j}c^4 + \Gamma_{\phi^{k+1}(l)b'}(t)s\\
&= (-1)^{k'+k}s^{k'-k}c^2\qty(s^{2k} + \sum_{j=1}^k s^{2k-2j}c^2) + \Gamma_{\phi^{k+1}(l)b'}(t)s\\
&= (-1)^{k'+k}s^{k'-k}c^2 + \Gamma_{\phi^{k+1}(l)b'}(t)s,
\end{align*}
where we have used the identity $s^{2k} + \sum_{j=1}^k s^{2k-2j}c^2 = 1$, which can be verified by induction.
If $b = b'$, then $k' = k$ and we get
\[
[\Gamma(t)^*\Gamma(t)]_{b'b} = c^2 + s^2 = 1.
\]
If $b \neq b'$, then we get
\[
[\Gamma(t)^*\Gamma(t)]_{b'b} = (-1)^{k'+k}s^{k'-k}c^2 + (-s)^{k'-k-1}c^2s = 0,
\]
since the two terms differ by a factor of $(-1)^{2k+1} = -1$. This concludes the proof.
\end{proof}

\begin{ex}
When $\phi : \N \rightarrow \N$ is the function $\phi(n)= n+1$, then the formula above produces a matrix $\Gamma(t)$ whose transpose $\Gamma(t)^*$ is given by 
\[
\begin{bmatrix}
 \cos (\frac{t\pi}{2}) & - \sin (\frac{t\pi}{2}) \cos (\frac{t\pi}{2}) & \sin ^2(\frac{t\pi}{2}) \cos (\frac{t\pi}{2}) &  -\sin ^3(\frac{t\pi}{2}) \cos (\frac{t\pi}{2}) &\cdots \\
 \sin (\frac{t\pi}{2}) & \cos ^2(\frac{t\pi}{2}) & -\sin (\frac{t\pi}{2}) \cos ^2(\theta)& \sin^2(\frac{t\pi}{2}) \cos ^2(\frac{t\pi}{2}) &\cdots \\
 0 & \sin (\frac{t\pi}{2}) & \cos ^2(\frac{t\pi}{2}) & -\sin (\frac{t\pi}{2}) \cos ^2(\frac{t\pi}{2}) &\cdots \\
 0 & 0 & \sin (\frac{t\pi}{2}) & \cos ^2(\frac{t\pi}{2})  &\cdots \\
 0 & 0 & 0 & \sin (\frac{t\pi}{2}) &\cdots \\
  0 & 0 & 0 & 0 &\cdots \\
\end{bmatrix}.
\]
\end{ex}

\begin{proposition}\label{prop:image_iso_transform}
Let $A \in \E(d, D)$ for some $d, D \in \bbN$.  Suppose $\Gamma$ is an $\bbN \times \bbN$ matrix and there exists a strictly increasing function $\phi:\bbN \rightarrow \bbN$ such that $\Gamma_{ij} = 0$ whenever $i > \phi(j)$. If $\Gamma^*\Gamma = \1$, then $B^i \defeq \sum_{j=1}^d \Gamma_{ij}A^j$ defines an element of $\E(\phi(d),D)$. 
\end{proposition}

\begin{proof}
Let $X$, $K$, and $M$ correspond to $A$ as in Definition \ref{def:EGL}. Observe that
\[
B^i = X\mqty(\sum_{j=1}^d \Gamma_{ij}K^j & 0\\ \sum_{j=1}^d \Gamma_{ij} M^j&0)X^{-1}.
\]
If $i > \phi(d)$, then for every $j \in \qty{1,\ldots, d}$ we have $i > \phi(d) \geq \phi(j)$, hence $\sum_{j=1}^d \Gamma_{ij}K^j = 0$ and $\sum_{j=1}^d \Gamma_{ij}M^j = 0$. Thus  $B^i = 0$ for $i > \phi(d)$. Therefore we see that $B \in \cM(\phi(d), D)$. 

It remains to show that $\sum_{j=1}^d \Gamma_{ij}K^j$ is a right-normalized MPS tensor. Observe that $\vecspan_i \sum_{j=1}^d \Gamma_{ij}K^j \subset \vecspan_i K^i = M_\chi(\bbC)$, where $\chi$ is the bond dimension of $K$. Moreover, for $i \in \qty{1,\ldots, d}$ we have
\[
\sum_{k=1}^{\phi(d)}(\Gamma^*)_{ik}\sum_{j=1}^d \Gamma_{kj}K^j  
= K^i,
\] 
implying $M_\chi(\bbC) = \vecspan_i K^i \subset \vecspan_i \sum_{j=1}^d \Gamma_{ij}K^j$, so that $\sum_{j=1}^d \Gamma_{ij}K^j$ is an injective MPS tensor. Furthermore, the fact that $\Gamma^*\Gamma = \1$ implies
\begin{align*}
\sum_{i=1}^{\phi(d)}\qty(\sum_{j=1}^d \Gamma_{ij}K^j )\qty(\sum_{k=1}^d \Gamma_{ik}K^{k})^* &= \sum_{j,k = 1}^d (\Gamma^*\Gamma)_{kj} K^j K^{k*}  \\
&= \sum_{j=1}^d K^j K^{j*} = \1,
\end{align*}
so $\sum_{j=1}^d \Gamma_{ij} K^j$ is properly normalized.
\end{proof}

For the next proposition, recall our notation that given $A \in \cM$ and an infinite matrix $\Delta$ with finitely many nonzero entries per column, $\Delta A \Delta^*$ refers to the element of $\cM$ obtained by multiplying each matrix $A^i$ by $\Delta$ on the left and $\Delta^*$ on the right. 

\begin{proposition}\label{prop:domain_iso_transform}
Let  $A \in \E(d, D)$ for some $d, D \in \bbN$. Suppose $\Delta$ is an $\bbN \times \bbN$ matrix and there exists a strictly increasing function $\phi:\bbN \rightarrow \bbN$ such that $\Delta_{\alpha \beta} = 0$ whenever $\alpha > \phi(\beta)$. If $\Delta^*\Delta = \1$, then $\Delta A\Delta^* \in \E(d, \phi(D))$. 
\end{proposition}

\begin{proof}
Let $X$, $K$, and $M$ correspond to $A$ as in Definition \ref{def:EGL}. Let $\Delta_D$ be the $\phi(D) \times D$ upper left block of $\Delta$. Then for any $\phi(D) \times (\phi(D) - D)$ matrix $\tilde \Delta$ we can write
\[
\Delta A^i \Delta^* = \mqty(\Delta_D & \tilde \Delta)\mqty(X & 0\\0& \1)\mqty(\mqty(K^i& 0 \\M^i&0 ) &  0 \\0 & 0 )\mqty(X^{-1} & 0\\0&\1) \mqty(\Delta_D^*\\\tilde \Delta^*).
\]
Note that all block matrices above are of total dimension $\phi(D) \times \phi(D)$ . Since $\Delta^*\Delta = \1$ and $\Delta_{\alpha \beta} = 0$ whenever $\alpha > \phi(\beta)$, it follows that $\Delta_D^*\Delta_D^{\phantom{*}} = \1_{D \times D}$. Thus, $\tilde \Delta$ can be chosen so as to make $\mqty(\Delta_D & \tilde \Delta)$ unitary, from which it is apparent that $\Delta A\Delta^* \in \E(d, \phi(D))$.
\end{proof}

\begin{lemma}\label{lem:EH_contract_first_homotopy}
Let $\phi_1, \phi_2:\bbN \rightarrow \bbN$ be strictly increasing functions and let $\Gamma(t)$ and $\Delta(t)$ be the $\bbN \times \bbN$ matrices associated to $\phi_1$ and $\phi_2$ as in Proposition \ref{prop:isometry_path}.  The map
\begin{equation}\label{eq:zeroth_homotopy}
E:\E \rightarrow \E, \quad E(A)^i = \sum_j \Gamma_{ij}(1) \Delta(1) A^j \Delta(1)^*
\end{equation}
is well-defined, continuous, and homotopic to the identity map.
\end{lemma}

\begin{proof}
The map
\begin{equation}\label{eq:EH_contract_first_homotopy}
\tilde E: \E \times [0,1] \rightarrow \E, \quad \tilde E(A,t)^i = \sum_j \Gamma_{ij}(t)\Delta(t)A^j\Delta(t)^*
\end{equation}
is well-defined by Proposition \ref{prop:image_iso_transform} and Proposition \ref{prop:domain_iso_transform}. At $t = 0$ this map restricts to the identity on $\E$ and at $t = 1$ it restricts to $B$.

To prove continuity, first note that the product topology on $\E \times [0,1]$ coincides with the final topology induced by the inclusions $\E(d,D) \times [0,1] \rightarrow \E \times [0,1]$. Therefore it suffices to prove continuity of the restriction of \eqref{eq:EH_contract_first_homotopy} to $\E(d,D) \times [0,1]$. The range of the restriction lies in $\E(\phi_1(d), \phi_2(D))$ by Proposition \ref{prop:image_iso_transform} and Proposition \ref{prop:domain_iso_transform}, and 
\[
\E(d,D) \times [0,1] \rightarrow \E(\phi_1(d), \phi_2(D)), \quad (A,t) \mapsto \sum_{j=1}^d \Gamma_{ij}(t)\Delta(t)A^j\Delta(t)^*
\]
is continuous since each matrix entry is a continuous function of $A$ and $t$. Composing with the inclusion $\E(\phi_1(d), \phi_2(D)) \rightarrow \E$ proves that the restriction of \eqref{eq:EH_contract_first_homotopy} to $\E(d, D) \times [0,1]$ is continuous, as desired.
\end{proof}

\begin{theorem}\label{thm:E_contractible}
The space $\E$ is contractible.
\end{theorem}

\begin{proof}
Define $\Gamma(t)$ and $\Delta(t)$ as the paths associated to $\phi_1(n) = 3n+1$ and $\phi_2(n) = n+1$, respectively, as in Proposition \ref{prop:isometry_path}. By Lemma \ref{lem:EH_contract_first_homotopy}, we know the identity on $\E$ is homotopic to the map $E:\E \rightarrow \E$ of \eqref{eq:zeroth_homotopy}. Given $A \in \E$, we shall always let $X$, $K$, and $M$ correspond to $A$ as in \eqref{eq:EUMPS_def}, with $\chi$ denoting the bond dimension of $K$. Then $E$ may be written as
\[
E(A)^{3j+1} = \mqty(0&1\\X&0)\mqty(\mqty(K^j&0\\M^j&0)&0\\0&0)\mqty(0&1\\X&0)^{*}, \quad j \in \bbN
\]
with $E(A)^i = 0$ if $i \neq 3j + 1$ for some $j \in \bbN$. We will deform $E$ to a constant in three homotopies. 

Let $\psi:\bbN \times \bbN \rightarrow \bbN$ be any bijection. 
For our first homotopy $F:\E \times [0,1] \rightarrow \E$, we continue to let $F(A,t)^i = 0$ if $i \equiv -1 \bmod 3$ or $i = 1$ and we let $F(A,t)^{3j+1} = E(A)$ for all $j \in \bbN$. 
If $i = 0 \bmod 3$, then there exists $(j, \gamma) \in \bbN \times \bbN$ such that $i = \psi(j,\gamma)$, and we define
\begin{align*}
F(A,t)^{3\psi(j,\gamma)}_{\alpha \beta} &= t\Tr(\rop(A))^{-1/2}\delta_{\alpha 1} A^j_{\gamma, \beta - 1}
\end{align*}
where we set $A^j_{\gamma, 0} = 0$, as occurs when $\beta = 1$. Note that $\Tr(\rop(A))$ is a continuous function of $A \in \E$ and $\Tr(\rop(A)) \geq \chi$ for all $A \in \E(\chi)$, so that $\Tr(\rop(A))$ is nonzero. In matrix form, this can be written as
\[
F(A,t)^{3\psi(j,\gamma)} = t\Tr(\rop(A))^{-1/2}\mqty(0&1\\X&0)\mqty(0_{D \times D}&0\\v(j,\gamma)&0)\mqty(0&1\\X&0)^*
\]
where $D$ is the size of $X$ and
\[
v(j,\gamma) = \bra{\gamma}A^jX = \bra{\gamma}X\mqty(K^j&0\\M^j&0)
\]
and $\ket{\gamma}$ is the $\gamma$th standard basis vector. We see that only the first $\chi$ entries of $v(j,\gamma)$ can be nonzero. Furthermore, if $A \in \E(d, D)$ then $F(A,t)^i = 0$ if
\[
i > d' \defeq \max\{3d+1, \max\{3\psi(j,\gamma):j \leq d, \gamma \leq D\}\}.
\]
From these remarks we see that $F$ maps $\E(d, D) \times [0,1]$ into $\E(d', D+1)$. All matrix entries of $F$ are manifestly continuous, so $F$ is a well-defined continuous function $F:\E \times [0,1] \rightarrow [0,1]$.

It is clear from the definition that $F(A,0) = E(A)$. At $t = 1$, we have an alternate  form for $F(A,1)$ as a matrix product:
\begin{equation}\label{eq:F(A,1)1}
F(A,1)^{3j+1} = \mqty(1 & 0 \\0 & X)\mqty(0&0\\0&\mqty(K^j&0\\M^j&0))\mqty(1 & 0 \\0 & X)^*
\end{equation}
and
\begin{equation}\label{eq:F(A,1)2}
F(A, 1)^{3\psi(j,\gamma)} = \Tr(R(A))^{-1}\mqty(1 & 0 \\0 & X)\mqty(0&v(j,\gamma)\\0&0_{D \times D})\mqty(1 & 0 \\0 & X)^*.
\end{equation}
From here we perform our second homotopy. We define $G:\E \times [0,1] \rightarrow \E$ by
\begin{align*}
G(A,t)^i_{\alpha \beta} &= \begin{cases} \sqrt{t}\delta_{\alpha 1}\delta_{\beta 1} &\tn{if $i = 1$} \\ \sqrt{t}\delta_{\beta 1}A^j_{\alpha -1 ,\gamma} &\tn{if $i = 3\psi(j,\gamma) - 1$}\\ \sqrt{1 - t}F(A,1)^i &\tn{if $i = 0\bmod 3$}\\ \sqrt{1 - t}F(A,1)^i &\tn{if $i = 1\bmod 3$ and $i > 1$} \end{cases}
\end{align*}
where we define $A^j_{0, \gamma} = 0$, as occurs when $i = 3\psi(j,\gamma) - 1$ and $\alpha = 1$.  Thus, we scale \eqref{eq:F(A,1)1} and \eqref{eq:F(A,1)2} to zero and we grow components $G(A,t)^i$ for $i = 1$ and $i = 3\psi(j,\gamma) - 1$, where $F(A,1)^i$ was zero. It is clear that all matrix entries of $G(A,t)$ are continuous and that $G$ maps $\E(d, D) \times [0,1]$ to $\cM(d', D+1)$. Furthermore, $G(A, 0) = F(A,1)$. It remains to show that $G$ maps into $\E$. 

We begin by writing $G$ in matrix form. We have:
\[
G(A,t)^1 = \sqrt{t}\mqty(1 & 0 \\0 & X)\mqty(1&0\\0&0_{D \times D})\mqty(1 & 0 \\0 & X)^*.
\]
and
\[
G(A,t)^{3\psi(j,\gamma) - 1} = \sqrt{t}\mqty(1 & 0 \\0 & X)\mqty(0&0\\w(j,\gamma)&0_{D \times D})\mqty(1 & 0 \\0 & X)^*.
\]
where
\[
w(j,\gamma) = X^*A^j\ket{\gamma} = \mqty(K^j&0\\M^j&0)X^*\ket{\gamma}.
\]

Given $A \in \E(d, D, \chi)$ and $t \in (0,1)$, to show that $G(A,t) \in \E(d', D+1)$ we must show that the upper left $(\chi + 1) \times (\chi + 1)$ block of the matrices conjugated above form a right-normalized injective MPS tensor. This can be rephrased as follows. 
Let $\overline v(j,\gamma)$ and $\overline w(j,\gamma)$ be the truncations of $v(j,\gamma)$ and $w(j,\gamma)$ to their first $\chi$ entries.
For $t \in (0,1)$, to show that $G(A,t) \in \E(d', D+1)$ we must show that the matrices
\begin{equation}\label{eq:G(A,t)_matrices}
\begin{aligned}
\sqrt{t}\mqty(1 & 0 \\0 &0_{\chi \times \chi} ), &\quad \sqrt{t}\mqty(0&0\\\overline w(j,\gamma) & 0_{\chi \times \chi}), \\
 \sqrt{\frac{1-t}{\Tr(\rop(A))}}\mqty(0&\overline v(j,\gamma)\\0&0_{\chi\times \chi}), &\quad \sqrt{1-t}\mqty(0&0\\0&K^j)
 \end{aligned}
\end{equation}
form a right-normalized injective MPS tensor as $j$ ranges from $1$ to $d$ and $\gamma$ ranges from $1$ to $D$. Since $K$ is an injective MPS tensor, we see from the definition of $v(j,\gamma)$, respectively of $w(j,\gamma)$, that any element of $\bbC^\chi$ can be obtained by an appropriate linear combination of $\overline{v}(j,\gamma)$, respectively of $\overline w(j,\gamma)$. Therefore the matrices above span $M_{\chi + 1}(\bbC)$, hence form an injective MPS tensor. To show that it is right-normalized, we compute:
\begin{align*}
\sum_{j, \gamma} \overline{w}(j,\gamma)\overline{w}(j,\gamma)^* = \sum_{j , \gamma} \mqty(K^j & 0) X^*\ketbra{\gamma}X\mqty(K^{j*}\\0) = \1_{\chi \times \chi}
\end{align*}
and
\begin{align*}
\sum_{j, \gamma} \overline{v}(j,\gamma)\overline{v}(j,\gamma)^* = \sum_{j ,\gamma} \bra{\gamma}X\mqty(K^j\\M^j)\mqty(K^{j*}&M^{j*})X^*\ket{\gamma} = \Tr(R(A)).
\end{align*}
With these in hand, one can check that \eqref{eq:G(A,t)_matrices} is right-normalized. Thus, $G(A,t) \in \E$ for $(A, t) \in \E \times (0,1)$.

At $t = 1$, we have 
\[
G(A,1)^1 = \mqty(1&0\\0&0)
\]
and
\[
G(A,1)^{3\psi(j,\gamma) - 1} = \mqty(0&0\\A^j\ket{\gamma}&0)
\]
where $G(A,1)^i = 0$ for all other $i$. This is in $\E$, so we have shown that $G:\E \times [0,1] \rightarrow \E$ is a well-defined continuous map. 

For our final homotopy, we define $H:\E \times [0,1] \rightarrow \E$ by
\[
H(A,t)^1 = G(A,1)^1
\]
and
\[
H(A,t)^i = (1-t)G(A,1)^i
\]
for all $i > 1$. This is well-defined and continuous, satisfies $H(A, 0) = G(A,1)$, and 
\[
H(A,1)^i = \delta_{1i}\mqty(1&0\\0&0)
\]
which is independent of $A$. This concludes the proof.
\end{proof}

Since $\B$ is the image of $\E$ under the continuous map $p$, we have the following corollary.

\begin{cor}
The space $\B$ is path-connected.
\end{cor}

We will no longer need the functions $E$, $F$, $G$, and $H$ as they are defined above. In later sections, these letters will be free to represent other objects.

\section{The Serre Fibration \texorpdfstring{$\E(\chi) \rightarrow \B(\chi)$}{EchitoBchi}}\label{sec:fixed_chi}

As preparation for showing that $p:\E \rightarrow \B$ is a quasifibration, we consider the restriction of this map to the spaces of fixed essential rank.

Let us fix $d, D, \chi \in \bbN$ with $\chi \leq D$. In this section we will show that the projection $p:\E(d, D, \chi) \rightarrow \B(d, D, \chi)$ is a fiber bundle with typical fiber 
\begin{equation}\label{eq:dD_fiber}
\cF(d,D, \chi) \defeq \Unitary(1) \times \qty(\Unitary(D) \times_{\Unitary(1) \times \Unitary(D - \chi)} M_{(D - \chi)\times \chi}(\bbC)^d).
\end{equation}
As mentioned in the introduction, a similar result was obtained in \cite{GeometryofMPS}. 
Taking a colimit in $d$ and $D$ we will obtain that $p:\E(\chi) \rightarrow \B(\chi)$ is a Serre fibration with fiber $\cF(\chi) = \colim_{d, D} \cF(d, D, \chi)$. Moreover, we will show that the there exists a homotopy commutative diagram
\[
	\begin{tikzcd}
	\cF(\chi) \arrow[r, hook] \arrow[d]& \E(\chi) \arrow[rr,"p"]\arrow[d]&& \cB(\chi) \arrow[d]\\
	\Unitary(1) \times B\Unitary(1) \arrow[r] & B\Unitary(\chi) \arrow[r] & B\bbP\Unitary(\chi) \arrow[r] & B\Unitary(1) \times B\bbP \Unitary(\chi)
	\end{tikzcd}
\]
where the vertical arrows are homotopy equivalences. 

Let us explain the notation in \eqref{eq:dD_fiber}.
Let $G$ be a group, let $\cX$ and $\cY$ be a topological spaces, and let $G$ act to the right on $\cX$ and to the left on $\cY$. Consider the equivalence relation on $\cX \times \cY$ given by $(x, y) \sim (x', y')$ if $x' = xg$ and $y' = g^{-1}y$ for some $g \in G$. Then the balanced product of $\cX$ and $\cY$ by $G$ is the set
\[
\cX \times_G \cY \defeq \cX \times \cY/{\sim},
\]
equipped with the quotient topology obtained from the natural projection $\cX \times \cY \rightarrow \cX \times_G \cY$. The image of $(x,y) \in \cX \times \cY$ in the balanced product will be denoted with square brackets $[x, y] \in \cX \times_G \cY$.

In \eqref{eq:dD_fiber}, the group $\Unitary(1) \times \Unitary(D - \chi)$ acts on the right of $\Unitary(D)$ by
\[
X.(\lambda, Z) = X \mqty(\lambda \1_{\chi \times \chi} & 0 \\0 & Z)
\]
and acts on the left of $M_{(D - \chi) \times \chi}(\bbC)$ by
\[
(\lambda, Z).M = \lambda^* ZM.
\]

Recall from Definition~\ref{defn:Idchi} that $\I(d,\chi)$ is the space of injective MPS tensors of physical dimension $d$ and bond dimension $\chi$ satisfying the right normalization condition.
Let $\bbP\Unitary(\chi) = \Unitary(\chi)/(\Unitary(1)  \1)$. Observe that $\Unitary(1) \times \bbP\Unitary(\chi)$ has a free and continuous left group action on $\I(d,\chi)$ by
\begin{equation}\label{eq:U1_PUchi_left}
(\mu, W).A = \mu WAW^*.
\end{equation}
We define a right group action of $\Unitary(1) \times \bbP \Unitary(\chi)$ on $\F(d, D, \chi)$ by
\begin{equation}\label{eq:U1_PUchi_right}
(\nu, [X, M]).(\mu, W) = \qty(\nu \mu^*, \qty[X\mqty(W & 0 \\0 & \1), \mu^*MW]).
\end{equation}
It is easy to check that this is well-defined and continuous.
The significance of these definitions will now become apparent.

\begin{theorem}\label{thm:E_d_D_chi_identification}
The map 
\begin{equation}\label{eq:E_d_D_chi_identification}
\begin{aligned}
\F(d, D, \chi) \times_{\Unitary(1) \times \bbP\Unitary(\chi)} \cI(d, \chi) &\rightarrow \E(d, D, \chi) \\
[(\nu, [X,M]), K] &\mapsto \nu X\mqty(K&0\\M&0)X^*
\end{aligned}
\end{equation}
is a well-defined homeomorphism.
\end{theorem}

\begin{proof}
It is easy to check that this map is well-defined. To show continuity, it suffices to prove continuity of the composition of \eqref{eq:E_d_D_chi_identification} with the projection $\cF(d, D, \chi) \times \I(d, \chi) \rightarrow \cF(d, D, \chi) \times_{\Unitary(1) \times \bbP \Unitary(\chi)} \I(d, \chi)$. Since $\Unitary(1) \times \I(d, \chi)$ is locally compact Hausdorff, the product topology on $\cF(d, D, \chi) \times \I(d, \chi)$ coincides with the quotient topology induced by the projection 
\[
\Unitary(1) \times \Unitary(D) \times M_{(D - \chi) \times \chi}(\bbC)^d \times \I(d, \chi) \rightarrow \F(d, D, \chi) \times \I(d, \chi).
\]
Thus, continuity of \eqref{eq:E_d_D_chi_identification} follows from continuity of the map
\[
\begin{aligned}
\Unitary(1) \times \Unitary(D) \times M_{(D - \chi) \times \chi}(\bbC)^d \times \I(d, \chi) &\rightarrow \E(d, D, \chi)\\
(\nu, X, M, K) &\mapsto \nu X \mqty(K & 0 \\M & 0)X^*
\end{aligned}
\]
which is clear.

It is clear from the definition of $\E(d, D, \chi)$ that the map \eqref{eq:E_d_D_chi_identification} is surjective. Let us show that it is injective. Suppose 
\begin{equation}\label{eq:equal_MPS_reps}
\nu X\mqty(K&0\\M&0)X^* = \mu Y\mqty(L&0\\N&0)Y^*.
\end{equation}
Write $Y^*X$ in block form
\[
Y^*X = \mqty(W&B\\C&Z).
\]
Then multiplying \eqref{eq:equal_MPS_reps} on the right by $Y^*$ and on the left by $X$, we find
\begin{equation}
\nu \mqty(WK+BM & 0 \\ CK + ZM& 0) = \mu \mqty(LW & LB \\ NW & NB).
\end{equation}
Examining the upper right block and using injectivity of $L$, it follows that $B = 0$. Since $Y^*X$ is unitary, we see that $C = 0$ as well. Thus, $W \in \Unitary(\chi)$ and $Z \in \Unitary(D - \chi)$. Examination of the remaining blocks yields
\[
K = \nu^*\mu W^*LW \qqtext{and} M = \nu^*\mu Z^*NW.
\] 
We see that
\begin{align*}
[(\nu, [X, M]), K] &= \qty[\qty(\nu (\nu^*\mu), \qty[Y\mqty(W & 0\\0&Z), \nu^*\mu Z^*NW]), \nu^*\mu W^*LW]\\
&= \qty[\qty(\nu, \qty[Y\mqty(\1&0\\0&Z), Z^*N]), L]\\
&= [(\nu, [Y, N]), L].
\end{align*}
This proves that the map \eqref{eq:E_d_D_chi_identification} is injective.

It remains to show continuity of the inverse. Let $A_0 \in \E(d, D, \chi)$. We show continuity of the inverse at $A_0$. Let $\cP(D, \chi)$ be the space of rank $\chi$ projections in $M_D(\bbC)$. Since the map $\E(d, D, \chi) \rightarrow \cP(D, \chi)$, $A \mapsto Q(A)$ is continuous and since $\Unitary(D) \rightarrow \cP(D, \chi)$, $U \mapsto UQ(A_0)U^*$ is a fiber bundle, we know there exists a neighborhood $O$ of $A_0$ and a continuous map 
\[
U:O \rightarrow \Unitary(D) \qqtext{such that} U(A)Q(A_0)U(A)^* = Q(A)
\]
for all $A \in O$. Let $A_0 = X\mqty(K&0\\M&0)X^*$ and $A = Y\mqty(L&0\\N&0)Y^* \in O$. Then by \eqref{eq:QA},
\[
Y^*U(A)X\mqty(\1 & 0 \\0 & 0) = \mqty(\1 & 0 \\0 & 0)Y^*U(A)X,
\]
hence
\begin{align*}
Y^*U(A)X = \mqty(W&0\\0&Z)
\end{align*}
for some $W \in \Unitary(\chi)$ and $Z \in \Unitary(D - \chi)$. Thus,
\begin{align*}
X^*U(A)^*AU(A)X  = \mqty(W^*LW&0\\Z^*NW&0).
\end{align*}
This implies that the blocks in the matrix on the right hand side are well-defined and continuous functions of $A$.
Therefore the map
\begin{align*}
O &\rightarrow \Unitary(1) \times \Unitary(D) \times M_{(D - \chi) \times \chi}(\bbC)^d \times \I(d, \chi),\\
A &\mapsto \qty(1, U(A)X, Z^*NW, W^*LW)
\end{align*}
is well-defined and continuous. 
Composing with the projections onto the balanced products yields the restriction of the inverse to $O$. This composition is continuous, so we're done.
\end{proof}

\begin{theorem}\label{thm:B_homeomorphism}
The map
\begin{align*}
\I(d,\chi)/\qty(\Unitary(1) \times \bbP\Unitary(\chi)) \rightarrow \B(d, D, \chi)\\
[K] \mapsto p\mqty(K&0\\0&0)
\end{align*}
is a well-defined homeomorphism and the diagram below commutes.
\[
\xymatrix{
\F(d, D, \chi) \times_{\Unitary(1) \times \bbP\Unitary(\chi)} \I(d, \chi) \ar[r]\ar[d]& \E(d, D, \chi) \ar[d]^-p\\
\I(d, \chi)/(\Unitary(1) \times \bbP\Unitary(\chi)) \ar[r] & \B(d, D, \chi)
}
\]
\end{theorem}

\begin{proof}
It is clear that this map is well-defined, bijective, and that the diagram commutes. Since the domain and codomain both have the quotient topology, it is clear that the map is a homeomorphism by Theorem \ref{thm:E_d_D_chi_identification}. 
\end{proof}

\begin{cor}
The map $p:\E(d, D, \chi) \rightarrow \B(d, D, \chi)$ is a fiber bundle. The base space $\B(d, D, \chi)$ is paracompact Hausdorff and $p$ is therefore a Hurewicz fibration.
\end{cor}

\begin{proof}
Since $\Unitary(1) \times \PU(\chi)$ is a compact Lie group acting continuously and freely on $\I(d, \chi)$, we know $\I(d, \chi) \rightarrow \I(d, \chi)/(\Unitary(1) \times \PU(\chi))$ is a fiber bundle. It follows that $\F(d, D, \chi) \times_{\Unitary(1) \times \PU(\chi)} \I(d, \chi) \rightarrow \I(d, \chi)/(\Unitary(1) \times \PU(\chi))$ is a fiber bundle.  That $p$ is a fiber bundle then follows by Theorem \ref{thm:B_homeomorphism}. 

The base space $\B(d, D, \chi)$ is homeomorphic to $\I(d, \chi)/\qty(\Unitary(1) \times \bbP\Unitary(\chi))$ by Theorem \ref{thm:B_homeomorphism}. The space $\I(d, \chi)$ is metrizable as it is a subspace of $M_\chi(\bbC)^d$, hence it is paracompact Hausdorff. Since $\Unitary(1) \times \bbP\Unitary(\chi)$ is compact, the projection map $\I(d, \chi) \rightarrow \I(d, \chi)/\qty(\Unitary(1) \times \bbP\Unitary(\chi))$ is closed. The result now follows since the image of a paracompact Hausdorff space under a continuous closed map is paracompact Hausdorff \cite{Michael}.
\end{proof}

\begin{cor}\label{cor:Serre_fibration}
The map $p:\E(\chi) \rightarrow \B(\chi)$ has the homotopy lifting property with respect to all compact spaces. In particular, it is a Serre fibration.
\end{cor}

\begin{proof}
Let $X$ be a compact space and assume to be given continuous maps making the diagram below commute.
\[
\begin{tikzcd}
X \times \qty{0} \arrow[r]\arrow[d] & \E(\chi) \arrow[d,"p"]\\
X \times [0,1] \arrow[r] & \B(\chi)
\end{tikzcd}
\]
Since $X$ is compact, the horizontal arrows factor through spaces $\E(d, D, \chi)$ and $\B(d, D, \chi)$. Since $p: \E(d, D, \chi) \rightarrow \B(d, D, \chi)$ is a Hurewicz fibration, we get a map $X \times [0,1] \rightarrow \E(d, D, \chi)$ making the diagram below commute.
\[
\begin{tikzcd}
X \times \qty{0} \arrow[r]\arrow[d] &\E(d, D, \chi) \arrow[r,hook] \arrow[d,"p"]& \E(\chi) \arrow[d,"p"]\\
X \times [0,1] \arrow[r]\arrow[ur,dashed] &\B(d, D, \chi) \arrow[r,hook]& \B(\chi)
\end{tikzcd}
\]
The composition of the dashed arrow with the inclusion $\E(d, D, \chi) \hookrightarrow \E(\chi)$ is then the required lifting of the homotopy $X \times [0,1] \rightarrow \B(\chi)$. 
\end{proof}

Let us consider the colimit $\B(\chi)$ further and identify it with a colimit of $\I(d, \chi)$ spaces. Observe that if $d \leq d'$ then $\I(d, \chi)$ has an obvious closed embedding into $\I(d', \chi)$ by adding zero matrices. We define 
\[
\I(\chi) \defeq \colim_{d} \I(d, \chi).
\]
The action of $\Unitary(1) \times \PU(\chi)$ commutes with the maps $\I(d, \chi) \hookrightarrow \I(d', \chi)$, so we get a group action of $\Unitary(1) \times \PU(\chi)$ on $\I(\chi)$. Compactness of $\Unitary(1) \times \PU(\chi)$ ensures that $\Unitary(1) \times \PU(\chi) \times \I(\chi) = \colim_d \Unitary(1) \times \PU(\chi) \times \I(d, \chi)$, from which it follows that the group action on $\I(\chi)$ is continuous. It is clear that the group action on $\I(\chi)$ is still free.

We note that $\I(\chi)$ is paracompact Hausdorff since it is a colimit of paracompact Hausdorff spaces along closed embeddings. Also, $\I(\chi)/(\Unitary(1) \times \PU(\chi))$ is paracompact Hausdorff since it is the image of a paracompact Hausdorff space under a continuous closed map.

The quotient topology on $\I(\chi)/(\Unitary(1) \times \PU(\chi))$ coincides with the final topology induced by the inclusions $\I(d, \chi)/(\Unitary(1) \times \PU(\chi)) \hookrightarrow \I(\chi)/(\Unitary(1) \times \PU(\chi))$. Moreover, if $d \leq d'$ and $\chi \leq D \leq D'$, then we have a commutative square:
\[
\begin{tikzcd}
\I(d, \chi)/(\Unitary(1) \times \PU(\chi)) \rar \dar & \I(d', \chi)/(\Unitary(1) \times \PU(\chi)) \dar\\
\B(d, D, \chi) \rar & \B(d', D', \chi)
\end{tikzcd}
\]
where the vertical arrows are the homeomorphisms from Theorem \ref{thm:B_homeomorphism}. The following proposition is now immediate.

\begin{prop}\label{prop:Bchihomeotype}
For any $\chi \in \bbN$, the space $\B(\chi)$ is homeomorphic to $\I(\chi)/(\Unitary(1) \times \PU(\chi))$. 
\end{prop}

\begin{thm}
The space $\I(\chi)$ is contractible, hence $\B(\chi)$ is a classifying space $B(\Unitary(1) \times \PU(\chi)) = B\Unitary(1) \times B\PU(\chi)$.
\end{thm}

\begin{proof}
The method of proof is similar to that in Section \ref{sec:contractibility}. Let $\phi:\bbN \rightarrow \bbN$ be the function $\phi(n) = n + \chi^2$ and let $\Gamma(t)$ be the $\bbN \times \bbN$ matrices associated with $\phi$ as in Proposition \ref{prop:isometry_path}. Then we have a map $B: \I(\chi) \times [0,1] \rightarrow \I(\chi)$ defined by
\[
B(K, t)^i = \sum_{j=1}^\infty \Gamma_{ij}(t)K^j.
\]
Then $B$ is well-defined and continuous by the same proof as in Proposition \ref{prop:image_iso_transform}. At $t = 0$, we have $B(K, 0) = K$ and at $t = 1$ we have $B(K,1)^i = 0$ for $i \leq \chi^2$ and $B(K, 1)^i = K^{i-\chi^2}$ for $i > \chi^2$.

Now we define a second homotopy as follows. Let $E_1,\ldots, E_{\chi^2}$ be a basis for $M_\chi(\bbC)$ such that $\sum_{i=1}^{\chi^2} E_iE_i^* = \1$. Define $C:\I(\chi) \times [0,1] \rightarrow \I(\chi)$ by
\[
C(K, t)^i = \begin{cases} \sqrt{2t - t^2} E_i &: i \leq \chi^2\\ (1-t)K^{i-\chi^2} &: i > \chi^2 \end{cases}
\]
Then $C$ is a well-defined continuous map such that $C(K, 0) = B(K, 1)$ and $C(K, 1)$ is independent of $K$. This proves that $\I(\chi)$ is contractible.
\end{proof}

We next identify the fibers $\F(\chi)$ more explicitly.
 Recall the fibers of the bundle $p:\E(d, D, \chi) \rightarrow \B(d, D, \chi)$ were
\[
\cF(d, D, \chi) = \Unitary(1) \times \qty(\Unitary(D) \times_{\Unitary(1)\times \Unitary(D - \chi)} M_{(D - \chi) \times \chi}(\bbC)^d)
\]
If $d' \geq d$ and $D' \geq D$, then the map $\cF(d, D, \chi) \rightarrow \cF(d', D', \chi)$ defined by
\[
(\lambda, [X, M]) \mapsto \qty(\lambda, \qty[\mqty(X & 0 \\0& \1), \mqty(M\\0)])
\]
is well-defined and continuous. These maps form a directed system and we define
\[
\cF(\chi) = \colim_{d, D} \cF(d, D, \chi).
\]
Note that each $\cF(d, D, \chi)$ is path-connected, from which it follows that $\cF(\chi)$ is path-connected.

\begin{proposition}\label{prop:fiber_homeomorphism}
Let $K$ be a right-normalized injective MPS of physical dimension $d$ and bond dimension $\chi$. The map
\[
\cF(d, D, \chi) \rightarrow p^{-1}(p(K)) \cap \E(d, D, \chi), \quad (\lambda, [X, M]) \mapsto \lambda X\mqty(K&0\\M&0)X^*
\]
is a well-defined homeomorphism. There exists a unique map $\cF(\chi) \rightarrow p^{-1}(p(K))$ such that the diagram below commutes.
\[
\begin{tikzcd}
\cF(d, D, \chi) \arrow[r]\arrow[d,hook] &  p^{-1}(p(K)) \cap \E(d, D, \chi)\arrow[d,hook]\\
\cF(\chi) \arrow[r] & p^{-1}(p(K))
\end{tikzcd}
\]
Moreover, $\cF(\chi) \rightarrow p^{-1}(p(K))$ is a homeomorphism.
\end{proposition}

\begin{proof}
The map $\cF(d, D, \chi) \rightarrow p^{-1}(p(K)) \cap \E(d, D, \chi)$ is a homeomorphism because it is the restriction of the homeomorphism in Theorem \ref{thm:E_d_D_chi_identification} to the fiber over $p(K)$ with respect to the diagram in Theorem \ref{thm:B_homeomorphism}. The maps $\cF(d, D, \chi) \rightarrow p^{-1}(p(K))$ commute with the maps $\cF(d, D, \chi) \rightarrow \cF(d', D', \chi)$ and therefore there exists a unique continuous map $\cF(\chi) \rightarrow p^{-1}(p(K))$ making the diagram above commute. It is a homeomorphism since $p^{-1}(p(K))$ is the colimit of the inclusions $p^{-1}(p(K)) \cap \E(d, D, \chi) \hookrightarrow p^{-1}(p(K))$ by Proposition \ref{prop:subspace_colimit_equal}.
\end{proof}

Note that the spaces $\cF(d, D, \chi)$ are paracompact Hausdorff and the maps $\cF(d, D, \chi) \rightarrow \cF(d', D', \chi)$ are closed embeddings as they are the restrictions of the closed embeddings $\E(d, D, \chi) \rightarrow \E(d', D', \chi)$ to the maps $p^{-1}(p(K)) \cap \E(d, D, \chi) \rightarrow p^{-1}(p(K)) \cap \E(d', D', \chi)$.  Thus $\cF(\chi)$ is paracompact Hausdorff and the inclusions $\cF(d, D, \chi) \hookrightarrow \cF(\chi)$ are closed embeddings.

We show that $\cF(\chi)$ is weakly homotopy equivalent to $\Unitary(1) \times \PV_{\chi}(\bbC^\infty)$. Here $V_{\chi}(\bbC^D)$ is the Stiefel manifold of orthonormal $\chi$-frames in $\bbC^D$, $V_{\chi}(\bbC^\infty)$ is the colimit
\[
V_{\chi}(\bbC^\infty) = \colim_{D} V_{\chi}(\bbC^D),
\]
and 
\[
\PV_{\chi}(\bbC^\infty) = V_{\chi}(\bbC^\infty)/\Unitary(1) = \colim_D \PV_{\chi}(\bbC^D),
\]
where $\Unitary(1)$ acts by multiplication on every vector in a $\chi$-frame.

\begin{proposition}\label{prop:Stiefel_weak_homotopy_equiv}
The map
\[
\cF(\chi) \rightarrow \Unitary(1) \times \PV_{\chi}(\bbC^\infty), \quad (\lambda, [X, M]) \mapsto (\lambda, [X_1,\ldots, X_{\chi}])
\]
is a homotopy equivalence, where $X_1,\ldots, X_{\chi}$ are the first $\chi$ columns of $X$. Consequently, $\cF(\chi)$ is homotopy equivalent to $\Unitary(1) \times B\Unitary(1)$.
\end{proposition}

\begin{proof}
For $d, D \in \bbN$, we have a continuous map going the other way defined by
\[
\Unitary(1) \times \PV_{\chi}(\bbC^D) \rightarrow \cF(d, D, \chi), \quad (\lambda, [X_1,\ldots, X_{\chi}]) \mapsto (\lambda, [X, 0]),
\]
where $X \in \Unitary(D)$ is any unitary whose first $\chi$ columns are $X_1,\ldots, X_{\chi}$. This is well-defined and continuous. For $d < d'$ and $D < D'$, the maps above commute with the inclusions $\Unitary(1) \times \bbP V_{\chi}(\bbC^D) \rightarrow \Unitary(1) \times \bbP V_{\chi}(\bbC^{D'})$ and $\cF(d, D, \chi) \rightarrow \cF(d', D', \chi)$ and therefore define a continuous function $\Unitary(1) \times \PV_{\chi}(\bbC^\infty) \rightarrow \cF(\chi)$.  Composing this with the map $\cF(\chi) \rightarrow \Unitary(1) \times \PV_{\chi}(\bbC^\infty)$ in the proposition statement gives the identity on $\Unitary(1) \times \PV_{\chi}(\bbC^D)$ and composing them the other way gives the map
\[
\cF(\chi) \rightarrow \cF(\chi), \quad (\lambda, [X, M]) \mapsto (\lambda, [X, 0]).
\]
This map is homotopic to the identity via the homotopy
\[
\cF(\chi) \times [0,1] \rightarrow \cF(\chi), \quad (\lambda, [X, M]) \mapsto (\lambda, [X, tM]).
\]
Finally, we note that $ V_{\chi}(\bbC^\infty)$ is a contractible free $\Unitary(1)$ space so that $ \PV_{\chi}(\bbC^\infty)$ is a $B\Unitary(1)$.
\end{proof}

We finish by giving a description of the weak homotopy type of $\E(\chi)$ and identifying the map $p$.

\begin{prop}\label{prop:BUchi}
There is a homotopy equivalence
\[ \E(\chi) \xrightarrow{\simeq}  V_\chi(\C^\infty) \times_{\Unitary(\chi)} \cI(\chi) \]
which fits in a commutative diagram
\[ \xymatrix{
\F(\chi) \ar[r] \ar[d]^-\simeq &  \E(\chi)   \ar[r]^-p \ar[d]_-\simeq  &  \B(\chi)    \ar[d]_-\simeq &    \\
\Unitary(1) \times \PV_{\chi}(\bbC^\infty) \ar[r]
&   V_\chi(\C^\infty) \times_{\Unitary(\chi)} \cI(\chi) \ar[r]    &  \cI(\chi)/(\Unitary(1) \times \PU(\chi))   
}\]
where the vertical equivalences for $\B(\chi)$ and  $\F(\chi)$ are those of Propositions~\ref{prop:Bchihomeotype} and \ref{prop:Stiefel_weak_homotopy_equiv} respectivtely.
The space $\E(\chi)$ is thus a classifying space $B\Unitary(\chi)$ and the $p \colon \E(\chi) \to \B(\chi)$ fits in a diagram
\[
	\begin{tikzcd}
	\cF(\chi) \arrow[r, hook] \arrow[d]& \E(\chi) \arrow[rr,"p"]\arrow[d]&& \cB(\chi) \arrow[d]\\
	\Unitary(1) \times B\Unitary(1) \arrow[r] & B\Unitary(\chi) \arrow[r] & B\bbP\Unitary(\chi) \arrow[r] & B\Unitary(1) \times B\bbP \Unitary(\chi)
	\end{tikzcd}
\]
where the vertical arrows are homotopy equivalences, the map $B\Unitary(\chi)  \to  B\PU(\chi) $ is the map on classifying spaces induced by the group quotient and the map $ B\PU(\chi) \to B(\Unitary(1)\times \PU(\chi)) $ is induced by the inclusion of $\PU(\chi)$ in $\Unitary(1)\times \PU(\chi)$.
\end{prop}
\begin{proof}We have shown that
\[\E(\chi) = \colim_{d,D} \cF(d,D,\chi)\times_{\Unitary(1) \times \PU(\chi) } \cI(d,\chi). \]
The maps of Proposition~\ref{prop:Stiefel_weak_homotopy_equiv} induce homotopy equivalences
\[\cF(d,D,\chi) \to \Unitary(1) \times \PV_\chi(\C^D)  \]
and they are equivariant if we give $ \Unitary(1) \times \PV_\chi(\C^D) $ the right action of $\Unitary(1) \times \PU(\chi)$ given by
\begin{equation}\label{eq:U1_PUchi_right2}
(\nu, X).(\mu, W) = \qty(\nu \mu^*, X\mqty(W & 0 \\0 & \1)).
\end{equation}
It follows that
\begin{align}\label{eq:Echiequivs}
\E(\chi) &\simeq  \colim_{d,D}  \ ( \Unitary(1) \times \PV_\chi(\C^D))\times_{\Unitary(1) \times \PU(\chi) } \cI(d,\chi)   \\
&\cong  \colim_{d,D}  \PV_\chi(\C^D)\times_{ \PU(\chi) } \cI(d,\chi)  \nonumber \\
&\cong  \colim_{d,D}  V_\chi(\C^D)\times_{ \Unitary(\chi) } \cI(d,\chi)  \nonumber \\
&\cong  \left(\colim_{d,D}  V_\chi(\C^D) \times \cI(d,\chi)\right)/\Unitary(\chi) .  
\end{align}
The isomorphism in the second line of the display is induced by the bottom horizontal maps in the diagram,
\[ \xymatrix{   
\PV_\chi(\C^D)\times \cI(d,\chi)  \ar[r] \ar[d] &    \Unitary(1) \times \PV_\chi(\C^D)\times \cI(d,\chi) \ar[d] \\
\PV_\chi(\C^D)\times_{ \PU(\chi) } \cI(d,\chi)   \ar[r]^-\cong &   ( \Unitary(1) \times \PV_\chi(\C^D))\times_{\Unitary(1) \times \PU(\chi) } \cI(d,\chi) }
  \]
where the top map sends $(X, K)$ to $(1, X, K)$. 
The isomorphism in the third line of the display is justified by using the fact that $\cI(d,\chi)$ is locally compact Hausdorff to show that the canonical continuous bijection 
\[V_\chi(\C^D)\times_{ \Unitary(\chi) } \cI(d,\chi)  \rightarrow \PV_\chi(\C^D)\times_{ \PU(\chi) } \cI(d,\chi)\] 
obtained from the universal properties of the objects on both sides is a quotient map and thus a homeomorphism. 
Since $V_\chi(\C^D)$ is compact Hausdorff, and $ \cI(d,\chi)$ is locally compact Hausdorff, by \cite[Theorem 4.1]{Hirai_et_al}, we can commute the colimit with the product to get
\begin{equation}\label{eq:Echieq}
\E(\chi) \simeq  \left( V_\chi(\C^\infty) \times \cI(\chi)\right)/\Unitary(\chi)= V_\chi(\C^\infty) \times_{\Unitary(\chi)} \cI(\chi). 
\end{equation}
But $V_\chi(\C^\infty) \times \cI(\chi)$ is a product of contractible spaces so is contractible.
Hence, $\E(\chi)$ is the quotient of a contractible space by a free action of $\Unitary(\chi)$, so is a classifying space $B\Unitary(\chi)$.

We now examine the map $p\colon \E(\chi) \to \B(\chi)$. From Theorem~\ref{thm:B_homeomorphism} and our construction of the equivalence \eqref{eq:Echieq},  we have a commutative diagram:
\[ \xymatrix{
\F(\chi) \ar[r] \ar[d]^-\simeq &  \E(\chi)   \ar[r]^-p \ar[d]_-\simeq  &  \B(\chi)    \ar[d]_-\simeq &    \\
\Unitary(1) \times \PV_{\chi}(\bbC^\infty) \ar[r]
&   V_\chi(\C^\infty) \times_{\Unitary(\chi)} \cI(\chi) \ar[r]  \ar[d]       &  \cI(\chi)/(\Unitary(1) \times \PU(\chi))   
 \\
 &    \cI(\chi)/\PU(\chi)   \ar[ur] & 
}\]
The identification of each spaces and maps in terms of classifying spaces now follows.
\end{proof}

\section{The Quasifibration \texorpdfstring{$\E \rightarrow \B$}{EtoB}}\label{sec:quasifibration}

In this section we will finally prove that the projection $p:\E \rightarrow \B$ is a quasifibration. Recall that, by definition, this means that the induced homomorphism $p_*:\pi_n(\E, p^{-1}(b), x_0) \rightarrow \pi_n(\B, b)$ is an isomorphism for all $b \in \B$, $x_0 \in p^{-1}(b)$, and $n \geq 0$. For example, every Serre fibration with path-connected fibers is a quasifibration. Therefore, by Corollary \ref{cor:Serre_fibration} of the previous section and by the identification of the fibers in Proposition \ref{prop:fiber_homeomorphism} of this section, we know $\E(\chi) \rightarrow \B(\chi)$ is a quasifibration. Likewise, for any $S \subset \B(\chi)$, the restriction $p^{-1}(S) \rightarrow S$ is a Serre fibration with path-connected fibers and is therefore a quasifibration.

In the following we will use some crucial results due to Dold--Thom \cite{DoldThomQuasifibrations} reproduced below.
For the convenience of the reader we also refer to Appendix A in the textbook \cite{AguilarGitlerPrieto}
where these results have been summarized.
Before we can formulate the theorem we need the following definition:
\begin{definition}
Let $p\colon E \rightarrow B$ be a continuous map. We say that a subset $X \subset B$ is \emph{distinguished} if $X\subset p(E)$ and the restriction $p^{-1}(X) \rightarrow X$ is a quasifibration.
\end{definition}

\begin{theorem}
  \label{thm:qf_identify}
A continuous map $p:E \rightarrow B$  between Hausdorff topological spaces 
is a quasifibration if any one of the following conditions is satisfied:
\begin{enumerate}[(1)]
\item\label{ite:qf_open_sets} {\rm (\cite[Korollar, Satz 2.2]{DoldThomQuasifibrations} \& \cite[A.1.3]{AguilarGitlerPrieto})}
       $B$ can be covered by distinguished open sets $U$ and $V$ whose intersection $U\cap V$ is also distinguished.
	\item\label{ite:qf_union} {\rm (\cite[Satz 2.15]{DoldThomQuasifibrations} \& \cite[A.1.17]{AguilarGitlerPrieto})}
          The topological space $B$ is the colimit of an increasing sequence of distinguished subspaces
          \[ B_1 \subset B_2 \subset \ldots \  . \] 
	\item\label{ite:qf_deformation} {\rm (\cite[Hilfssatz 2.10]{DoldThomQuasifibrations} \& \cite[A.1.11]{AguilarGitlerPrieto})} $p$ is surjective and there exists a distinguished subspace $B'$ such that the following holds: \\
          Letting $E' = p^{-1}(B') \rightarrow B'$, there are homotopies $H_t \colon E \rightarrow E$ and  $h_t \colon B\rightarrow B$ such that
	\begin{align*}
	H_0 &=\id & H_t(E')&\subset E'  & H_1(E)&\subset E' \\
	h_0 &=\id &h_t(B')&\subset B'  & h_1(B)&\subset B' 
		\end{align*}
	and
	\begin{align*}
	 p\circ H_1&= h_1\circ p 
	\end{align*}
	 with the property that, for all $b\in B$, the restriction 
	 \[H_{1,b} \colon p^{-1}(b) \rightarrow p^{-1}(h_1(b))\] 
	 of $H_1$ to the fiber over $b$ is a weak homotopy equivalence.
\end{enumerate}
\end{theorem}

We note that it is not required that $H_t$ cover $h_t$ for all $t$: this is only enforced at times $t=0,1$. Furthermore, it is not required that $H_t$ and $h_t$ be the identities on the subspaces $E'$ and $B'$. However, both these stronger requirements will hold in our applications: our $H_t$'s will be (strong) deformation retracts covering  (strong) deformation retracts $h_t$.

We will apply the lemma in an inductive fashion as follows. We filter our spaces $\E$ and $\B$ by essential rank:
\[
\begin{tikzcd}
\E(\leq 1) \rar\dar & \E(\leq 2)\dar \rar & \cdots\rar &\E(\lchi)\rar\dar & \cdots \rar & \E\dar \\
\B(\leq 1) \rar & \B(\leq 2) \rar & \cdots \rar & \B(\lchi)\rar &\cdots \rar & \B
\end{tikzcd}
\]
Then, in this section, for each $\chi \geq 2$, we will construct:
\begin{itemize}
	\item A set $O(\lchi)$ such that 
	\[\E(\lchiminus) \subset O(\lchi) \subset \E(\lchi)\] and $O(\lchi)$ is open in $\E(\lchi)$.
	Furthermore, the set $O(\lchi)$ will be saturated with respect to $p$, so that $p(O(\lchi))$ is open in $\B(\lchi)$.
	\item A strong deformation retract
	\[H\colon O(\lchi) \times [0,1] \rightarrow O(\lchi)\] 
	of $O(\lchi)$ onto $\E(\lchiminus)$ such that $H(A,t) \sim H(B,t)$ whenever $A \sim B$, where ${\sim}$ is the relation of gauge equivalence. Thus, $H$ covers a deformation retract
	\[
	h:p(O(\lchi)) \times [0,1] \rightarrow p(O(\lchi))
	\]
	of $p(O(\lchi))$ onto $\B(\lchiminus)$. Furthermore, we will prove that on fibers the restriction 
	\[
	H_{1,b}:p^{-1}(b) \rightarrow p^{-1}(h(b,1)), \quad H_{1,b}(A) = H(A, 1)
	\]
	is a weak homotopy equivalence.
\end{itemize}
As mentioned earlier, $\E(1) = \E(\leq 1) \rightarrow \B(\leq 1) = \B(1)$ is a quasifibration. Now assume that $\E(\lchiminus) \rightarrow \B(\lchiminus)$ is a quasifibration for $\chi - 1$ for some $\chi \geq 2$. Then item (\ref{ite:qf_deformation}) of Theorem \ref{thm:qf_identify} implies that $O(\lchi) \rightarrow p(O(\lchi))$ is a quasifibration. The set $\B(\lchi)\setminus\B(\lchiminus)$ is open in $\B(\lchi)$ and the maps
\begin{align*}
\E(\lchi)\setminus \E(\lchiminus) &\rightarrow \B(\lchi)\setminus \B(\lchiminus) \\
O(\lchi) \cap \E(\lchi)\setminus \E(\lchiminus) &\rightarrow p(O(\lchi)) \cap \B(\lchi) \setminus \B(\lchiminus)
\end{align*}
are quasifibrations since they are restrictions of $\E(\chi) \rightarrow \B(\chi)$. Therefore, item (\ref{ite:qf_open_sets}) of Theorem \ref{thm:qf_identify} implies that $\E(\lchi) \rightarrow \B(\lchi)$ is a quasifibration. By induction, we have that $\E(\lchi) \rightarrow \B(\lchi)$ is a quasifibration for all $\chi$. Item (\ref{ite:qf_union}) of Theorem \ref{thm:qf_identify} then implies that $p:\E \rightarrow \B$ is a quasifibration. 

Thus, all that remains to do is construct the data in bullet points above. We begin working towards this now.

Given $d, D \in \bbN$ and $A \in \cM(d, D)$, recall that we defined
\[
\lop(A) = \sum_i A^{i*}A^i  \qqtext{and} \rop(A) = \sum_i A^iA^{i*}.
\]
Recall that $\lop(A), \rop(A) \in M_D(\bbC)$ and $\lop$ and $\rop$ are continuous functions $\cM(d, D) \rightarrow M_D(\bbC)$ whose values are positive matrices. We also defined $Q(A)$ to be the projection onto the range of $\lop(A)$.

Given $A \in \cM(d, D)$ and $B, C \in M_D(\bbC)$, we denote by $BAC$ the element of $\cM(d, D)$ whose $i$th matrix is $BA^iC$. We will be particularly interested in $Q(A)A$ for $A \in \E(d, D, \chi)$. We denote this by 
\[
\underline{A} \defeq Q(A)A
\]
for ease of notation. Observe that if $X$, $K$, and $M$ correspond to $A$ as in \eqref{eq:EUMPS_def}, then
\[
\underline{A} = X\mqty(K&0\\0&0)X^{-1}
\] 
and
\[
\lop(\underline{A}) = X\mqty(\lop(K)&0\\0&0)X^{-1}.
\]
We will be considering the eigenvalues of this matrix.

\begin{definition}
Given $D \in \bbN$, let $\lambda_1,\ldots, \lambda_D:M_D(\bbC)_\tn{sa} \rightarrow \bbR$ be the ordered list of eigenvalues of a matrix repeated according to their multiplicity in the characteristic polynomial, i.e., given a self-adjoint $A \in M_D(\bbC)$, we have that the $\lambda_i(A)$ are the eigenvalues of $A$ and
\[
\lambda_1(A) \geq \cdots \geq \lambda_D(A).
\] 
Crucially, we note that the $\lambda_i$ are continuous functions.
\end{definition}

Since $K$ is an injective MPS of bond dimension $\chi$, we know $\lop(K)$ is an invertible matrix of rank $\chi$. Thus, the smallest nonzero eigenvalue of $\lop(\underline{A})$ is $\lambda_\chi(\lop(\underline{A}))$. Moreover, the strictly positive part of the spectrum of $\lop(\underline{A})$ coincides with the spectrum of $\lop(K)$. 

While $\lop(\underline{A})$ is continuous as a function of $A \in \E(d, D, \chi)$, it is not continuous on $\E(d, D, \lchi)$ since $Q$ undergoes a discontinuity whenever $\chi$ changes. 
On the other hand, the function $\lop(A)$ is continuous on $\E(d, D, \lchi)$, but its spectrum is gauge dependent and is therefore not a desirable quantity to work with for our purposes. We are therefore compelled to work with $\lop(\underline{A})$; we will have to develop some lemmas for dealing with the discontinuity and controlling the spectrum of this operator.

The following lemma is immediate from Lemma \ref{lem:projection_pt_continuity} and continuity of the function $L:\cM(d, D) \rightarrow M_D(\bbC)$.

\begin{lem}\label{lem:L(Q(A_0)Q(A)A)}
Let $d, D \in \bbN$ and $A_0 \in \E(d, D)$. The function 
\[
\E(d,D) \rightarrow M_D(\bbC), \quad A \mapsto \lop(Q(A_0)Q(A)A) = \lop(Q(A_0)\underline{A})
\]
is continuous at $A_0$.
\end{lem}

We can bound $\lop(\underline{A})$ using the following matrix inequalities. 

\begin{lemma}\label{lem:L(A)_op_bounds}
Let $d, D \in \bbN$. If $A \in \cM(d, D)$, then
\[
\lop(\underline{A}) \leq \lop(A).
\]
If $A_0 \in \E(d,D)$ and $\varepsilon > 0$, then there exists a neighborhood $U \subset \E(d, D)$ of $A_0$ such that
\[
\lop\qty(\underline{A}_0) - \varepsilon\1_{D \times D} \leq \lop(\underline{A})
\]
for all $A \in U$.
\end{lemma}

\begin{proof}
We observe:
\[
\lop(Q(A)A) = \sum_{i} A^{i*}Q(A)A^i \leq \sum_i A^{i*}A^i = \lop(A).
\]
This gives the first bound. 
Given $A_0 \in \E$, we have:
\begin{align*}
\lop(\underline{A}_0) - \norm{\lop(\underline{A}_0) - \lop(Q(A_0)\underline{A})}\1 &\leq \lop(Q(A_0)\underline{A}).
\end{align*}
The norm above can be made small as $A \rightarrow A_0$ by Lemma \ref{lem:L(Q(A_0)Q(A)A)}. We then observe that
\[
\lop(Q(A_0)\underline{A}) = \sum_{i} A^{i*}Q(A)Q(A_0)Q(A)A^i \leq \sum_{i} A^{i*}Q(A)A^i = \lop(\underline{A}),
\]
completing the proof.
\end{proof}

We now work towards obtaining greater control over the spectrum of $\lop(\underline{A})$. We will need the following lemma.

\begin{lemma}\label{lem:L(Q(A)A)_almost_range}
Let $d, D \in \bbN$. Given $A_0 \in \E(d, D)$ and $\varepsilon > 0$, there exists an open set $U \subset \E(d, D)$ containing $A_0$ such that
\[
\norm{\lop(\underline{A}) - \lop(\underline{A})Q(A_0)} < \varepsilon
\]
for all $A \in U$.
\end{lemma}

\begin{proof}
Choose $\tilde \varepsilon > 0$. Since $\lop(A)^{1/2}$ is a continuous function of $A$ and $\lop(A_0)^{1/2} = \lop(A_0)^{1/2}Q(A_0)$, there exists a neighborhood $U \subset \E(d,D)$ containing $A_0$ such that 
\[
\norm{\lop(A)^{1/2} - \lop(A)^{1/2}Q(A_0)} < \tilde \varepsilon
\]
for all $A \in U$. Recall that $\lop(\underline{A}) \leq \lop(A)$ by Lemma \ref{lem:L(A)_op_bounds}. Conjugating both sides of this inequality by $\1 - Q(A_0)$, taking norms, and using the $C^*$-property of the norm yields
\[
\norm{\lop(\underline{A})^{1/2}(\1 - Q(A_0))}^2 \leq \norm{\lop(A)^{1/2}(\1 - Q(A_0))}^2 < \tilde \varepsilon^2.
\]
Using submultiplicativity of the norm yields
\begin{align*}
\norm{\lop(\underline{A})(\1 - Q(A_0))} &\leq \norm{\lop(\underline{A})^{1/2}}\tilde \varepsilon\\
&\leq \norm{\lop(A)^{1/2}}\tilde \varepsilon\\
&\leq \norm{\lop(A)^{1/2} - \lop(A_0)^{1/2}}\tilde \varepsilon + \norm{\lop(A_0)^{1/2}}\tilde \varepsilon,
\end{align*}
where we have used operator monotonicity of the square root in the second step. Since $\lop(A)$ is continuous, the above may be made strictly less than $\varepsilon$ by shrinking $\tilde \varepsilon$ and the neighborhood $U$.
\end{proof}

The next lemma gives us the control we need on the spectrum $\sigma(\lop(\underline{A}))$. Recall that given $A_0 \in \E(d, D, \chi_0)$, the smallest nonzero eigenvalue of $\lop(\underline{A}_0)$ is the $\chi_0$th eigenvalue $\lambda_{\chi_0}(\lop(\underline{A}_0))$.

\begin{lemma}\label{lem:spectrum_gap}
Let $d, D, \chi_0 \in \bbN$ with $\chi_0 \leq D$. Given $A_0 \in \E(d, D, \chi_0)$ and $\varepsilon > 0$, there exists an open set $U \subset \E(d, D)$ containing $A_0$ such that:
\begin{equation}\label{eq:spectrum_gap}
\sigma(\lop(\underline{A})) \subset [0,\varepsilon) \cup (\lambda_{\chi_0}(\lop(\underline{A}_0)) - \varepsilon, \infty)
\end{equation}
for all $A \in U$. If $\chi_0 < D$, then $U$ can be chosen such that $\sigma(\lop(\underline{A})) \setminus \qty{0}$ has nonempty intersection with each of the intervals on the right hand side above for all $A \in U \cap \E(d, D, \chi)$ with $\chi_0 < \chi \leq D$.
\end{lemma}

\begin{proof}
Choose $\tilde \varepsilon > 0$. By Lemma \ref{lem:L(A)_op_bounds} and Lemma \ref{lem:L(Q(A)A)_almost_range}, there exists an open set $U \subset \E(d,D)$ containing $A_0$ such that
\[
L\qty(\underline{A}_0) - \tilde \varepsilon \1 \leq \lop(\underline{A})
\]
and
\[
\norm{\lop(\underline{A}) - \lop(\underline{A})Q(A_0)} < \tilde \varepsilon
\]
for all $A \in U$. Note that taking a Hermitian conjugate of the above yields
\[
\norm{\lop(\underline{A}) - Q(A_0)\lop(\underline{A})} < \tilde \varepsilon
\]
as well.

Fix $A \in U$. Suppose $\lambda$ is an eigenvalue of $\lop(\underline{A})$ and $\lambda \geq \varepsilon$. Let $x$ be a corresponding unit eigenvector. We can decompose the eigenvector equation as
\[
\lop(\underline{A})Q(A_0)x + \lop(\underline{A})(\1 - Q(A_0))x = \lambda Q(A_0)x + \lambda (\1 - Q(A_0))x.
\]
Multiplying by $\1 - Q(A_0)$ on both sides yields
\[
\norm{\lambda(\1 - Q(A_0))x} < 2\tilde \varepsilon
\]
hence
\[
\norm{(\1 - Q(A_0))x} < \frac{2\tilde \varepsilon}{\varepsilon}.
\]
Multiplying by $Q(A_0)$ on both sides of the eigenvector equation yields
\[
\norm{Q(A_0)\lop(\underline{A})Q(A_0)x - \lambda Q(A_0)x} < \tilde \varepsilon.
\]
Thus,
\begin{align*}
\lambda &\geq \ev{Q(A_0)x, \lambda Q(A_0)x}\\
&> \ev{x, Q(A_0)\lop(\underline{A})Q(A_0)x} - \tilde \varepsilon\\
&\geq \ev{x, Q(A_0)\lop(\underline{A}_0)Q(A_0)x} - \tilde \varepsilon \norm{Q(A_0)x}^2 - \tilde \varepsilon \\
&\geq \lambda_{\chi_0}(\lop(\underline{A}_0))\norm{Q(A_0)x}^2 - \tilde \varepsilon \norm{Q(A_0)x}^2 - \tilde \varepsilon\\
&\geq \lambda_{\chi_0}(\lop(\underline{A}_0))\qty(1 - \qty(\frac{2\tilde \varepsilon}{\varepsilon})^2) - \tilde \varepsilon \norm{Q(A_0)x}^2 - \tilde \varepsilon.
\end{align*}
Choosing $\tilde \varepsilon$ sufficiently small, we can achieve
\[
\lambda > \lambda_{\chi_0}(\lop(\underline{A}_0)) - \varepsilon
\]
as desired. This proves that we can find an open set $U \subset \E(d,D)$ containing $A_0$ such that \eqref{eq:spectrum_gap} holds for all $A \in U$.

If $A \in U$, then we observe that
\begin{align*}
\lambda_1(\lop(\underline{A}_0)) &= \norm{\lop(\underline{A}_0)} \leq \norm{\lop(\underline{A}) + \tilde \varepsilon \1} \leq \lambda_1(\lop(\underline{A})) + \tilde \varepsilon.
\end{align*}
Thus, choosing $\tilde \varepsilon < \varepsilon$ we see that $\sigma(\lop(\underline{A}))$ intersects the second interval in \eqref{eq:spectrum_gap}.  

Suppose $\chi_0 < D$. We know $\lambda_{\chi_0+1}(\lop(A_0)) = 0$ since the rank of $\lop(A_0)$ is $\chi_0$. Since $\lambda_{\chi_0 + 1} \circ \lop$ is continuous, we may shrink $U$ so that
\[
\lambda_{\chi_0+1}(\lop(A)) < \varepsilon
\]
for all $A \in U$. We claim that $0 < \lambda_{\chi}(\lop(\underline{A})) < \varepsilon$ for all $A \in U \cap \E(d, D, \chi)$ with $\chi_0 < \chi \leq D$. For such an $A$ we know $\lambda_{\chi}(\lop(\underline{A})) > 0$ since $\chi$ is the rank of $\lop(\underline{A})$. To show $\lambda_\chi(\lop(\underline{A})) < \varepsilon$, let $v$ be a unit eigenvector of $\lop(A)$ with eigenvalue $\lambda_\chi(\lop(A))$ and let $w_1,\ldots, w_D$ be an orthonormal basis of eigenvectors of $\lop(\underline{A})$ with corresponding eigenvalues $\lambda_1(\lop(\underline{A})),\ldots, \lambda_D(\lop(\underline{A}))$. Then
\begin{align*}
\varepsilon &> \lambda_{\chi}(\lop(A)) = \ev{v, \lop(A)v} \geq \ev{v, \lop(\underline{A})v} \\
&= \sum_{i=1}^D \lambda_i(\lop(\underline{A})) \abs{\ev{v, w_i}}^2 \geq \lambda_\chi(\lop(\underline{A})) \sum_{i=1}^\chi \abs{\ev{v, w_i}}^2 
\end{align*}
Observe that $\lop(\underline{A})$ and $\lop(A)$ have the same range, which contains $v$ and is spanned by $w_1,\ldots, w_\chi$, so $\sum_{i=1}^\chi \abs{\ev{v, w_i}}^2 = 1$. Thus, $\varepsilon > \lambda_\chi(\lop(\underline{A}))$, concluding the proof.
\end{proof}

We can finally start constructing the open set $O(\lchi)$ of $\E(\lchi)$ that we will eventually show retracts onto $\E(\lchiminus)$. Recall from Remark \ref{rem:topology_discussion} that $\E(d, D, \lchi)$ is closed in $\E(d, D)$ and $\E(\lchi)$ is closed in $\E$, and that the subspace topologies on $\E(\lchi)$ and $\E(\chi)$ coincides with the final topologies induced by the inclusions $\E(d, D, \lchi) \hookrightarrow \E(\lchi)$ and $\E(d, D, \chi) \hookrightarrow \E(\chi)$.

\begin{proposition}\label{prop:N(d,D,chi)}
Let $d, D, \chi \in \bbN$ with $2 \leq \chi \leq D$. Define
\[
N(d, D, \chi) \defeq \qty{A \in \E(d,D, \chi): \begin{array}{c} \lop(\underline{A}) \tn{ has at least two} \\ \tn{distinct nonzero eigenvalues} \end{array}}.
\]
Then $N(d, D, \chi)$ is open in both $\E(d, D, \chi)$ and $\E(d, D, \lchi)$. Furthermore, define
\[
N(\chi) \defeq \bigcup_{d, D} N(d, D, \chi).
\]
Then $N(\chi) \cap \E(d, D) = N(d, D, \chi)$, hence $N(\chi)$ is open in $\E(\chi)$ and $\E(\lchi)$.
\end{proposition}

\begin{proof}
Let $A_0 \in  N(d, D, \chi)$.  By definition of $N(d, D, \chi)$, we know that $0 < \lambda_\chi(\lop(\underline{A}_0)) < \lambda_1(\lop(\underline{A}_0))$. Since $L$, $Q$, $\lambda_{\chi}$, and $\lambda_{1}$ are continuous on $\E(d, D, \chi)$, we know there exists a neighborhood $A_0 \in U \subset \E(d, D, \chi)$ such that $A \in U$ implies $0 < \lambda_{\chi}(\lop(\underline{A})) < \lambda_{1}(\lop(\underline{A}))$,  hence $A \in N(d, D, \chi)$. This proves that $N(d, D, \chi)$ is open in $\E(d, D, \chi)$. Since $\E(d, D, \lchiminus)$ is closed in $\E(d, D)$, we know $\E(d, D, \chi)$ is open in $\E(d, D, \lchi)$, hence $N(d, D, \chi)$ is open in $\E(d, D, \lchi)$. The remarks about $N(\chi)$ are easily verified.
\end{proof}

\begin{proposition}
Let $d, D, \chi \in \bbN$ with $2 \leq \chi \leq D$. The set 
\[
O(d, D, \lchi) \defeq \E(d, D, \lchiminus) \cup N(d, D, \chi)
\]
is open in $\E(d, D, \lchi)$. Furthermore, define
\[
O(\lchi) \defeq \bigcup_{d, D} O(d, D, \lchi).
\]
Then $O(\lchi) \cap \E(d, D) = O(d, D, \lchi)$, hence $O(\lchi)$ is open in $\E(\lchi)$.
\end{proposition}

\begin{proof}
We must show that for every $A_0 \in \E(d, D, \lchiminus)$, there exists an open subset of $\E(d, D, \lchi)$ containing $A_0$ contained in $O(d, D, \lchi)$. Let $\chi_0$ be the essential rank of $A_0$. We may find an open set $U \subset \E(d, D)$ containing $A_0$ as in Lemma \ref{lem:spectrum_gap} with $\varepsilon < \lambda_{\chi_0}(\lop(\underline{A}_0))/2$. Then $U \cap \E(d, D, \lchi)$ is open in $\E(d, D, \lchi)$ and is contained in $O(d, D, \lchi)$, proving that $O(d, D, \lchi)$ is open in $\E(d, D, \lchi)$. The remarks about $O(\lchi)$ are easily verified.
\end{proof}

\begin{remark}
Note that both $N(\chi)$ and $O(\lchi)$ are intersections of an open and a closed subset of $\E$ and are therefore compactly generated in the subspace topology. By Proposition \ref{prop:subspace_colimit_equal}, the subspace topologies on $N(\chi)$ and $O(\lchi)$ coincide with the final topologies induced by the inclusions $N(d, D, \chi) \hookrightarrow N(\chi)$ and $O(d, D, \lchi) \hookrightarrow O(\lchi)$.
\end{remark}

Our next goal is to define a gauge equivariant deformation retract of $O(\lchi)$ onto $\E(\lchiminus)$. We therefore fix $\chi \geq 2$ for now. Let us set up some notation. First, we define a family of real-valued functions parametrized by $t \in [0,1]$ and $\delta \geq 0$. We define $f_{t,\delta}:\bbR \rightarrow \bbR$ by
\[
f_{t,\delta}(x) = \begin{cases} 0 &\tn{if $x \leq t\delta$} \\ \sqrt{1 - \frac{t\delta}{x}} &\tn{if $x > t\delta$.} \end{cases}
\]
Observe that if $t = 0$ or $\delta = 0$, then $f_{t,\delta}$ is the Heaviside step function $\theta$. On the other hand, if $t > 0$ and $\delta > 0$, then $t$ is continuous.

We then define $F:O(\lchi) \times [0,1] \rightarrow M_\infty(\bbC)$ using functional calculus by
\[
F(A,t) = f_{t,\delta}(\lop(\underline{A})) \qqtext{where} \delta = \delta(A) = \lambda_\chi(\lop(\underline{A})).
\]
We will typically suppress the argument in $\delta(A)$ and just write $\delta$ for cleaner notation. We note that $F$ maps $O(d, D, \lchi) \times [0,1]$ into $M_D(\bbC)$. Next, we define $G:O(\lchi) \times [0,1] \rightarrow \cM$ by 
\[
G(A,t) = AF(A,t).
\]
Again, note that $G$ maps $O(d, D, \lchi) \times [0,1]$ into $\cM(d, D)$. If $X$, $K$, and $M$ correspond to $A$ as in Definition \ref{def:EGL}, then
\[
F(A,t) = X\mqty(f_{t,\delta}(\lop(K)) & 0\\0 & 0)X^*
\]
and
\[
G(A,t)^i = X \mqty(K^if_{t,\delta}(\lop(K)) & 0 \\ M^i f_{t,\delta}(\lop(K)) & 0)X^*.
\]

While $F$ is not continuous because of discontinuities in $L(\underline{A})$ and the discontinuity in $f_{t,\delta}$ when $t = 0$ and $\delta = 0$, we can in fact show that $G$ is continuous. 
We will require the following lemma. 

\begin{lemma}\label{lem:h_(t,delta)(L(Q(A)A))Q(A_0)}
Given $A_0 \in \E(d, D, \lchiminus)$, the function
\[
O(d, D, \lchi) \times [0,1] \rightarrow M_D(\bbC), \quad (A, t) \mapsto F(A,t)Q(A_0)
\]
is continuous at $(A_0, t_0)$ for all $t_0 \in [0,1]$.
\end{lemma}

\begin{proof}
For $A_0 \in \E(d, D, \lchiminus)$ we have $\delta = 0$, so
\[
F(A_0,t)Q(A_0) = \theta(L(\underline{A}_0))Q(A_0) = Q(A_0).
\]
It therefore suffices to show that for every $\varepsilon > 0$, there exists an open set $U \subset O(d, D, \lchi)$ containing $A_0$ such that
\begin{equation}\label{eq:f_t_delta_range}
\norm{F(A,t)Q(A_0) - Q(A_0)} < \varepsilon
\end{equation}
for all $A \in U$ and $t \in [0,1]$. 

Let $\chi_0$ be the essential rank of $A_0$ and 
set $\gamma = \lambda_{\chi_0}(\lop(\underline{A}_0))$ for ease of notation. By Lemmas \ref{lem:L(A)_op_bounds}, \ref{lem:L(Q(A)A)_almost_range}, and \ref{lem:spectrum_gap}, we know that for any $\tilde \varepsilon \in (0, \gamma/2)$, there exists an open set $U \subset O(d, D, \lchi)$ containing $A_0$ such that for all $A \in U$ we have:
\begin{itemize}
	\setlength\itemsep{0.5em}
	\item $\lop(\underline{A}_0) - \tilde \varepsilon \1 \leq \lop(\underline{A})$,
	\item $\norm{\lop(\underline{A}) - \lop(\underline{A})Q(A_0)} < \tilde \varepsilon$,
	\item $\sigma(\lop(\underline{A})) \subset [0,\tilde \varepsilon) \cup (\gamma - \tilde \varepsilon, \infty)$, with $\delta = \lambda_\chi(\lop(\underline{A})) \in [0,\tilde \varepsilon)$.
\end{itemize}
We show that \eqref{eq:f_t_delta_range} holds for all $A \in U$ and $t \in [0,1]$, provided that $\tilde \varepsilon$ is small enough.

Suppose $v$ is a unit eigenvector of $\lop(\underline{A})$ with eigenvalue $\lambda$. Then
\begin{align*}
\lambda \norm{v - Q(A_0)v} = \norm{\lop(\underline{A})v - Q(A_0)\lop(\underline{A})v} < \tilde \varepsilon.
\end{align*}
If $\lambda \in (\gamma - \tilde \varepsilon, \infty)$, then
\[
\norm{v - Q(A_0)v} \leq \frac{\tilde \varepsilon}{\gamma - \tilde \varepsilon}.
\]
Thus,
\begin{align*}
\norm{F(A,t)Q(A_0)v - Q(A_0)v} &\leq \norm{f_{t,\delta}(\lop(\underline{A}))v - v}  + \frac{2\tilde \varepsilon}{\gamma - \tilde \varepsilon} \\
&= \qty(1 - \sqrt{1 - \frac{t\delta}{\lambda}}) + \frac{2\tilde \varepsilon}{\gamma - \tilde \varepsilon}\\
&\leq 1 - \sqrt{1 - \frac{t\tilde \varepsilon}{\gamma - \tilde \varepsilon}} + \frac{2\tilde \varepsilon}{\gamma - \tilde \varepsilon}.
\end{align*}
This tends to zero as $\tilde \varepsilon \rightarrow 0$.

Now suppose $\lambda \in [0,\tilde \varepsilon)$. Then, by the first bullet point above,
\[
\lambda + \tilde \varepsilon \geq \ev{v, \lop(\underline{A}_0)v} \geq \gamma \norm{Q(A_0)v}^2,
\]
hence
\[
\norm{F(A,t)Q(A_0)v - Q(A_0)v} \leq 2\norm{Q(A_0)v} < 2\sqrt{\frac{2\tilde \varepsilon}{\gamma}}.
\]
Again, this can be made arbitrarily small as $\tilde \varepsilon \rightarrow 0$. Since we have an orthonormal basis of unit eigenvectors of $\lop(\underline{A})$ on which $F(A,t)Q(A_0)$ is close to $Q(A_0)$, this proves the result. 
\end{proof}

\begin{lemma}\label{lem:G_continuous}
The function $G:O(\lchi) \times [0,1] \rightarrow \cM$ is continuous.
\end{lemma}

\begin{proof}
As $O(\lchi)$ is the colimit of the spaces $O(d, D, \lchi)$, 
it suffices to show that the restriction $G:O(d, D, \lchi) \times [0,1] \rightarrow \cM(d, D)$ is continuous. We note that $G$ is manifestly continuous on $N(d, D, \chi) \times (0,1]$, which is open in $O(d, D, \lchi) \times [0,1]$. We next show continuity on points of $N(d, D, \chi) \times \qty{0}$.  Given $A_0 \in N(d, D, \chi)$, we have $G(A_0, 0) = A_0Q(A_0) = A_0$. Then given $(A,t) \in O(d, D, \lchi) \times [0,1]$, we observe that
\begin{align*}
\norm{G(A,t)^i - A^i_0} &\leq \norm{A^i[F(A,t) + \1 - Q(A)] - A_0^i}\\
&\leq \norm{A^i - A_0^i} + \norm{A^i}\norm{F(A,t) - Q(A)}.
\end{align*}
By functional calculus, we have
\[
\norm{F(A,t) - Q(A)} = 1 - \sqrt{1 - t}.
\]
Plugging this in, we obtain
\[
\norm{G(A,t)^i - A^i_0} \leq (2 - \sqrt{1 - t})\norm{A^i - A_0^i} + \norm{A_0^i}\qty(1 - \sqrt{1 - t}).
\]
This is arbitrarily small for sufficiently small $t$ and $A$ sufficiently close to $A_0$.

It remains to prove continuity at $(A_0, t_0) \in \E(d, D, \lchiminus) \times [0,1]$. Observe that in this case $\delta(A_0) = 0$ so $G(A_0,t_0) = A_0Q(A_0) = A_0$. If $(A, t) \in O(d, D, \lchi) \times [0,1]$, then 
\[
\norm{G(A,t)^i - A_0^i} \leq \norm{A^i - A_0^i} + \norm{A_0^iF(A,t) - A_0^i}.
\]
It suffices to show that the last norm above can be made small for $A$ sufficiently close to $A_0$. We observe that
\begin{align*}
\norm{A_0^iF(A,t) - A_0^i}^2 &= \norm{\qty[F(A,t) - Q(A_0)]A_0^{i*}A_0^i\qty[F(A,t) - Q(A_0)]}  \\
&\leq \norm{\qty[F(A,t) - Q(A_0)]\lop(A_0)\qty[F(A,t) - Q(A_0)]}\\
&\leq \norm{\lop(A_0)^{1/2}[F(A,t) - Q(A_0)]}\\
&= \norm{\lop(A_0)^{1/2}}\norm{Q(A_0)F(A,t) - Q(A_0)},
\end{align*}
and this can be made arbitrarily small for $(A,t)$ sufficiently close to $(A_0, t_0)$ by Lemma \ref{lem:h_(t,delta)(L(Q(A)A))Q(A_0)}. 
\end{proof}

We are almost ready to define our deformation retract. The issue with $G(A,t)$ is that $K f_{t,\delta}(\lop(K))$ is not right-normalized as an MPS tensor. To normalize it, we do the following. Given $(A, t) \in O(\lchi) \times [0,1]$ with $X$, $K$, and $M$ corresponding to $A$, observe that
\begin{align*}
\rop(Q(A)G(A,t)) = X\mqty(\sum_{i}K^i f_{t,\delta}(\lop(K))^2 K^{i*} & 0 \\0 & 0)X^*.
\end{align*}
Observe that $f_{t,\delta}(\lop(K))$ is positive and nonzero since there exists $\lambda \in \sigma(K)$ such that $\lambda > \delta$ by definition of $O(\lchi)$. 
Since $K$ is an injective MPS tensor, $\sum_i K^i f_{t,\delta}(\lop(K))^2K^{i*}$ is positive and, moreover, invertible. Therefore, the matrix
\begin{equation}\label{eq:S(A,t)}
R(Q(A)G(A,t)) + \1 - Q(A) = X\mqty(\sum_{i}K^i f_{t,\delta}(\lop(K))^2 K^{i*} & 0 \\0 & \1)X^*
\end{equation}
is positive and invertible, where the $\1$ on the left hand side is the $\bbN \times \bbN$ identity matrix. Let us define $S:O(\lchi) \times [0,1] \rightarrow \GL(\infty)$ by
\[
S(A,t) = \rop(Q(A)G(A,t)) + \1 - Q(A)
\]
for ease of notation.

\begin{lemma}\label{lem:S_continuous}
The function $S:O(\lchi) \times [0,1] \rightarrow \GL(\infty)$ is continuous and maps $O(d, D, \lchi) \times [0,1]$ into $\GL(D) \subset \GL(\infty)$. 
\end{lemma}

\begin{proof}
It is clear that $S$ maps $O(d, D, \lchi) \times [0,1]$ into $\GL(D)$. As $O(\lchi)$ is the colimit of the $O(d, D, \lchi)$ it suffices to show continuity of this restriction. Since $G$ is continuous by Lemma \ref{lem:G_continuous} and $Q$ is continuous on $N(d, D, \chi)$, we see that $S$ is continuous on $N(d, D, \chi) \times [0,1]$. 

We must show continuity at $(A_0, t_0) \in \E(d, D, \lchiminus) \times [0,1]$. Since $G(A_0, t_0) = A_0$ and $\rop(Q(A_0)A_0) = Q(A_0)$, we see that $S(A_0,t_0) = \1$. Given $\varepsilon > 0$, there exists a neighborhood $(A_0, t_0) \in U \subset O(d, D, \lchi) \times [0,1]$ such that 
\[
\norm{\rop(G(A,t)) - \rop(A)} < \varepsilon
\]
for all $(A,t) \in U$. Then for all $(A,t) \in U$ we have:
\begin{align*}
\norm{S(A,t) - \1} &= \norm{R(Q(A)G(A,t)) - Q(A)}\\
&=\norm{Q(A)\rop(G(A,t))Q(A) - Q(A)}\\
&< \varepsilon + \norm{Q(A)\rop(A)Q(A) - Q(A)} = \varepsilon,
\end{align*}
where the identity $Q(A)\rop(A)Q(A) = Q(A)$ can be checked by representing $A$ in terms of $X$, $K$, and $M$ as in Definition \ref{def:EGL} and using the right normalization condition on $K$. 
\end{proof}

\begin{theorem}\label{thm:deformation_retract}
The function $H:O(\lchi) \times [0,1] \rightarrow O(\lchi)$ defined by
\[
H(A,t) = S(A,t)^{-1/2}G(A,t)
\]
is a well-defined deformation retract of $O(\lchi)$ onto $\E(\lchiminus)$ such that $H(A,t) \sim H(B,t)$ whenever $A \sim B$. Furthermore, $H$ preserves $O(d, D, \lchi)$.
\end{theorem}

\begin{proof}
Let us first show that $H$ is well-defined, i.e., that $H(A,t) \in O(\lchi)$ for all $(A,t) \in O(\lchi) \times [0,1]$. If $A \in \E(\lchiminus)$ or $t = 0$, then $G(A,t) = A$ and $S(A,t) = \1$, hence $H(A,t) = A$. Let us examine the case where $A \in N(\chi)$ and $t > 0$.

Let $X$, $K$, and $M$ correspond to $A$ as in Definition \ref{def:EGL}. Then
\[
H(A,t)^i = X\mqty(\rop(Kf_{t,\delta}(K))^{-1/2}K^if_{t,\delta}(\lop(K)) & 0 \\ M^i f_{t,\delta}(\lop(K)) & 0 )X^*.
\]
Consider first the case where $t < 1$. Then $f_{t,\delta}(\lop(K))$ is invertible, hence the upper left corner 
\[
\tilde K \defeq \rop(Kf_{t,\delta}(K))^{-1/2}K f_{t,\delta}(\lop(K))
\]
is an injective MPS tensor. One can easily check that it is right-normalized. Thus, $H(A,t) \in \E(\chi)$.

To show that $H(A,t) \in O(\lchi)$ we must show that $\lop(\tilde K)$ has two distinct nonzero eigenvalues. For this we need estimates for $R(Kf_{t,\delta}(L(K)))^{-1}$. By functional calculus we have
\[
(1 - t)\1 \leq f_{t,\delta}(L(K))^2 \leq \qty(1 - \frac{t\delta}{\lambda_1(\lop(K))})\1,
\]
hence
\[
(1 - t)\1 \leq R(Kf_{t,\delta}(L(K))) \leq \qty(1 - \frac{t\delta}{\lambda_1(\lop(K))})\1
\]
hence
\[
\qty(1 - \frac{t\delta}{\lambda_1(\lop(K))})^{-1}\1 \leq \rop(Kf_{t,\delta}(\lop(K)))^{-1}\leq (1 - t)^{-1}\1.
\]

Let $v_1$ and $v_\chi$ be unit eigenvectors of $\lop(K)$ with eigenvalues $\lambda_1 = \lambda_1(\lop(K))$ and $\lambda_\chi = \lambda_\chi(\lop(K))$ respectively. We know $\lambda_1(\lop(K)) > \lambda_\chi(\lop(K))$ since $A \in N(\chi)$. Since $\tilde K$ is an injective MPS, we also know $\lop(\tilde K)$ is positive and invertible. Therefore if $\lop(\tilde K)$ does not have two distinct nonzero eigenvalues, then it is a scalar multiple of the identity. But:
\begin{align*}
\ev{v_1, \lop(\tilde K) v_1} &= \qty(1 - \frac{t\delta}{\lambda_1}) \sum_i \ev{v_1, K^{i*}R(Kf_{t,\delta}(\lop(K)))^{-1}K^i v_1}\\
&\geq \ev{v_1, \lop(K)v_1} = \lambda_1
\end{align*}
while
\begin{align*}
\ev{v_\chi, \lop(\tilde K)v_\chi} &= \qty(1 - t) \sum_i \ev{v_\chi, K^{i*}R(Kf_{t,\delta}(\lop(K)))^{-1}K^i v_\chi}\\
&\leq \ev{v_\chi, \lop(K)v_\chi} = \lambda_\chi.
\end{align*}
Since $\lambda_\chi < \lambda_1$, we see that $\lop(\tilde K)$ is not a scalar multiple of the identity. This proves that $H(A,t) \in N(\chi)$.

Now consider the case where $t = 1$. In this case, $f_{1,\delta}(\lop(K))$ is not invertible and may be diagonalized as
\[
f_{1,\delta}(\lop(K)) = Y\mqty(\Lambda & 0 \\ 0 & 0)Y^*
\]
where $Y \in \Unitary(\chi)$ and $\Lambda$ is a $\chi' \times \chi'$ diagonal matrix with strictly positive entries on the diagonal and $1\leq \chi' < \chi$. For ease of notation, let us write $\tilde R =  R(Kf_{1,\delta}(\lop(K)))$, so
\[
S(A,t)  = X\mqty(\tilde R & 0 \\0 & \1) X^*.
\] 
Then
\begin{equation}\label{eq:H_at_1}
H(A,1)^i = X\mqty(Y&0\\0&\1)\mqty(Y^*\tilde R^{-1/2}K^iY\mqty(\Lambda & 0 \\0 & 0) & 0 \\ M^iY\mqty(\Lambda  & 0 \\0 & 0) & 0)\mqty(Y^* & 0 \\0 & \1)X^*.
\end{equation}
Let us decompose into block form:
\begin{equation}\label{eq:J_matrices}
Y^*\tilde R^{-1/2}K^i Y = \mqty(J_{11}^i & J_{12}^i \\ J_{21}^i & J_{22}^i),
\end{equation}
where $J_{11}$ is $\chi' \times \chi'$. Since $Y^*\tilde R^{-1/2}K Y$ is an injective MPS tensor, we know $J_{11}$ is an injective MPS tensor of bond dimension $\chi' < \chi$ and therefore so is $J_{11}\Lambda$. To see that $J_{11}\Lambda$ is right-normalized, observe that
\begin{align*}
\1 = Y^*\1Y &= \sum_i Y^*\tilde R^{-1/2} K^{i}f_{1,\delta}(\lop(K))^2 K^{i*}\tilde R^{-1/2}Y\\
&= \sum_i Y^*\tilde R^{-1/2}K^iY\mqty(\Lambda^2 & 0\\0 & 0)Y^*K^{i*}\tilde R^{-1/2}Y\\
&= \sum_i \mqty(J_{11}^i & J_{12}^i \\ J_{21}^i & J_{22}^i)\mqty(\Lambda^2 & 0 \\0 & 0)\mqty(J_{11}^{i*} & J_{21}^{i*} \\ J_{12}^{i*} & J_{22}^{i*})\\
&= \sum_i \mqty(J_{11}^i\Lambda^2J_{11}^{i*} & J_{11}^i\Lambda^2 J_{21}^{i*} \\ J_{21}^i\Lambda^2J_{11}^{i*} & J_{21}^{i}\Lambda^2 J_{21}^{i*}) .
\end{align*}
Examination of the upper left block of this equation proves that $J_{11}\Lambda$ is right-normalized. Thus, $H(A,1) \in \E(\lchiminus)$.

Let us show next that $H(A,t)$ is gauge equivariant. Let $A, B \in O(\lchi)$ such that $A \sim B$. Then $A \in \E(\lchiminus)$ if and only if $B \in \E(\lchiminus)$, in which case $H(A,t) = A \sim B = H(B,t)$. Similarly, if $t = 0$, then $H(A,0) = A \sim B = H(B,0)$. Suppose $A, B \in N(\chi)$ and $t > 0$. We may write $B$ as a gauge transformation $B = \lambda Z(A + \tilde A)Z^*$ as in Definition \ref{def:gauge_transform}. Then $\underline{B} = \lambda Z\underline{A}Z^*$, hence $\lop(\underline{B}) = Z\lop(\underline{A})Z^*$. Thus, $\delta(A) = \delta(B)$, so $F(B,t) = ZF(A,t)Z^*$ and
\[
G(B,t) = \lambda ZG(A,t)Z^* + \lambda Z\tilde AF(A,t)Z^*.
\]
Continuing, one can check that $S(B,t) = ZS(A,t)Z^*$, hence
\begin{align*}
H(B,t) &= \lambda Z\qty[H(A,t) + S(A,t)^{-1/2}\tilde AF(A,t)]Z^*\\
&= \lambda Z \qty[H(A,t) + \tilde AF(A,t)]Z^*,
\end{align*}
where we see that $S(A,t)^{-1/2}\tilde A = \tilde A$ from the expression \eqref{eq:S(A,t)} for $S(A,t)$. Furthermore, we see that $Q(H(A,t))Q(A) = Q(H(A,t))$, and this implies that $Q(H(A,t))\tilde A = 0$ since $Q(A)\tilde A = 0$ by Definition \ref{def:gauge_transform}. Recall that $Q(H(A,t))$ is the projection onto the range of:
\[
\lop(H(A,t)) = \sum_i F(A,t)A^{i*}S(A,t)^{-1}A^iF(A,t).
\]
If $x \in \ker \lop(H(A,t))$, then $S(A,t)^{-1/2}A^iF(A,t)x = 0$ for all $i$. Thus, $A^iF(A,t)x = 0$ for all $i$. Multiplying by $A^{i*}$ and summing over $i$, we see that $\lop(A)F(A,t)x = 0$. Since $Q(A)F(A,t) = F(A,t)$, we see that $F(A,t)x$ is in the orthogonal complement of the kernel of $\lop(A)$. This implies that $F(A,t)x = 0$. Thus, $\ker Q(H(A,t)) = \ker \lop(H(A,t)) \subset \ker F(A,t)$, from which it follows that $F(A,t) = F(A,t)Q(H(A,t))$. Thus, $H(A,t) \sim H(B,t)$.

Continuity of $H$ is clear from Lemmas \ref{lem:G_continuous} and \ref{lem:S_continuous} and it is clear from the definition that $H$ preserves $O(d, D, \lchi)$, so we're done.
\end{proof}

Let $p:\E \rightarrow \B$ be the projection. 
Since $H$ is gauge equivariant, it covers a deformation retract $h:p(O(\lchi)) \times [0,1] \rightarrow p(O(\lchi)) \times [0,1]$ onto $\B(\lchiminus)$, i.e., the diagram below commutes.
\[
\begin{tikzcd}
O(\lchi) \times [0,1] \arrow[r,"H"] \arrow[d,"p \times \id"] & O(\lchi) \arrow[d,"p"]\\
p(O(\lchi)) \times [0,1] \arrow[r,"h"]& p(O(\lchi))
\end{tikzcd}
\]
Observe that $O(\lchi)$ is saturated with respect to the projection $p$, hence $p(O(\lchi))$ is open in $\B$.

We need to show that the homotopy $H$ is a weak homotopy equivalence on fibers at time $t = 1$. In other words, given $b \in p(O(\lchi))$, we must show that $H_{1,b}:p^{-1}(b) \rightarrow p^{-1}(h(b,1))$, $H_{1,b}(A) = H(A, 1)$ is a weak homotopy equivalence. If $b \in p(\E(\lchiminus))$, then $H_{1,b}$ is the identity, so this is trivial. We will consider the case where $b \in p(N(\chi))$. Then $b = p(K)$ for some right-normalized injective MPS tensor $K$ of bond dimension $\chi$ and some physical dimension $d$.

\begin{theorem}
Let $p(K) \in p(N(\chi))$, where $K$ is an injective MPS tensor of bond dimension $\chi$ and some physical dimension $d$. Let $\chi' < \chi$ be such that $h(p(K), 1) \in \B(\chi')$.  There exists a unitary $Y \in \Unitary(\chi)$, a positive invertible diagonal $\chi' \times \chi'$ matrix $\Lambda$, and $(\chi - \chi') \times \chi'$ matrices $J^i$ for $i \in \qty{1,\ldots, d}$, all depending only on $K$, such that the diagram below commutes
\[
\begin{tikzcd}
p^{-1}(p(K)) \arrow[r,"H_{1,b}"]\arrow[d,"\cong"]& p^{-1}(h(p(K),1))\arrow[d,"\cong"]\\
\cF(\chi) \arrow[r,"\zeta"] & \cF(\chi')
\end{tikzcd}
\]
where the map $\zeta: \cF(\chi) \rightarrow \cF(\chi')$ is given by
\[
\zeta(\lambda, [X, M]) = \qty(\lambda, \qty[X\mqty(Y&0\\0&\1), \mqty(J\Lambda \\ MY\mqty(\Lambda \\ 0))]).
\]
Furthermore, $\zeta$ is a weak homotopy equivalence, so $H_{1,b}$ is a weak homotopy equivalence.
\end{theorem}

\begin{proof}
That the matrices $Y$, $\Lambda$, and $J$ exist such that $\zeta$ makes the diagram commute follows from direct examination of \eqref{eq:H_at_1} and \eqref{eq:J_matrices} in the proof of Theorem \ref{thm:deformation_retract} (where $J = J_{21}$).

We show that $\zeta$ is a weak homotopy equivalence. 
We have a commutative diagram
\[
\xymatrix{
\cF(\chi) \ar[r]^-\zeta \ar[d] & \cF(\chi') \ar[d] \\
\Unitary(1) \times \PV_{\chi}(\bbC^\infty) \ar[r]^-{\bar{\zeta}} & \Unitary(1) \times \PV_{\chi'}(\bbC^\infty)
}
\]
where the vertical arrows are the weak homotopy equivalences of Proposition \ref{prop:Stiefel_weak_homotopy_equiv} and
\[
\bar \zeta(\lambda, [X_1,\ldots, X_\chi]) = (\lambda, [X_1',\ldots, X_{\chi'}'])
\]
where $X_i'$ is the $i$th column of the matrix $\mqty(X_1&\cdots&X_\chi)Y$. Since the vertical arrows are weak homotopy equivalences, it suffices to show that $\bar \zeta$ is a weak homotopy equivalence.

Note that $\Unitary(\chi)$ has a continuous right group action on $\PV_\chi(\bbC^\infty)$ given by $[X_1,\ldots, X_\chi]Y = [X_1',\ldots, X_\chi']$, where $X_i'$ is the $i$th column of the matrix $\mqty(X_1&\cdots & X_{\chi'})Y$. Then observe that $\bar \zeta$ is the composition of the action of $Y$ with the map
\begin{equation}\label{eq:forget_last_vectors}
\eta(\lambda, [X_1,\ldots, X_\chi]) = (\lambda, [X_1,\ldots, X_{\chi'}]).
\end{equation}
Since the action of $Y$ is a homeomorphism (with inverse given by the action of $Y^{-1}$), it suffices to show that $\eta$ is a weak homotopy equivalence.

The map $\eta$ is the colimit of maps $\eta_D:\Unitary(1) \times \PV_\chi(\bbC^D) \rightarrow \Unitary(1) \times \PV_{\chi'}(\bbC^D)$ defined by the same formula \eqref{eq:forget_last_vectors}. Each map $\eta_D$ is a fiber bundle with typical fiber $V_{\chi - \chi'}(\bbC^D)$. So $\eta$ is a Serre fibration, thus a quasi-fibration, and the typical fiber is the contractible space $V_{\chi - \chi'}(\bbC^\infty)$. Since its target is path connected, $\eta$ is a weak homotopy equivalence.
\end{proof}

We have now constructed the data in the bullet points at the start of the section. Therefore, by the argument at the start of the section, we now have the following result.

\begin{theorem}\label{thm:quasifibration}
The map $p : \E \to \B$ is a quasifibration.
\end{theorem}

\section{Weak Homotopy type of \texorpdfstring{$\B$}{B}}\label{sec:homotopytype}\label{sec:homotopy_type}

In this section, we identify the weak homotopy type of $\B$. Choose, for concreteness, the base point $e$ of $\E$ to be the image of $1 \in \E(1,1,1)$ under the inclusion, the base point of $b\in \B$ to be the image of $e$ under the quotient map $p$. We let $\F = p^{-1}(b)$ and $e\in \cF$ be the basepoint, and $i\colon \F \to \E$ be the inclusion. All homotopy groups $\pi_n(-)$ for $n>0$ will be taken relative to these based points.

Recall that for $A$ an abelian group and $n>0$, a based space $X$ is an Eilenberg--Mac Lane space of type $K(A,n)$ if $\pi_nX=A$ and $\pi_mX$ is trivial for $m\neq n$. More colloquially, we just call $X$ a $K(A,n)$. In Proposition~\ref{prop:Stiefel_weak_homotopy_equiv}, we proved there was a homotopy equivalence 
\[\cF = \cF(1) \simeq \Unitary(1)\times \PV_{1}(\bbC^\infty) . \]
But $\Unitary(1)$ is a $K(\Z,1)$ and $\PV_{1}(\bbC^\infty)$ is the infinite complex projective space, so is a $K(\Z,2)$.
\begin{lem}
The fiber $\F$ is homotopy equivalent  to a product of Eilenberg--Mac Lane spaces $K(\Z,1)\times K(\Z,2)$. In particular, 
\[\pi_n\F \cong \begin{cases} \Z & n=1,2 \\
0 & \text{otherwise.}
\end{cases}\]
\end{lem}

The spaces $\B,\F,\E$ are path connected, so all have trivial set of path components $\pi_0$.
Since $p: \E\rightarrow \B$ is a quasifibration, we obtain a long exact sequence on homotopy groups,
\[ \xymatrix@C=1.2pc{ \cdots \ar[r] & \pi_n\F \ar[r]^-{i_*}  & \pi_n\E \ar[r]^{p_*} & \pi_n\B \ar[r]^-{\delta} & \pi_{n-1}\F \ar[r]^-{i_*} & \cdots \ar[r]^-{i_*} \ar[r] & \pi_1\E  \ar[r]^-{p_*} & \pi_1\B  \ar[r]^-{\delta} & \pi_0\F.
}\]
Since $\E$ is contractible, all of its homotopy groups are trivial. From this, we obtain:
\begin{prop}\label{prop:connectingiso}
The connecting homomorphisms $\delta : \pi_n \B \rightarrow \pi_{n-1}\F$ induces an isomorphism for all $n>1$. In particular,
\[\pi_n\B \cong \begin{cases} \Z & n=2,3 \\
0 & \text{otherwise.}
\end{cases}\]
\end{prop}

	It remains to prove that $\B$ is a product of Eilenberg--Mac Lane spaces. To do this, we break the gauge transformations of Definition \ref{def:gauge_transform} into two parts. We define $\B_2$ to be the quotient of $\E$ by the gauge transformations $A \mapsto \lambda A$ for $\lambda \in \Unitary(1)$. We define $\B_3$ to be the quotient of $\E$ by the gauge transformations $A \mapsto Z(A + \tilde A)Z^*$ where $Z \in \Unitary(\infty)$ and $\tilde A \in \cM$ satisfies $Q(A)\tilde A^i = 0$ and $\tilde A^iQ(A) = \tilde A^i$ for all $i$. Thus, these are intermediate quotients
	\[ \xymatrix{\E \ar[r]^-{p_2} &  \B_2 \ar[r]^-{q_2} & \B} \quad \quad \text{and} \quad \quad \xymatrix{\E \ar[r]^-{p_3} &  \B_3 \ar[r]^{q_3} & \B} \]
	where $p = q_2 \circ p_2 = q_3 \circ p_3$.

	\begin{prop}\label{prop:B3_B_bundle}
	The map $q_3:\B_3 \rightarrow \B$ is a fiber bundle with fiber $\Unitary(1)$. 
	\end{prop}

	\begin{proof}
	For each $i \in \bbN$, we have well-defined continuous maps 
	\begin{alignat*}{2}
	\B_3 &\rightarrow \bbC, &\quad p_3(A) &\mapsto \tr(A^i)\\
	\B &\rightarrow \bbC, &\quad p(A) &\mapsto \abs{\tr(A^i)}.
	\end{alignat*}
	Fix an arbitrary point $p(A_0) \in \B$ with $A_0 = X\mqty(K&0\\M&0)X^* \in \E(\chi)$. We know there exists $i \in \bbN$ such that $\tr(A^i_0) \neq 0$, otherwise the $K^i$ could not span all of $M_\chi(\bbC)$ since they would all be traceless. Consider the open neighborhood $p(A_0) \in O \defeq \qty{p(A): \abs{\tr(A^i)} \neq 0}$. Define a local trivialization $\phi:q_3^{-1}(O) \rightarrow O \times \Unitary(1)$ by 
	\[
	\phi(p_3(A)) = \qty(p(A), \frac{\tr(A^i)}{\abs{\tr(A^i)}}).
	\]
	Then $\phi$ is obviously continuous and commutes with the projection onto $O$. A continuous inverse is given by the map
	\[
	(p(A), \lambda) \mapsto p_3\qty(\frac{\lambda\abs{\tr(A^i)}}{\tr(A^i)} A^i).
	\]
	Thus, $\phi$ is a homeomorphism, proving local triviality of $q_3$.
	\end{proof}

	The following is an immediate corollary of Proposition \ref{prop:B3_B_bundle} and Corollary \ref{cor:B_colimits_Hausdorff}.

	\begin{cor}
	The space $\B_3$ is Hausdorff.
	\end{cor}

	We continue to use $e = 1$ and its images $b_j=p_j(e)$ as our base points for $\B_j$. The map $p_2:\E \to \B_2$ is the quotient of a completely regular space by the free action of the compact Lie group and is therefore a fiber bundle. By the same proof as in Section \ref{sec:quasifibration}, we have that $p_3:\E \rightarrow \B_3$ is a quasifibration. Summarizing, we have the following proposition.

\begin{proposition}
The maps $p_2 : \E \to \B_2$ and $ p_3 : \E \to \B_3$ are quasifibrations, with fibers $\F_2=U(1)$ and 
\[\F_3 = \colim_{d,D} \qty(\Unitary(D) \times_{\Unitary(1) \times \Unitary(D - 1)} M_{(D - 1)\times 1}(\bbC)^d) \xrightarrow{\simeq}  \PV_{1}(\bbC^\infty).\] 
 Consequently, $\B_2$ is a $K(\Z,2)$ and $\B_3$ a $K(\Z,3)$. 
\end{proposition}
Both $\B_2$ and $\B_3$ admit continuous maps to $\B$ via the universal property of the quotient. We thus get commutative diagrams
\begin{equation}\label{eq:Fj_E_Bj}
\xymatrix{
\F_j \ar[r] \ar[d]& \E \ar[d]\ar[r] & \B_j\ar[d] \\
\F \ar[r] & \E \ar[r] & \B .
}
\end{equation}
By inspection, we see that $\F_2 \rightarrow \F$ is the map
\[ \Unitary(1) \to \F\]
which sends $\lambda$ to $(\lambda, [\1,0])$. This induces an isomorphism on $\pi_1$. On the other hand, $\F_3 \rightarrow \F$ is the map which sends $[X,M]$ to $(1,[X,M])$. This induces an isomorphism on $\pi_2$. 

The diagram \eqref{eq:Fj_E_Bj} induces a diagram of long exact sequences
\[ \xymatrix@C=1.2pc{ 
 0 \ar@{=}[r]  & \pi_3\E \ar[r]^-{p_*} \ar[d] & \pi_3\B_j \ar[r]^-{\delta} \ar[d] & \pi_{2}\F_j \ar[r]^-{i_*} \ar[d]& \pi_{2}\E \ar[r]^-{p_*} \ar[d] &  \pi_{2}\B_j \ar[r]^-{\delta}  \ar[d]&  \pi_{1}\F_j \ar[r] \ar[d] & \pi_1 \E \ar[d] \ar@{=}[r] &  0\\
 0 \ar@{=}[r]  & \pi_3\E \ar[r]^-{p_*} & \pi_3\B \ar[r]^-{\delta} & \pi_{2}\F \ar[r]^-{i_*} &  \pi_{2}\E \ar[r]^-{p_*}   &  \pi_2\B \ar[r]^-{\delta} &  \pi_{1}\F \ar[r] & \pi_1 \E \ar@{=}[r] &  0.
}\]
Since $\pi_2\cE=0$, this splits into two diagrams, and we extract for $j=2,3$,
\[ \xymatrix@C=1.2pc{ 
0  \ar[r] &  \pi_{j}\B_j \ar[r]^-{\delta}_-\cong  \ar[d]&  \pi_{j-1}\F_j \ar[r] \ar[d]^-\cong &  0\\
0 \ar[r]   &  \pi_j\B \ar[r]^-{\delta}_-\cong   &  \pi_{j-1}\F \ar[r] & 0
}\]
which implies that $\B_j \rightarrow \B$ induces an isomorphism on $\pi_j$. Note further that these are based maps, preserving our base points for each space.

We make a remark to motivate the proof below. Consider the product $\B_2 \times \B_3 $ and let
\[\B_2 \vee \B_3  = \B_2 \times \{b_3\} \cup \{b_2\} \times \B_3.\]
Since $q_2(b_2)=b=q_3(b_3)$ we get a continuous function
\[ q_0 \colon \B_2 \vee \B_3 \to \B,\]
We would like to extend this to the whole product. Unfortunately, the universal property of the product is useful for constructing maps into the product, not for constructing maps out of it. Fortunately, we only care about the weak homotopy type, so we can replace $\B_2\vee \B_3$ by a CW-complex and realize the corresponding wedge as a sub-complex of the product. CW-complexes, being built inductively from the pushout diagrams attaching their cells, have the right universal properties these kinds of extension. With this in mind, we state and prove the final result of this section.

\begin{theorem}\label{thm:weakequivalence}
There is a weak homotopy equivalence
\[  K(\Z,2)\times K(\Z,3) \xrightarrow{\simeq} \B.\]
\end{theorem}

\begin{proof}
Since homotopy groups of a simply connected finite CW-complex are finitely generated abelian groups, we can construct $K(\Z,j)$ to have finitely many cells in each dimension.
Then, a generator $S^j \rightarrow \B_j$ for $\pi_j$ extends to a map $K(\Z,j) \rightarrow \B_j$ which is a weak homotopy equivalence. We thus get a CW-replacement for $\B_j$ which is a CW-complex of finite type. 
Composing with these replacements with the maps to $\B$, and gluing along the base points, we get a continuous map
\[ \iota_0 \colon K(\Z,2) \vee K(\Z,3) \to \B. \]
Now, using, for example, Theorem A.6 in \cite{Hatcher}, we note that the product cell structure on $K(\Z,2) \times K(\Z,3)$ coincides with the product topology since both have countably many cells. We have an inclusion of sub-complexes
\[ K(\Z,2) \vee K(\Z,3)  \subset K(\Z,2) \times K(\Z,3),\]
where here the wedge is included as the subspace of pairs where one coordinate is the base point.
Any cell in $K(\Z,2) \times K(\Z,3)$ which is not in the wedge has degree at least $6$. Since $\pi_n\B =0$ for $n>3$, $\iota_0$ extends to a map
\[\iota : K(\Z,2) \times K(\Z,3) \to \B. \]
Since the inclusions 
\[K(\Z,j) \to K(\Z,2) \vee K(\Z,3)  \to K(\Z,2) \times K(\Z,3)\]
induce an isomorphism on $\pi_j$, $\iota$ is a weak homotopy equivalence.
\end{proof}

\section{The Fundamental Examples}
\label{sec:example}

A continuous family of translation invariant injective MPS can now be defined rigorously as a map from a topological space $X$ into our space $\B$. 
\begin{defn}
A  system of translation invariant injective matrix product states parametrized by a space $X$ is a continuous function $\psi \colon X \to  \B$. We say that two systems $\psi$ and $\psi'$ are in the same phase if they are homotopic as maps from $X$ to $\B$.
\end{defn}

In this section, we go over the two fundamental examples of 
parametrized families of translation invariant injective MPS.  The first example is a system parametrized by $S^2$ and the second is a system parametrized by $S^3$. We call these the fundamental examples because they are given by continuous functions
\[\psi_2 \colon S^2 \to \B \quad \text{and} \quad \psi_3 \colon S^3 \to \B\]
with the property that the homotopy class of $\psi_2$ and $\psi_3$ are generators for $\pi_2\B$ and $\pi_3\B$.

\subsection{The Generator of \texorpdfstring{$\pi_2$}{pi2}}

We start by showing the following useful lemma.
\begin{lem}
The map $\B(1) \to \B$ induces an isomorphism on $\pi_2$. 
\end{lem}
\begin{proof}
This follows from the commutative diagram 
\[\xymatrix{
 \pi_2\cE(1) \ar[r]^-0 \ar[d] & \pi_2\cB(1) \ar[r]^-\cong \ar[d]& \pi_1\cF(1)  \ar[d]^-=  \ar[r] &  0 \ar[d] \\
0\ar[r] & \pi_2\cB \ar[r]^-\cong & \pi_1\cF \ar[r]  & 0. \\
}\]
Extracted from the long exact sequence in homotopy. We explain the diagram further.
We have used the fact that $\E(1)$ is a $B\Unitary(1)$ and so has vanishing $\pi_1$. By Proposition \ref{prop:BUchi}, the map $\cE(1) \to \B(1)$ models the composite $B\Unitary(1) \to B\PU(1) \to B(\Unitary(1) \times \PU(1) )$ and $B\PU(1) $ is contractible, so this map is trivial on homotopy groups. 
\end{proof}

It follows from the previous lemma that, to construct a parametrized system $S^2 \to \B$ whose homotopy class generates $\pi_2\B$, it suffices to construct a parametrized system $S^2 \to \B(1)$ with this property. But $\B(1) \cong \cI(1)/\Unitary(1)$. An element of $\cI(1)$ is a right-normalized injective tensor $K$ of bond dimension $1$. We thus have a vector $K=(K^1, K^2, \ldots) \in \C^\infty$ such that 
\[\sum_{i\geq 1} K^i(K^i)^* = 1 ,\]
i.e., $K$ is a unit vector, or a point of $\mathbb{S}^\infty \subset \C^\infty$. The $\Unitary(1)$ action is given by $\mu.K =(\mu K^1, \mu K^2, \ldots)$, so is the standard action of $\Unitary(1)$ on $\mathbb{S}^\infty $. This gives an explicit identification of $\B(1)$ as $\C P^\infty$. Using that $S^2 \cong \C P^1$ we let 
\[\psi_2 \colon S^2 \to \B(1)\]
be the inclusion $\C P^1 \subset \C P^\infty$ given by
\[\psi_2([K^1 : K^2]) = [K^1 : K^2: 0 : \cdots] .\] 
It is well-known that $\psi_2$ induces an isomorphism on $\pi_2$.

Note that the image of this map lies in $\B(2, 1, 1)$, which can be mapped into the pure state space $\sP(2)$. For any tensor $K=(K^1, K^2, 0,  \ldots) \in \I(1)$ representing the state $[K]=[K^1 : K^2: 0 : \cdots] \in \B(2, 1, 1)$, it's easy to see that $\mathbb{E}_{\1}$ has unique positive invertible eigenvector $T=1$, and that
\[\omega_K(C_1\otimes \cdots \otimes C_n) =  \prod_{i=1}^n \bra{\Omega_K}C_i\ket{\Omega_K} \]
where
\[
	\ket{\Omega_K} = K^1 \ket{1} + K^2 \ket{2}.
\]
So, $\omega_K$ is the product state corresponding to the pure state $\ket{\Omega_K}$ at each site.

\subsection{The Generator of \texorpdfstring{$\pi_3$}{pi3}}
Let us consider the example of the Chern number pump introduced in \cite{qpump} and further examined in \cite{ChartingGroundStates}. The parameter space is taken to be the 3-sphere $X = S^3$, with elements written as $(\w, w_4)$, where $\w = (w_1, w_2, w_3)$ satisfies $\w^2 + w_4^2 = 1$. After defining a map $\psi_3: S^3 \rightarrow \B$, we will show that it generates $\pi_3$. 

To define our map $\psi_3: S^3 \rightarrow \B$, we cover $S^3$ with two open sets:
\begin{align*}
U_N &= \qty{(\w, w_4) \in S^3: w_4 > - \frac{1}{2}} \\
U_S &= \qty{(\w, w_4) \in S^3 : w_4 < \frac{1}{2}}.
\end{align*}
We will define continuous functions $A_N: U_N \rightarrow \E(d = 4, D = 2)$ and $A_S:U_S \rightarrow \E(d = 4, D = 1)$ such that the projections $U_N \rightarrow \E(4, 2) \rightarrow \B$ and $U_S \rightarrow \E(4, 1) \rightarrow \B$ agree on the overlap $U_N \cap U_S$. Thus, this will define a continuous map $\psi_3: S^3 \rightarrow \B$. 

Physically, we think of this example as describing a model with two qubits on each site of the one-dimensional lattice $\bbZ$. The Hilbert space describing a single lattice site is therefore $\bbC^2 \otimes \bbC^2$. It is therefore convenient to index the $d = 4$ physical dimensions as pairs $ij$ where $i$ and $j$ index the two standard basis vectors on the first and second $\bbC^2$ factor of the tensor product, respectively. We will write $i,j \in \qty{\uparrow, \downarrow}$, where
\[
\ket{\uparrow} = \mqty(1 \\ 0 ) \qqtext{and} \ket{\downarrow} = \mqty(0\\1).
\]

To define $A_N$ and $A_S$, we first define unitary matrices for $\theta, \phi \in \bbR$:
\[
X(\theta, \phi) = \mqty(\cos \frac{\theta}{2} & -e^{-i\phi}\sin \frac{\theta}{2} \\ e^{i \phi} \sin \frac{\theta}{2} & \cos \frac{\theta}{2})
\]
and we define $\Lambda^{N} : U_N \rightarrow M_2(\bbC)$ and $\Lambda^S:U_S \rightarrow M_2(\bbC)$ by
\[
\Lambda^N(\w, w_4) = \left\{ \begin{array}{cl} \mqty(0 & - \sqrt{\frac{1}{2} - \frac{\norm{\w}}{\sqrt{3}}} \\ \sqrt{\frac{1}{2} + \frac{\norm{\w}}{\sqrt{3}}} & 0)  &\tn{if $w_4 \geq \frac{1}{2}$} \\ \mqty(0&0\\1&0) &\tn{if $-\frac{1}{2} < w_4 \leq \frac{1}{2}$}  \end{array} \right.
\]
and
\[
\Lambda^S(\w, w_4) = \left\{ \begin{array}{cl} \mqty(0 &  \sqrt{\frac{1}{2} + \frac{\norm{\w}}{\sqrt{3}}} \\ -\sqrt{\frac{1}{2} - \frac{\norm{\w}}{\sqrt{3}}} & 0)  &\tn{if $w_4 \leq -\frac{1}{2}$} \\ \mqty(0&1\\0&0) &\tn{if $-\frac{1}{2} \leq w_4 < \frac{1}{2}$}  \end{array} \right.
\]
Now we define $A_N^{ij}:U_N \rightarrow M_2(\bbC)$ and $A_S^{ij}:U_S \rightarrow M_1(\bbC)$ by
\begin{align*}
A^{ij}_N(\w, w_4) &= \ketbra{i}{j} X(\theta, \phi)\Lambda^N(\w, w_4)X(\theta, \phi)^T\\
A^{ij}_S(\w, w_4) &= \mel{i}{X(\theta, \phi)\Lambda^S(\w, w_4)X(\theta, \phi)^T}{j}
\end{align*}
where $\theta, \phi \in \bbR$ are chosen so that $\w = \norm{\w}(\sin \theta \cos \phi, \sin \theta \sin \phi, \cos \theta)$. If one expands out the products $X(\theta, \phi)\Lambda^{N/S}(\w, w_4)X(\theta, \phi)^T$, then one sees that each $A^{ij}_{N/S}$ is a well-defined and continuous function of $(\w, w_4)$, in particular at the poles.

We must show that $A_N$ and $A_S$ are valued in $\E(4,2)$ and $\E(4,1)$, respectively. Since $A^{ij}_S(\w, w_4)$ obviously forms an injective MPS tensor, we need only verify the right normalization condition. 
Using a bar to denote complex conjugation, we compute:
\begin{align*}
\sum_{i,j} A^{ij}_S(A^{ij}_S)^* &= \sum_{i,j} \bra{i}X\Lambda^SX^T\ketbra{j} \overline{X} (\Lambda^S)^* X^*\ket{i}\\
&= \sum_i \bra{i}X\Lambda^S (\Lambda^S)^* X^* \ket{i}\\
&= \tr (\Lambda^S (\Lambda^S)^*)\\
&= 1,
\end{align*}
which holds for all $(\w, w_4) \in U_S$. Note also that for $w_4 \in [-\frac{1}{2}, \frac{1}{2})$, we can compute 
\begin{alignat*}{2}
A^{\uparrow \uparrow}_S(\w, w_4) &= -\frac{1}{2}e^{-i\theta} \sin \theta &\qquad A^{\uparrow \downarrow}_S(\w, w_4) &= \frac{1+ \cos \theta}{2}\\
A^{\downarrow \uparrow}_S(\w, w_4) &= \frac{\cos \theta - 1}{2} &\qquad A^{\downarrow \downarrow}_S(\w, w_4) &= \frac{1}{2}e^{i\phi}\sin \theta
\end{alignat*}

To see that $A_N(\w, w_4) \in \E(4, 2)$, write 
\begin{equation}\label{eq:Aij_north_decomposition}
A_N^{ij}(\w, w_4) = \overline{X(\theta, \phi)}B^{ij}(\theta, \phi, w_4)X(\theta, \phi)^T.
\end{equation}
where
\[
B^{ij}(\theta, \phi, w_4) \defeq X(\theta, \phi)^T\ketbra{i}{j} X(\theta, \phi)\Lambda^N(\w, w_4)
\]
For $w_4 \geq \frac{1}{2}$, the matrices $B^{ij}(\theta, \phi)$
span $M_2(\bbC)$ since $\Lambda^N(\w, w_4)$ is invertible in this case; these matrices therefore form an injective MPS. Furthermore, we observe right normalization:
\begin{align*}
\sum_{i,j} X^T \ketbra{i}{j}X\Lambda^N(\Lambda^N)^* X^* \ketbra{j}{i} \overline{X} &= \tr\qty[\Lambda^N(\Lambda^N)^*] \sum_i X^T \ketbra{i}\overline{X}  \\
&= \tr\qty[\Lambda^N(\Lambda^N)^*] \1\\
&= \1.
\end{align*}
For $w_4 \in (-\frac{1}{2}, \frac{1}{2}]$, we have that $B^{ij}(\theta, \phi, w_4)$ is independent of $w_4$ and we compute:
\begin{equation}\label{eq:Bij}
\begin{aligned}
B^{\uparrow \uparrow}(\theta, \phi) &= \mqty(-\frac{1}{2}e^{-i\phi}\sin \theta & 0 \\  e^{-2i\phi}\sin^2 \frac{\theta}{2} & 0) &\qquad B^{\uparrow \downarrow}(\theta, \phi) &= \mqty(\frac{1}{2}(1 + \cos \theta) & 0 \\ -\frac{1}{2}e^{-i\phi} \sin \theta & 0) \\ 
B^{\downarrow \uparrow}(\theta, \phi) &= \mqty(\frac{1}{2}(\cos \theta - 1) & 0 \\  -\frac{1}{2}e^{-i\phi}\sin \theta & 0) &\qquad B^{\downarrow \downarrow}(\theta, \phi) &= \mqty(\frac{1}{2}e^{i\phi}\sin \theta & 0 \\ \frac{1}{2}(1 + \cos \theta)&0)
\end{aligned}
\end{equation}
We see that $B$ is of the form $B = \mqty(K&0\\M&0)$ with $K^{ij} = A^{ij}_S$, which is a right normalized injective MPS tensor. This proves that $A_N(\w, w_4) \in \E(4,2)$ for all values of $(\w, w_4)$, and that the projections $U_N \rightarrow \E(4,2) \rightarrow \B$ and $U_S \rightarrow \E(4,1) \rightarrow \B$ agree on the overlap $U_N \cap U_S$.

It remains to show that the resulting map $\psi_3: S^3 \rightarrow \B$ is a generator of $\pi_3\B \cong \bbZ$. For this purpose it is convenient to consider the basepoint of $\E$ to be $A_0 \defeq A_S(\bm{0},-1)$ with  $p(A_S(\bm{0},-1)) = \psi_3(\bm{0}, -1)$ as the basepoint of $\B$. Let $\tilde \cF = p^{-1}(\psi_3(\bm{0}, -1))$ be the fiber over this basepoint. Since $p$ is a quasifibration and $\E$ is contractible, we have isomorphisms
\[
\pi_3(\B, \psi_3(\bm{0},-1)) \xrightarrow{p_*^{-1}} \pi_3(\E, \tilde \cF, A_0) \xrightarrow{\partial} \pi_2(\tilde \cF, A_0),
\]
where $\partial$ is the connecting homomorphism of the long exact sequence for the based pair $(\E, \tilde \cF, A_0)$. Therefore, to show that the homotopy class $[\psi_3]$ is a generator, it suffices to show that $\partial p_*^{-1}[\psi_3]$ is a generator of $\pi_2(\tilde \cF, A_0)$. 

Let us first compute $p_*^{-1}[\psi_3]$. Recall that $\pi_3(\E, \tilde \cF, A_0)$ is the set of homotopy classes of maps $(D^3, S^2, (0,0,1)) \rightarrow (\E, \tilde \F, A_0)$, where 
\[
D^3 = \qty{\bv \in \bbR^3: \norm{\bv} \leq 1}.
\]
We regard $S^3$ as $D^3$ with its $S^2$ boundary identified to a point. Explicitly, we have a map $f: (D^3, S^2) \rightarrow (S^3, (\bm{0}, -1))$ defined by
\[
f(\bv) = \qty(2\sqrt{1 - \norm{\bv}^2}  \cdot \bv, 1 - 2\norm{\bv}^2) =\vcentcolon (\w(\bv), w_4(\bv)).
\]
Then $\pi_3(\B, \psi_3(\bm{0}, -1))$ may be regarded as the set of homotopy classes of maps of pairs $(D^3, S^2) \rightarrow (\B, \psi_3(\bm{0}, -1))$. In particular, $[\psi_3]$ is regarded as the homotopy class $[\psi_3 \circ f]$. Then $p_*^{-1}[\psi_3 \circ f]$ is computed as a lift of $\psi_3 \circ f$ to a map $A: (D^3, S^2, (0,0,1)) \rightarrow (\E, \tilde \cF, A_0)$. We compute this lift below.

For $\norm{\bv} < \frac{\sqrt{3}}{2}$, equivalently for $w_4(\bv) > -\frac{1}{2}$, we define 
\begin{equation}\label{eq:Aij_north}
A^{ij}(\bv) = A^{ij}_N(\w(\bv), w_4(\bv)).
\end{equation}
For $\norm{\bv} > \frac{1}{2}$, equivalently for $w_4(\bv) < \frac{1}{2}$, we define
\begin{equation}\label{eq:Aij_south_extension}
A^{ij}(\bv) = \overline{X(\theta, \phi)}\mqty(A^{ij}_S(\w(\bv), w_4(\bv))  & 0 \\ M^{ij}(\theta, \phi, \norm{\bv}) & 0 )X(\theta, \phi)^T
\end{equation}
where $\theta, \phi \in \bbR$ are chosen so that $\bv = \norm{\bv}(\sin \theta \cos \phi, \sin \theta \sin \phi, \cos \theta)$ and $M^{ij}(\theta, \phi, \norm{\bv})$ is defined by the lower left corner of \eqref{eq:Bij} for $w_4(\bv) \in [-\frac{1}{2}, \frac{1}{2})$ and for $w_4(\bv) \leq -\frac{1}{2}$ by
\begin{equation*}
\begin{aligned}
M^{\uparrow \uparrow}(\theta, \phi, \norm{\bv}) &= e^{-2i\phi}\qty(\frac{1 - \cos \theta}{2})\qty(\sqrt{\frac{1}{2} + \frac{\norm{\w(\bv)}}{\sqrt{3}}} - \sqrt{\frac{1}{2} - \frac{\norm{\w(\bv)}}{\sqrt{3}}})\\
M^{\downarrow \downarrow}(\theta, \phi, \norm{\bv}) &= \qty(\frac{1 + \cos \theta}{2})\qty(\sqrt{\frac{1}{2} + \frac{\norm{\w(\bv)}}{\sqrt{3}}} - \sqrt{\frac{1}{2} - \frac{\norm{\w(\bv)}}{\sqrt{3}}})\\
M^{\uparrow \downarrow}(\theta, \phi, \norm{\bv}) &= M^{\downarrow \uparrow}(\theta, \phi, \norm{\bv}) = -\frac{1}{2}e^{-i\phi}\sin \theta.
\end{aligned}
\end{equation*}
If one multiplies out the matrices in \eqref{eq:Aij_south_extension}, then one sees that \eqref{eq:Aij_south_extension} is indeed a well-defined continuous function of $\bv$ for $\norm{\bv} > \frac{1}{2}$. By \eqref{eq:Aij_north_decomposition} and \eqref{eq:Bij}, we see that \eqref{eq:Aij_north} and \eqref{eq:Aij_south_extension} agree on the overlap $w_4(\bv) \in (-\frac{1}{2}, \frac{1}{2})$.  Thus, $A^{ij}:(D^3, S^2, (0,0,1)) \rightarrow (\E, \tilde \cF, A_0)$ is well-defined and continuous. It is clearly a lift of $\psi_3 \circ f$. Thus,
\[
p_*^{-1}[\psi_3 \circ f] = [A].
\]

The connecting homomorphism $\partial:\pi_3(\E, \tilde \cF, A_0) \rightarrow \pi_2(\tilde \cF, A_0)$ is defined by restricting a map $(D^3, S^2, (0,0,1)) \rightarrow (\E, \tilde \cF, A_0)$ to the $S^2$ boundary. Thus, $\partial p_*^{-1}[\psi_3 \circ f] = [A|_{S^2}]$. Using the homeomorphism $\tilde \cF \cong \cF(1)$ of Proposition \ref{prop:fiber_homeomorphism} and the homotopy equivalence $\cF(1) \cong \Unitary(1) \times \bbP V_1(\bbC^\infty) = \Unitary(1) \times \bbC \bbP^\infty$ of Proposition \ref{prop:Stiefel_weak_homotopy_equiv}, we get an isomorphism
\[
\pi_2(\tilde \cF) \cong \pi_2(\Unitary(1) \times \bbC \bbP^\infty) \cong \pi_2(\bbC \bbP^\infty),
\]
where $\Unitary(1)$ is based at $1$ and $\bbC \bbP^\infty$ is based at the projectivization of the first standard basis vector. From the definitions in Propositions \ref{prop:fiber_homeomorphism} and \ref{prop:Stiefel_weak_homotopy_equiv}, we see that $[A|_{S^2}]$ maps to the homotopy class of the map
\[
S^2 \rightarrow \bbC\bbP^\infty, \quad \bv \mapsto \mqty[\cos \frac{\theta}{2} \\ e^{-i\phi} \sin \frac{\theta}{2}\\ 0 \\ \vdots],
\]
i.e., the projectivization of the first column of $\overline{X(\theta, \phi)}$. This is indeed a generator of $\pi_2(\bbC \bbP^\infty)$, physically corresponding to the ground state of (the complex conjugate of) Berry's Hamiltonian $H(\theta, \phi) = -\overline{X(\theta, \phi)}\sigma^z X(\theta, \phi)^T$. This proves that $\psi_3$ is a generator of $\pi_3(\cB)$.

\appendix

\section{Equivalence of Subspace and Colimit Topologies}

\begin{proposition}\label{prop:subspace_colimit_equal}
Let $X$ be a topological space with an increasing sequence of $T_1$ subspaces 
\[
X_1 \subset X_2 \subset \cdots \subset X_n \subset \cdots
\]
such that $X = \bigcup_{n=1}^\infty X_n$ and the topology on $X$ is the final topology induced by the inclusions $X_n \hookrightarrow X$. If $A \subset X$ and $A$ is compactly generated with the subspace topology obtained from $X$, then the subspace topology on $A$ coincides with the final topology induced by the inclusions $A \cap X_n \hookrightarrow A$, where each $A \cap X_n$ is topologized as a subspace of $X$.
\end{proposition}

\begin{proof}
Let $\sT_f$ and $\sT_s$ be the final and subspace topology on $A$, respectively. The identity map $\id:(A, \sT_f) \rightarrow (A, \sT_s)$ is continuous since the composition $A \cap X_n \hookrightarrow A \xrightarrow{\id} A \hookrightarrow X$ is the restriction of the inclusion map $X_n \hookrightarrow X$. Therefore $\sT_s \subset \sT_f$.

Suppose $C \subset A$ is closed in the final topology. Since $\sT_s$ is compactly generated, it suffices to show that $C$ is compactly closed relative to $\sT_s$. That is, for any compact Hausdorff space $K$ and continuous function $f:K \rightarrow (A, \sT_s)$, we need to prove $f^{-1}(C)$ is closed in $K$. Since $X$ is a sequential colimit of $T_1$ spaces, we know $f(K) \subset X_n$ for some $n \in \bbN$. The restriction $\overline f:K \rightarrow A \cap X_n$ is continuous, and with $\iota_n:A \cap X_n \rightarrow (A, \sT_f)$ as the inclusion, we have 
\[
f^{-1}(C) =  \overline{f}^{-1}(\iota_n^{-1}(C)),
\]
which is closed since $\iota_n$ is continuous and $C$ is closed in the final topology. Thus, $\sT_f \subset \sT_s$.
\end{proof}

\subsection*{Acknowledgements}

The authors would like to thank Michael Hopkins, Alexei Kitaev and Bruno Nachtergaele for helpful conversations. This material is based upon work supported by the National Science Foundation under Grant No.  DMS 2055501, DMS 2143811 and DMS 2303063. 

\subsection*{Data Availability Statement}

No datasets were generated or analysed during the current study.

\subsection*{Conflicts of Interest}

The authors have no relevant financial or non-financial interests to disclose.

\bibliographystyle{amsalpha}
\bibliography{KZ3-bib.bib}
\end{document}